\documentclass[12pt,letterpaper,onecolumn]{article}

\usepackage{setspace}

\usepackage{natbib}
\usepackage{setspace}
\usepackage{arxiv}
\usepackage[dvipsnames]{xcolor}
\usepackage[utf8]{inputenc} 
\usepackage[T1]{fontenc}    
\usepackage{hyperref}       
\usepackage{url}            
\usepackage{booktabs}       
\usepackage{amsfonts}       
\usepackage{nicefrac}       
\usepackage{microtype}      
\usepackage{lipsum}
\usepackage{algorithm,algpseudocode}
\usepackage{graphicx}
\graphicspath{ {./images/} }
\usepackage{mathtools}
\usepackage{amssymb}
\usepackage{tikz, pgfplots}
\usetikzlibrary{positioning,fit,calc}

\usepackage{mathrsfs}
\usepackage{amsmath,mleftright,mathtools}
\allowdisplaybreaks[1]

\usetikzlibrary{arrows.meta}

\newcommand{\indep}{\perp\!\!\!\!\perp}
\newif\ifshowrevisions
\showrevisionstrue

\usepackage{float}
\usepackage{amsthm}
\newtheorem{prop}{Proposition}

\newtheorem{assumption}{Assumption}

\allowdisplaybreaks

\title{A Bayesian machine learning approach to direct and spillover effects in cluster-randomized trials with noncompliance}

\author{
 Saurabh Bhandari \\
  Department of Public Health Sciences\\
  University of Chicago\\
  Chicago, IL 60637 \\
  \texttt{sbhandari52@uchicago.edu} \\
   \And
  Fan Li \\
   Department of Biostatistics \\
   Yale University \\
   New Haven, CT 06510 \\
   \texttt{fan.f.li@yale.edu} \\
   \And
   Michael J. Daniels \\
   Department of Statistics\\
   University of Florida\\
   Gainesville, FL 32611 \\
   \texttt{daniels@ufl.edu} \\
}

\newif\ifkeepmathsize
\newcommand{\smallmathsize}{\ifkeepmathsize\else\edef\smallmathbls{\the\baselineskip}\small
  \setlength{\baselineskip}{\smallmathbls}\fi}
\newcommand{\footnotemathsize}{\keepmathsizetrue\edef\smallmathbls{\the\baselineskip}\fontsize{8}{9.5}\selectfont
  \setlength{\baselineskip}{\smallmathbls}}
\AddToHook{env/equation/before}{\begingroup\smallmathsize}
\AddToHook{env/equation/after}{\endgroup}
\AddToHook{env/equation*/before}{\begingroup\smallmathsize}
\AddToHook{env/equation*/after}{\endgroup}
\AddToHook{env/align/before}{\begingroup\smallmathsize}
\AddToHook{env/align/after}{\endgroup}
\AddToHook{env/align*/before}{\begingroup\smallmathsize}
\AddToHook{env/align*/after}{\endgroup}

\begin{document}
\doublespacing
\maketitle

\begin{singlespace}
\begin{abstract}
We develop a Bayesian machine learning approach for estimating direct and spillover effects in cluster-randomized trials with treatment noncompliance. To flexibly model outcomes and treatment receipt while accounting for within-cluster dependence, we combine Bayesian additive regression trees with cluster-level random effects. We augment the outcome regression with a peer-treatment clever covariate, defined as the predicted mean proportion of treated peers from a preliminary receipt model, to mitigate regularization-induced distortion of the effect decomposition. Within a Bayesian g-computation framework, we estimate cluster-average intent-to-treat effects and cluster-average local average treatment effects and examine sensitivity to the cross-world restrictions underlying their decomposition. Simulation studies show reduced bias under model misspecification, while the correctly specified parametric models show expected frequentist properties. We find that the clever covariate improves the decomposition of the total effect into direct and spillover components under moderate or strong spillover, with similar gains in smaller clusters. Its benefits vary with spillover strength, and the $95\%$ credible intervals remain conservative. An application to the Primary School Deworming Project in Kenya suggests reduced infection risk, with a larger estimated contribution from individual treatment receipt than from peer receipt. Overall, our approach provides a flexible Bayesian estimation of direct and spillover effect decomposition in the presence of treatment noncompliance.
\end{abstract}
\end{singlespace}


\newpage

\section{Introduction}\label{Sec1}

Cluster-randomized trials (CRTs) often evaluate interventions whose benefits extend beyond the individuals who receive treatment. In infectious disease studies, for example, treating one individual may reduce transmission to others in the same cluster. When treatment receipt differs from randomized assignment, the intervention can affect outcomes through changes in both an individual's own treatment and their peers' treatment. Separating direct and spillover effects is more difficult because treatment receipt is not randomized and the relevant contrasts combine receipt under different assignment regimes.

Our motivating example is the Primary School Deworming Project (PSDP), which randomized schools in Kenya to receive mass deworming \citep{miguel2004worms}. Some students in treatment-assigned schools did not take the medication but could still benefit if the deworming their schoolmates received reduced parasite transmission. The scientific question is therefore how much of the intervention's effect on infection operates through a student's own treatment receipt and how much operates through treatment of peers. We address this question through direct and spillover decompositions (Section \ref{Sec4}) of the cluster-average ITT effect and the cluster-level local average treatment effect (LATE).

The existing work on interference and noncompliance provides complementary foundations for this analysis. \citet{hudgens2008toward} develop causal contrasts under interference in two-stage randomized designs, with subsequent work extending identification and estimation under partial interference \citep{tchetgen2012causal}. For noncompliance, principal-stratification approaches distinguish assignment effects from effects within latent compliance groups \citep{angrist1995identification}. Methods accommodating both features include approaches for peer encouragement and two-stage randomized designs \citep{kang2016peer,imai2021causal}. In a CRT, however, all individuals within a cluster share the same assignment. This design provides limited information for separating individual and peer receipt effects without additional identifying restrictions \citep{keele2022introduction}.

In a related setting, \citet{park2022efficient} develop efficient influence functions and doubly robust estimators for network effects under partial interference. Allowing noncompliance with a monotone outcome restriction, they provide assumption-lean ITT analyses and sharp bounds for network effects within compliance strata. We instead study a g-formula approach under additional restrictions on cross-world outcome means and treatment receipt, without imposing outcome monotonicity. Similarly, \citet{cheng2025identification} develop point identification under extended principal ignorability and estimate principal causal effects through efficient influence functions. Their network assignment effect can include an assignment effect beyond treatment receipt, which our network exclusion restriction removes. Consequently, the two decompositions generally differ.

We develop a Bayesian g-computation approach \citep{josefsson2021bayesian,bhandari2025bayesian,bhandari2025causal} using mixed-effects Bayesian additive regression trees (BART-RE; Section \ref{Sec5}), combining flexible covariate adjustment with random intercepts for residual within-cluster dependence \citep{chipman2010bart}. Our primary modeling contribution is to augment the outcome regression with a peer-treatment \textit{clever covariate}, defined here as the predicted mean proportion of treated peers from a preliminary treatment-receipt model. Using this low-dimensional summary of measured determinants of peer uptake, we adapt propensity-based regularization strategies \citep{hahn2020bayesian,linero2025mediation} to peer exposure in CRTs. 

The remainder of the article is structured as follows. Section~\ref{Sec2} introduces the motivating example. Sections~\ref{Sec3} and~\ref{Sec4} describe the experimental design and causal estimands in detail. In Section~\ref{Sec5}, we present our proposed modeling approach using BART-RE. Section~\ref{Sec6} and Section~\ref{Sec7}  outline the non-parametric identification and estimation steps, respectively, for the causal parameters. Sections~\ref{Sec8} and~\ref{Sec9} report results from simulation studies and an empirical application, respectively. We conclude with a discussion in Section~\ref{Sec10}.

\section{Motivating example}\label{Sec2}

The PSDP trial evaluated school-based treatment for intestinal helminth infections in Busia, Kenya \citep{miguel2004worms}. The study enrolled $75$ primary schools with more than $30{,}000$ pupils aged $6$--$18$. Schools were stratified by administrative subzone and prior participation in nongovernmental assistance programs, then randomized to three groups of $25$ schools. Deworming began in $1998$, $1999$, and $2001$ for Groups~1, 2, and 3, respectively.

We analyze the initial phase, comparing assignment to deworming in $1998$ in Group~1 ($Z_j=1$) with Group~2 ($Z_j=0$). Group~3 is excluded because the relevant follow-up infection outcomes are unavailable. Individual receipt, $D_{ji}$, indicates whether a student received deworming medication in $1998$. The outcome, $Y_{ji}$, indicates any moderate-to-heavy helminth infection at the $1999$ follow-up. Of $4{,}018$ students in the survey, $2{,}421$ had observed follow-up infection status. Excluding $37$ students without receipt records and one with missing school assignment and cluster covariates yields $N=2{,}383$ students in $J=49$ schools ($25$ treatment and $24$ control). The analyzed cluster sizes range from $23$ to $80$ students.

Treatment receipt exhibits one-sided noncompliance. Among $871$ students in treatment schools, $157$ ($18\%$) did not receive medication, while none of the $1{,}512$ control students received treatment. Within-school spillover is plausible because deworming may reduce parasite contamination in the shared environment, protecting untreated schoolmates. 


Table~\ref{Table::PSDP_baseline} presents the individual-level and school-level baseline covariates used for adjustment. The unadjusted follow-up infection rates were $27.2\%$ in treatment schools and $52.2\%$ in control schools. This difference motivates assessment of the overall intervention effect, but does not determine how much operates through individual versus peer treatment receipt.

\begin{table}[ht]
\caption{Summary statistics of baseline covariates by treatment assignment in the PSDP analysis sample ($N = 2{,}383$, $J = 49$ schools).}
\label{Table::PSDP_baseline}
\centering
\begin{tabular}{@{}llcc@{}}
\hline
& & Treatment  & Control  \\
& & $n = 871$ & $n = 1{,}512$ \\
\hline
\multicolumn{4}{@{}l}{\textit{Individual-level covariates}} \\
& Male (\%)  & 53.0 & 48.5 \\
& Age (years) & 12.7 & 12.7 \\
& Has household latrine (\%) & 83.7 & 81.7 \\
& Owns livestock (\%) & 65.4 & 69.9 \\
& Weight-for-age $z$-score  & $-1.50$ & $-1.46$ \\
& Blood in stool (\%)  & 26.4 & 23.3 \\
& Often sick (\%)  & 8.6 & 10.5 \\
& Clean hands (\%)  & 60.0 & 65.3 \\
& Malaria (\%)  & 36.3 & 39.7 \\[4pt]
\multicolumn{4}{@{}l}{\textit{School-level covariates}} \\
& Zone baseline infection rate  & 0.369 & 0.365 \\
& School exam score 1996  & $-0.197$ & 0.217 \\
& Distance to Lake Victoria (km)  & 10.0 & 9.9 \\
& Latrines per pupil  & 0.007 & 0.006 \\[4pt]
\multicolumn{4}{@{}l}{\textit{Treatment receipt and outcome}} \\
& Received deworming, $D=1$ (\%)  & 82.0 & 0.0 \\
& Any infection at follow-up, $Y=1$ (\%)  & 27.2 & 52.2 \\
\hline
\end{tabular}
\end{table}

\section{Notation and setup}\label{Sec3}

Consider a CRT with $J$ clusters, where cluster $j$ contains $n_j$ individuals and $N=\sum_{j=1}^{J}n_j$ is the total sample size. Let $Z_j\in\{0,1\}$ denote randomized cluster assignment, $D_{ji}\in\{0,1\}$ individual treatment receipt, and $Y_{ji}\in\{0,1\}$ the binary outcome. The baseline covariate vector $\mathbf{X}_{ji}\in\mathbb{R}^p$ includes individual- and cluster-level characteristics, with $\mathbf{X}_j=\{\mathbf{X}_{ji}\}_{i=1}^{n_j}$ denoting all baseline information in cluster $j$.

In the PSDP, $J=49$ schools were randomized to deworming or control in 1998. Here, $D_{ji}$ indicates whether student $i$ in school $j$ received deworming medication, and $Y_{ji}$ ($Y_{ji} = 1$ if infected, $0$ otherwise) indicates moderate-to-heavy intestinal helminth infection at the 1999 follow-up. The $p=13$ baseline covariates are listed in Appendix~\ref{AppendixC3}.

We use the potential outcomes framework to define our causal effects. 
For each $z_j \in \{0,1\}$, let $D_{ji}(z_j)$ denote the potential treatment receipt status of subject $i$ in cluster $j$ if cluster $j$ were assigned to treatment $z_j$. The vector of potential treatment receipts in cluster $ j $ under assignment $ z_j $ is denoted by $ \textbf{D}_{j}(z_j)$.
The corresponding vector excluding subject $ i $ is denoted by $ \textbf{D}_{j(-i)}(z_j)$.
Let $ \textbf{z}_{(-j)}$
denote the vector of treatment assignments for all clusters except cluster $ j $, and let $ \textbf{D}_{(-j)}(\textbf{z}_{(-j)})$
denote the vector of potential treatment receipt statuses for all clusters excluding cluster $ j $.

Let $ Y_{ji}\big(z_j, \textbf{z}_{(-j)}, d_{ji}, \textbf{d}_{j(-i)}, \textbf{d}_{(-j)} \big) $ denote the potential outcome of subject $ i $ in cluster $ j $ if cluster $ j $ were assigned treatment $ z_j $, all other clusters were assigned $ \textbf{z}_{(-j)} $, subject $ i $'s treatment receipt were $ d_{ji} $, their within-cluster peers had treatment receipt statuses given by $ \textbf{d}_{j(-i)} $, and subjects in all other clusters had treatment receipt statuses given by $ \textbf{d}_{(-j)}$. Similarly, let $Y_{ji}\big(z_j, \textbf{z}_{(-j)}, D_{ji}(z_j), \textbf{D}_{j(-i)}(z_j), \textbf{D}_{(-j)}(z_{(-j)}) \big)$ denote the potential outcome of subject $ i $ in cluster $ j $ if cluster $ j $ were assigned treatment $ z_j $, all other clusters were assigned $ \textbf{z}_{(-j)} $, subject $ i $'s treatment receipt were $ D_{ji}(z_j) $, their within-cluster peers had treatment receipt statuses given by $ \textbf{D}_{j(-i)}(z_j) $, and subjects in all other clusters had treatment receipt statuses given by $ \textbf{D}_{(-j)}(z_{(-j)}) $. The observed outcome equals the potential outcome evaluated at the treatment receipt vector realized under the actual cluster assignment $Z_j$. 


Following \citet{park2023assumption}, we index potential receipt by the common cluster assignment, $D_{ji}(z_j)$. Individual assignments cannot vary within a cluster under this design. The resulting potential receipt pair defines the four conventional principal strata in Section~\ref{Sec32}. A directed acyclic graph (DAG) of individual receipt and within-cluster spillover pathways appears in Appendix~\ref{AppendixG}, Figure~\ref{Fig:DAG}. In summary, we define the collection of all random variables in cluster $j$ as $\textbf{W}_j = \Big\{\textbf{X}_j, \textbf{D}_j(0), \textbf{D}_j(1), \textbf{Y}_{j}\big(0,  \textbf{z}_{(-j)},  \textbf{d}_{j}, \textbf{d}_{(-j)}(z_{(-j)}) \big), \textbf{Y}_{j}\big(1, \textbf{z}_{(-j)}, \textbf{d}_{j}, \textbf{d}_{(-j)}(z_{(-j)}) \big)   \Big\}$
for all $\textbf{d}_{j} \in \{0,1\}^{n_j}$ and $\textbf{d}_{(-j)}(z_{(-j)}) \in \{0,1\}^{N - n_j}$. 


\subsection{Principal stratification}\label{Sec32}
Principal strata are defined by the potential treatment receipt indicators $(D_{ji}(0),D_{ji}(1))$, which classify individuals as compliers $(0,1)$, defiers $(1,0)$, always-takers $(1,1)$, or never-takers $(0,0)$. In the PSDP, treatment was unavailable in control schools, implying $D_{ji}(0)=0$ for all students and ruling out always-takers and defiers. Thus, the population consists of compliers, who receive deworming medication when assigned to treatment, and never-takers, who do not receive it regardless of assignment.

\subsection{Structural assumptions}\label{Sec33}

In this section, we formalize the assumptions that underpin our CRT framework and the associated treatment noncompliance.
\begin{assumption}[Partial interference]\label{Assump:1}
For any $\textbf{z}_{(-j)} \neq \textbf{z}'_{(-j)}$ and $\textbf{d}_{(-j)} \neq \textbf{d}'_{(-j)}$:
$Y_{ji}\big(z_j, \textbf{z}_{(-j)}, d_{ji}, \textbf{d}_{j(-i)}, \textbf{d}_{(-j)} \big) 
        = Y_{ji}\big(z_j, \textbf{z}'_{(-j)}, d_{ji}, \textbf{d}_{j(-i)}, \textbf{d}'_{(-j)} \big)
        = Y_{ji}\big(z_j, d_{ji}, \textbf{d}_{j(-i)} \big).$
\end{assumption}
Assumption~\ref{Assump:1} implies that the potential outcome for individual $i$ in cluster $j$ does not depend on the treatment assignment or treatment receipt of individuals in other clusters ($j' \neq j$). This allows for \textit{within-cluster interference}---i.e., the outcome of an individual may be influenced by peers in the same cluster---but rules out \textit{across-cluster interference}. This assumption is commonly adopted in analyses of CRTs and is plausible when clusters are sufficiently geographically separated that spillovers between clusters are negligible.

\begin{assumption}[Cluster randomization]\label{Assump:2}
The treatment assignment for each cluster is an independent realization from a Bernoulli distribution with $P\big(Z_j = 1\big) = \pi_z \in \big(0,1\big)$. 
\end{assumption}
We consider a CRT with $J$ clusters, of which $m < J$ are randomly assigned to the treatment group ($Z_j = 1$), and the remaining $J - m$ clusters are assigned to the control group ($Z_j = 0$). 
This assumption holds by design.

\begin{assumption}[Nonzero causal effect of $Z$ on $D$]\label{Assump:3}
$ \sum_{j=1}^{J}\sum_{i=1}^{n_j} \left\{ D_{ji}(1) - D_{ji}(0) \right\} \neq 0.$
\end{assumption} 
Assumption~\ref{Assump:3} ensures that $Z_j$ has a nonzero causal effect on $D_{ji}$. This is typically satisfied in CRTs and holds in the PSDP example.

\begin{assumption}[Network exclusion restriction]\label{Assump:4}
Given $d_{ji}$ and $\textbf{d}_{j(-i)}$:
$Y_{ji}\big(z_j = 1, d_{ji}, \textbf{d}_{j(-i)}\big) 
    = Y_{ji}\big(z_j = 0, d_{ji}, \textbf{d}_{j(-i)}\big) 
    = Y_{ji}\big(d_{ji}, \textbf{d}_{j(-i)}\big)$.
\end{assumption}
Assumption~\ref{Assump:4} is the network exclusion restriction used by \citet{park2023assumption}. It rules out effects of cluster assignment that persist when the entire treatment-receipt vector is held fixed. Assignment can affect an individual's outcome through both their own receipt and their peers' receipt. This restriction gives the spillover component a treatment-receipt interpretation, but does not by itself identify that component. It would be violated if assignment changed health education, sanitation, or other outcome determinants independently of medication receipt.

\begin{assumption}[Monotonicity of receipt]\label{Assump:5} For every individual $ji$:
$D_{ji}(0) \leq D_{ji}(1).$
\end{assumption}
Assumption~\ref{Assump:5} rules out the presence of \textit{defiers}---that is, we assume that no individual would receive the treatment under control assignment but not under treatment assignment. Formally, a defier is defined as an individual for whom $D_{ji}(1) < D_{ji}(0)$. Under the monotonicity of potential treatment receipt, such individuals do not exist in the population. This restricts the population to consist only of compliers, always-takers, and never-takers.

As in \citet{park2023assumption}, we assume partial interference, cluster randomization, treatment relevance, network exclusion, and treatment receipt monotonicity. These assumptions identify the overall ITT effect but do not generally lead to point identification of network effects within compliance strata. With additional boundedness and monotonicity restrictions on outcomes, \citet{park2023assumption} derive sharp bounds for these network effects. We instead impose conditional expected homogeneity and cross-world independence of treatment receipt across different individuals to obtain point identification of the direct and spillover components through the g-formula, without requiring outcome monotonicity (Section~\ref{Sec6}). Our identification strategy also differs from that of \citet{cheng2025identification}, who use extended principal ignorability to link conditional outcome means across compliance strata and obtain point identification without network exclusion. We retain network exclusion so that our decomposition describes effects operating through individual and peer treatment receipt.

\section{Causal estimands}\label{Sec4}
In the presence of noncompliance, the primary causal estimands of interest are the intent-to-treat (ITT) effect and the local average treatment effect (LATE) \citep{angrist1995identification}. The overall ITT effect is defined as the expected difference in potential outcomes under treatment versus control across all individuals. In contrast, the LATE quantifies the average causal effect among compliers. 

Under interference, compliers are of interest because their own receipt changes with cluster assignment. Always-takers and never-takers can nevertheless experience assignment effects through changes in their peers' receipt. Monotonicity removes defiers and identifies compliance probabilities from treatment-receipt differences, but further restrictions are needed to identify outcome contrasts within the latent compliance strata.

Within a CRT framework, we define the ITT and LATE effects while explicitly accounting for the cluster-level structure, given by:
\begin{equation}\label{Eq:clusterITT}
    \begin{aligned}
        \text{ITT}^{C} &= \mathbf{E}\Big[\frac{1}{n_j} \sum_{i=1}^{n_j} \big\{Y_{ji}\big(D_{ji}(1), \textbf{D}_{j(-i)}(1)\big) - Y_{ji}\big(D_{ji}(0), \textbf{D}_{j(-i)}(0)\big)\big\} \Big], \text{ and }
        \\
        \text{LATE}^{C} &= \mathbf{E}\Big[\frac{1}{n_j} \sum_{i=1}^{n_j} \big\{Y_{ji}\big(D_{ji}(1), \textbf{D}_{j(-i)}(1)\big) - Y_{ji}\big(D_{ji}(0), \textbf{D}_{j(-i)}(0)\big)\big\}\big| D_{ji}(0) = 0, D_{ji}(1) = 1 \Big],
    \end{aligned}
\end{equation}
where the expectation is taken with respect to the joint distribution of the cluster-level random variables $\mathbf{W}_j$ and the cluster size $n_j$, denoted 
$\mathcal{P}_{W,n} = \mathcal{P}_{W|n} \times \mathcal{P}_{n}$,
as implied by the super-population framework (Appendix~\ref{AppendixC1}).

The cluster-ITT effect averages the consequences of assigning a cluster to treatment rather than control, including changes in both individual and peer receipt. The cluster-LATE targets the corresponding assignment contrast among compliers. Under interference, this complier contrast includes both direct and spillover pathways and need not equal the conventional no-interference instrumental-variable estimand.

In the PSDP trial, the outcome is coded as $Y_{ji} = 1$ for infection, so a \emph{negative} $\text{ITT}^{C}$ indicates that assignment to the deworming intervention reduces the average probability of helminth infection at the school level, which is the beneficial direction. The $\text{ITT}^{C}$ captures this effect across all students regardless of whether they personally took the medication, reflecting the overall impact of the school-level deworming policy. The $\text{LATE}^{C}$ restricts attention to compliers, that is, the students who take the medication when offered and do not take it otherwise, and a negative $\text{LATE}^{C}$ indicates that, among this subpopulation, being assigned to deworming reduces infection probability.


In the presence of clustered interference, as in a CRT, the $\text{ITT}^{C}$ and $\text{LATE}^{C}$ effects can be further decomposed into direct and spillover effects. Specifically, the decomposition of the ITT$^{C}$ is:
\begin{itemize}
    \item[(1)] \textbf{Cluster-ITT direct effect} (ITT$^C_{DE}$):
    \begin{equation}\label{Eq:HDE_ITT}
        \begin{aligned}
            \text{ITT}^C_{DE} &= \mathbf{E}\Big[\frac{1}{n_j} \sum_{i=1}^{n_j} \big\{Y_{ji}\big(D_{ji}(1), \textbf{D}_{j(-i)}(0)\big) - Y_{ji}\big(D_{ji}(0), \textbf{D}_{j(-i)}(0)\big)\big\}\Big],
        \end{aligned}
    \end{equation}
    \end{itemize}
    which captures the effect of changing the individual’s treatment receipt from the level it would take under control assignment to that under treatment assignment, while holding the treatment receipt status of all other individuals in the cluster fixed at the level they would take under control assignment.
\begin{itemize}
    \item[(2)] \textbf{Cluster-ITT spillover effect} (ITT$^C_{SE}$):
    \begin{equation}\label{Eq:HSE_ITT}
        \begin{aligned}
            \text{ITT}^C_{SE} &= \mathbf{E}\Big[\frac{1}{n_j} \sum_{i=1}^{n_j} \big\{Y_{ji}\big(D_{ji}(1), \textbf{D}_{j(-i)}(1)\big) - Y_{ji}\big(D_{ji}(1), \textbf{D}_{j(-i)}(0)\big)\big\}\Big],
        \end{aligned}
    \end{equation}
    \end{itemize}
    which measures the effect of changing the treatment receipt status of all other individuals in the cluster from the levels they would take under control to those under treatment assignment, while holding the individual’s own treatment receipt fixed at the level it would take under treatment assignment.

Similarly, for LATE$^{C}$, we define:
\begin{itemize}
    \item[(1)] \textbf{Cluster-LATE direct effect} (LATE$^C_{DE}$):
    \begin{equation}\label{Eq:LATE_DE}
        \begin{aligned}
            \text{LATE}^C_{DE} &= \mathbf{E}\Big[\frac{1}{n_j} \sum_{i=1}^{n_j} \big\{Y_{ji}\big(D_{ji}(1), \textbf{D}_{j(-i)}(0)\big) - Y_{ji}\big(D_{ji}(0), \textbf{D}_{j(-i)}(0)\big)\big\} \big| D_{ji}(0) = 0, D_{ji}(1) = 1 \Big].
        \end{aligned}
    \end{equation}
    \item[(2)] \textbf{Cluster-LATE spillover effect} (LATE$^C_{SE}$):
    \begin{equation}\label{Eq:LATE_SE}
        \begin{aligned}
             \text{LATE}^C_{SE} &= \mathbf{E}\Big[\frac{1}{n_j} \sum_{i=1}^{n_j} \big\{Y_{ji}\big(D_{ji}(1), \textbf{D}_{j(-i)}(1)\big) - Y_{ji}\big(D_{ji}(1), \textbf{D}_{j(-i)}(0)\big)\big\} \big| D_{ji}(0) = 0, D_{ji}(1) = 1 \Big].
        \end{aligned}
    \end{equation}
\end{itemize}
In the PSDP, the direct effects (ITT$^C_{DE}$ and LATE$^C_{DE}$) capture the change in a student's infection probability attributable to their own deworming receipt, holding fixed the treatment receipt of their schoolmates at the level that would occur under the control assignment. The spillover effects (ITT$^C_{SE}$ and LATE$^C_{SE}$) capture the additional change in a student's infection probability that arises from the increase in schoolmates' deworming receipt when the school moves from the control to the treatment assignment, while holding the student's own treatment receipt fixed at the level under treatment assignment. 


The ITT components also have a principal-stratum interpretation. Let $S_{ji}$ denote compliance type, and select a cluster from the super-population followed by an individual $I$ uniformly within that cluster. Define $\omega_s=\Pr(S_{jI}=s)=\mathbf{E}[n_j^{-1}\sum_{i=1}^{n_j}\mathbf{1}(S_{ji}=s)]$. If $\mathrm{PCE}^{C}_{s,DE}$ and $\mathrm{PCE}^{C}_{s,SE}$ denote the corresponding direct and spillover contrasts conditional on $S_{jI}=s$, then $\mathrm{ITT}^{C}_{DE}=\sum_s\omega_s\mathrm{PCE}^{C}_{s,DE}$ and $\mathrm{ITT}^{C}_{SE}=\sum_s\omega_s\mathrm{PCE}^{C}_{s,SE}$. The complier components are $\mathrm{LATE}^{C}_{DE}$ and $\mathrm{LATE}^{C}_{SE}$. Under monotonicity and network exclusion, the direct components for always-takers and never-takers are zero because their own receipt is unchanged, whereas their spillover components can be nonzero.

\section{Model specification}\label{Sec5}

We extend probit Bayesian additive regression trees (BART) \citep{chipman2010bart} by incorporating cluster-specific random intercepts to account for residual within-cluster dependence in CRTs. The resulting mixed-effects BART (BART-RE) model combines a flexible sum-of-trees regression function, $f(\mathbf{X})=\sum_{t=1}^{T}g(\mathbf{X};\mathcal{T}_t,\mathcal{M}_t)$, with cluster-level random effects. Here, $\mathcal{T}_t$ and $\mathcal{M}_t$ denote the tree structure and terminal node parameters, respectively. Two key advantages of BART-RE in our context are: i) the tree ensemble captures nonlinear relationships and interactions without prespecifying their functional forms, and ii) the random intercepts account for residual between-cluster heterogeneity. We use the default BART priors to regularize individual trees and retain posterior uncertainty quantification.

\subsection{Outcome model with a clever covariate}\label{Sec51}

In specifying the outcome model, it is neither meaningful nor practical to include the entire vector of within-cluster treatment receipt indicators $\mathbf{D}_{j(-i)}$ as covariates, since these are exchangeable across peers and do not provide subject-specific explanatory power beyond their aggregate influence. Instead, we summarize peer treatment receipt by the proportion of treated peers,
$\bar{D}_{j(-i)} = (n_j - 1)^{-1}\sum_{\ell \neq i} D_{j\ell}$.
This mapping treats peers symmetrically and assumes that their receipt influences the conditional outcome mean through its aggregate proportion. We additionally include the peer-treatment clever covariate, $\hat{\pi}^{\bar{D}}_{j(-i)}$, which summarizes predicted peer uptake and is distinct from the observed proportion $\bar{D}_{j(-i)}$. We specify the outcome model as
\begin{equation}\label{Eq::OutcomeModel}
\begin{aligned}
\Pr\big(Y_{ji} = 1 \mid Z_j, D_{ji}, \bar{D}_{j(-i)}, \mathbf{X}_{ji}, n_j, b_j^{Y}\big)
&= \Phi\Big(f^{Y}\big(Z_j, D_{ji}, \bar{D}_{j(-i)}, \mathbf{X}_{ji}, n_j,
\hat{\pi}^{\bar{D}}_{j(-i)}\big) + b_j^{Y}\Big), \quad
f^{Y}\big(\cdot\big) &\sim \text{BART},
\end{aligned}
\end{equation}
where $\Phi(\cdot)$ is the standard normal cumulative distribution function and
$b_j^{Y} \overset{\text{i.i.d.}}{\sim} \mathcal{N}(0,\tau_{b^Y}^2)$,
with a suitable prior on $\tau_{b^Y}^2$, such as an inverse-Gamma or half-Cauchy prior. The BART ensemble flexibly captures relationships among treatment assignment, individual and peer receipt, and baseline covariates, while the random intercept accounts for residual within-cluster dependence. We obtain posterior draws using Albert--Chib data augmentation \citep{albert1993bayesian}, introducing
$U^Y_{ji}=f^Y(\cdot)+b_j^Y+\varepsilon^Y_{ji}$,
where $\varepsilon^Y_{ji}\sim\mathcal{N}(0,1)$ and
$Y_{ji}=\mathbf{1}(U^Y_{ji}>0)$.

\paragraph{Clever covariate for spillover deconfounding.}
Individual and peer receipt may be correlated because they share measured cluster characteristics and residual cluster-level influences. Under a regularizing tree prior as in BART, the outcome variation associated with these predictors can be allocated unevenly, distorting the decomposition even when their sum is estimated reasonably well \citep{hahn2020bayesian,linero2025mediation}. We therefore augment the outcome regression with the fitted mean peer-treatment proportion. This peer-treatment clever covariate is an auxiliary predictor for regularization, rather than the inverse-probability fluctuation covariate of targeted maximum likelihood estimation \citep{van2011targeted}. Its use provides no double-robustness guarantee. We define it as 
\begin{equation}\label{Eq::CleverCovariate}
\hat{\pi}^{\bar{D}}_{j(-i)}
=
\frac{1}{n_j-1}
\sum_{\ell\neq i}
\Phi\!\left(
\hat{\xi}_0+\hat{\xi}_Z Z_j+
\hat{\boldsymbol{\xi}}_X^{\top}\mathbf{X}_{j\ell}
\right).
\end{equation}
The estimates $(\hat\xi_0,\hat\xi_Z,\hat{\boldsymbol\xi}_X)$ come from a parametric probit regression fitted to the observed treatment receipt before the BART models. We hold the resulting predictor fixed during outcome-model MCMC. It summarizes measured determinants of peer uptake, and does not include the latent receipt random effect or account for uncertainty in the pre-fit. This adapts propensity-augmented Bayesian tree modeling to a peer exposure \citep{linero2025mediation,hahn2020bayesian}. Its potential benefit is empirical and depends on the quality of the receipt summary, the regularization prior, and the support of the receipt configurations. The current implementation evaluates this predictor at the peer assignment $z_*$ in g-computation under $\Theta_{(z,z_*)}$, $\hat{\pi}^{\bar{D}}_{j(-i)}(z_*)
=
\frac{1}{n_j-1}
\sum_{\ell\neq i}
\Phi\!\left(
\hat{\xi}_0+\hat{\xi}_Z z_*+
\hat{\boldsymbol{\xi}}_X^{\top}\mathbf{X}_{j\ell}
\right),$ and holds it fixed across Monte Carlo replicates. This convention requires an explicit justification for mixed assignment pairs, because the observed predictor is a deterministic function of the factual assignment and baseline covariates. It does not create support for otherwise unobserved receipt configurations.

\subsection{Treatment receipt indicator model}\label{Sec52} 

For the binary treatment receipt data, we adopt a probit-BART model with a cluster-level random effect $b_j^{D}$. Specifically, we let
\begin{equation}\label{Eq::TrtReceiptModel}
\begin{aligned}
\Pr(D_{ji}=1 \mid Z_j, \mathbf{X}_{ji}, n_j, b_j^{D})
&= \Phi\Big(f^{D}\big(Z_j, \mathbf{X}_{ji}, n_j\big) + b_j^{D}\Big), \qquad
f^{D}\big(\cdot\big) &\sim \text{BART},
\end{aligned}
\end{equation}
where $b_j^{D} \overset{\text{i.i.d.}}{\sim} \mathcal{N}(0,\tau_{b^D}^2)$, with $\tau_{b^D}^2$ given a suitable prior. Here, $\Phi(\cdot)$ denotes the standard normal cumulative distribution function. The details of the posterior sampling algorithm for the observed-data BART models are provided in Appendix~\ref{AppendixB}.

Note that the specification of the treatment receipt model in \eqref{Eq::TrtReceiptModel} implicitly assumes conditional independence of peers’ treatment receipt statuses within a cluster, given the cluster-level random effect $b_j^{D}$. The shared cluster-level influences are explicitly captured through the cluster-level random effect $b_j^{D}$ (and, when available, cluster-level covariates). Formally, we assume that, 
conditional on $(Z_j,\mathbf{X}_j,b_j^{D})$, the individual receipt indicators are independent. 


\subsection{Baseline covariates}\label{Sec53} 

Let $\mathbf{X}_j=\{\mathbf{X}_{ji}\}_{i=1}^{n_j}$ denote the $n_j\times p$ matrix of baseline covariates in cluster $j$. Suppose $(\mathbf{X},n)$ denotes a synthetic cluster with size $n$ and corresponding $n\times p$ covariate matrix $\mathbf{X}$. To model the joint super-population distribution $G$ of cluster sizes and baseline covariates (Appendix~\ref{AppendixC1}), we use the Bayesian Bootstrap (BB) \citep{rubin1981bayesian}.  
An advantage of using BB in this setup is that it preserves the observed within-cluster covariate structure without imposing parametric assumptions.

\section{Non-parametric identification of causal parameters}\label{Sec6}

In this section, we introduce additional assumptions, beyond Assumptions~\ref{Assump:1}--\ref{Assump:5}, that are required for the nonparametric identification of the causal estimands from the observed data. We also outline the corresponding identification strategies. Specifically, we require the following conditions to hold:

\begin{assumption}[Conditional expected homogeneity]\label{Assump:6} 
For $z \neq z_{*}$:
\begin{equation*}
    \begin{aligned}
        & \mathbf{E}\Big[\sum_{i=1}^{n} \big\{ Y_{.i}\big(D_{.i}(z) = d_{.i}, \textbf{D}_{.(-i)}(z_{*}) = \textbf{d}_{.(-i)}\big)\big\}\big| Z=z, D_{.i}(z) = d_{.i}, \textbf{D}_{.(-i)}(z_{*}) = \textbf{d}_{.(-i)},   \textbf{X}_{.i} = \textbf{x}_{.i}, n\Big]
        \\&
        = \mathbf{E}\Big[\sum_{i=1}^{n} \big\{ Y_{.i}\big(D_{.i}(z) = d_{.i}, \textbf{D}_{.(-i)}(z) = \textbf{d}_{.(-i)}\big)\big\}\big| Z=z, D_{.i}(z) = d_{.i}, \textbf{D}_{.(-i)}(z) = \textbf{d}_{.(-i)},   \textbf{X}_{.i} = \textbf{x}_{.i}, n\Big] 
    \end{aligned}
\end{equation*}

\end{assumption}
Assumption~\ref{Assump:6} requires that, conditional on baseline covariates and cluster size, the expected outcome at a fixed individual and peer receipt configuration is the same whether peers would have that configuration under assignment $z$ or $z_*$. This restricts unmeasured confounding by ruling out differences in mean outcomes between the groups defined by potential peer receipt under the two assignments. In the PSDP, unmeasured factors affecting both uptake and infection could violate this assumption if they produce differences in expected infection between these groups after covariate adjustment.

A related conditional homogeneity assumption appears in \citet{ohnishi2026bayesian}. They assume that the cluster-average effect through one mediator remains unchanged after conditioning on potential values of another mediator, given baseline covariates and cluster size. Our assumption equates outcome means at the same receipt configuration across groups defined by potential peer receipt under different assignments. Both impose restrictions on conditional means, but neither assumption implies the other. We examine sensitivity to violations of our assumption in Appendix~\ref{AppendixD}.

\begin{assumption}[Cross-world independence of the within-cluster peer treatment receipt status]\label{Assump:7}
For all $j$, $i \neq i^{'}$ and $z \neq z_{*}$:
$D_{ji}\big(z\big) \indep D_{ji^{'}}\big(z_{*}\big)\big| \textbf{X}_{ji} = \textbf{x}_{ji}, n_j$.
\end{assumption}
Assumption~\ref{Assump:7} states that, conditional on $\mathbf{X}_{ji}$ and $n_j$, different individuals' treatment receipt statuses are independent when defined under different assignments. It rules out residual dependence due to shared, unmeasured determinants of compliance. The treatment receipt may still be dependent across individuals under the same assignment and across assignments for the same individual. Note that this assumption restricts treatment receipt, whereas Assumption~\ref{Assump:6} restricts expected outcomes. A similar independence assumption for different individuals' potential mediators is used by \citet{ohnishi2026bayesian}.

In the PSDP trial, one-sided noncompliance implies $D_{ji}(0)=0$ for every student. Assumption~\ref{Assump:7} therefore holds automatically because one student's receipt under treatment assignment is independent of another student's constant receipt under control assignment. The negligible changes in the sensitivity analysis in Appendix~\ref{AppendixE} should therefore be interpreted in light of one-sided noncompliance. 


\paragraph{Remark.} It is important to distinguish between two types of counterfactuals: those where $z_j = z_{j,*}$ and those where $z_j \neq z_{j,*}$. 
In the former case ($z_j = z_{j,*}$), identification relies solely on the randomization of the treatment assignment and does not require Assumptions~\ref{Assump:6} and~\ref{Assump:7}. 
In contrast, counterfactuals of the form $Y_{ji}\big(D_{ji}(z_j), \mathbf{D}_{j(-i)}(z_{j,*})\big)$ with $z_j \neq z_{j,*}$ require both the randomization of the treatment assignment and Assumptions~\ref{Assump:6} and~\ref{Assump:7} for identification.

\begin{assumption}[Consistency]\label{Assump:8} For any $z_{j} \in \{0,1\}$: \\
i) $D_{ji}(z_{j}) = D_{ji}$  given  $Z_j = z_{j}$, ii) $\mathbf{D}_{j(-i)}(z_{j}) = \mathbf{D}_{j(-i)}$  given  $Z_j = z_{j}$, and iii) $Y_{ji}\big(d_{ji}, \textbf{d}_{j(-i)}\big) = Y_{ji}$ given $D_{ji} =d_{ji}$ and $\textbf{D}_{j(-i)} = \textbf{d}_{j(-i)}$.
\end{assumption}

\begin{assumption}[Positivity]\label{Assump:9}
We assume $0 < \Pr\big(Z_j = 1 \mid \mathbf{X}_j, n_j\big) < 1,$
where $\mathbf{X}_j$ contains the baseline covariates of all individuals in cluster $j$, cluster-level covariates, and any randomization-stratum indicators.
\end{assumption}
Assumption~\ref{Assump:9} requires that clusters, regardless of baseline covariate profile and size in the target population, have a positive probability of assignment to either arm. This condition holds by design in a CRT when both arms have positive allocation probabilities within every randomization stratum.

For notational simplicity, let $n$ denote an arbitrary realization drawn from the unknown distribution $\mathcal{P}_{n}$, that is, a single cluster size sampled from the conceptual infinite population of clusters in the super-population framework. Similarly, let $Y_{.i}$, $D_{.i}$, and $\mathbf{D}_{.(-i)}$ denote the outcome, individual treatment receipt status, and the vector of treatment receipt statuses of all other individuals within this arbitrarily chosen cluster of size $n$. The notation $\textbf{b}_{.} = \big\{ b_{.}^{Y}, b_{.}^{D}\big\}$ denotes the vector of random effects for this cluster. Proposition~\ref{Prop1:ITTidenditification} provides non-parametric identification results for the cluster-ITT and cluster-LATE effects.

For the set of treatment assignment pairs
$\big(z,z_{*}\big) \in \{(1,1),\,(0,0),\,(1,0)\}$,
we define the cluster-ITT and cluster-LATE estimands as
\begin{equation}\label{Eq::GcompEstimand}
    \begin{aligned}
        \Theta_{(z,z_{*})}^{ITT}
        &= \int_{\textbf{b}_{.}}\mathbf{E}_{(\mathbf{X},n)\sim G}
        \Bigg[ \frac{1}{n}\sum_{i=1}^{n}
        \int_{d_{.i}} \int_{\mathbf{d}_{.(-i)}}
        \mathbf{E}\Big[
        Y_{.i}\mid Z = z,\; D_{.i} = d_{.i},\;
        \mathbf{D}_{.(-i)} = \mathbf{d}_{.(-i)},\;
        \mathbf{X}_{.i} = \mathbf{x}_{.i},\;
        n,\; b^{Y}_{.}\Big]
        \\& \qquad \qquad
        dF_{\mathbf{D}_{.(-i)}\mid Z = z_{*},
        \mathbf{X}_{.i},n,\; b^{D}_{.}}
        (\mathbf{d}_{.(-i)})\,
        dF_{D_{.i}\mid Z = z,
        \mathbf{X}_{.i},n,\; b^{D}_{.}}
        (d_{.i})\Bigg]\times dF_{\textbf{b}_{.}},
        \text{ and}
        \\
        \Theta_{(z,z_{*})}^{LATE}
        &= \int_{\textbf{b}_{.}}\mathbf{E}_{(\mathbf{X},n)\sim G}
        \Bigg[ \frac{1}{n}\sum_{i=1}^{n}
        \int_{\mathbf{d}_{.(-i)}}
        \mathbf{E}\Big[
        Y_{.i}\mid Z = z,\; D_{.i} = z,\;
        \mathbf{D}_{.(-i)} = \mathbf{d}_{.(-i)},\;
        \mathbf{X}_{.i} = \mathbf{x}_{.i},\;
        n,\; b^{Y}_{.}\Big]
        \\& \qquad \qquad
        dF_{\mathbf{D}_{.(-i)}\mid Z = z_{*},
        \mathbf{X}_{.i},n,\; b^{D}_{.}}
        (\mathbf{d}_{.(-i)})\Bigg]
        \times dF_{\textbf{b}_{.}}, 
    \end{aligned}
\end{equation}
respectively, where the outer expectation is taken with respect to the
super-population distribution $G$ of cluster-level baseline
covariates and cluster size.

\begin{prop}\label{Prop1:ITTidenditification}
Under Assumptions~\ref{Assump:1}--\ref{Assump:9}, we can
identify the cluster-ITT effects and the cluster-LATE effects from the observed data
distributions, respectively, as:
\begin{equation}\label{Eq::Prop1ITT}
    \begin{aligned}
        \text{ITT}^C_{DE}
        = \Theta_{(1,0)}^{ITT} - \Theta_{(0,0)}^{ITT}, \quad
        \text{ITT}^C_{SE}
        = \Theta_{(1,1)}^{ITT} - \Theta_{(1,0)}^{ITT},
        \text{ and} \quad
        \text{ITT}^C
        = \Theta_{(1,1)}^{ITT} - \Theta_{(0,0)}^{ITT}.
    \end{aligned}
\end{equation}

\noindent and
\begin{equation}\label{Eq::Prop2LATE}
    \begin{aligned}
        \text{LATE}^C_{DE}
        = \Theta_{(1,0)}^{LATE}
        - \Theta_{(0,0)}^{LATE},
        \quad
        \text{LATE}^C_{SE}
        = \Theta_{(1,1)}^{LATE}
        - \Theta_{(1,0)}^{LATE},
        \text{ and} \quad
        \text{LATE}^C
        = \Theta_{(1,1)}^{LATE}
        - \Theta_{(0,0)}^{LATE}.
    \end{aligned}
\end{equation}
\end{prop}
\begin{proof}
    See Appendix~\ref{AppendixA}
\end{proof}

\paragraph{Connection to causal mediation analysis.}
Certain components of our direct and spillover effects bear resemblance to the direct and indirect effects in causal mediation analysis, with peers' treatment receipt playing a role similar to that of mediators. A key distinction is the network exclusion restriction, under which cluster assignment $Z_j$ affects $Y_{ji}$ only through treatment receipt, with no direct effect of assignment on the outcome.

\section{G-computation}\label{Sec7}

The G-computation algorithm used to approximate $\Theta_{(z,z_{*})}^{ITT}$, Algorithm \ref{Algo1}, proceeds as follows. As described in Section~\ref{Sec5}, we estimate the joint distribution of $(\mathbf{X},n)$ under the super-population framework using BB. Specifically, given the observed data $\{(\mathbf{X}_j, n_j)\}_{j=1}^J$, the BB places a Dirichlet$(1,\ldots,1)$ prior on the probability masses $(p_1,\ldots,p_J)$ and samples $(p^{(q)}_1,\ldots,p^{(q)}_J) \sim \mathrm{Dirichlet}(1,\ldots,1)$ at MCMC iteration $q$. Each draw then defines a super-population distribution,
$\widehat{G}(\cdot \mid p^{(q)}) \;=\; \sum_{j=1}^J p^{(q)}_j \,\delta_{(\mathbf{X}_{j},n_j)}(\cdot).$

Because the BB places probability mass directly on the observed clusters, each posterior draw $q$ reduces to computing weighted averages over these clusters. This eliminates the need to sample synthetic clusters. Specifically, for each cluster $(\mathbf{X}_j,n_j)$, we use its posterior weight $p_j^{(q)}$ and draw the cluster-level random effects
$b^{D,(q)}_{j} \sim \mathcal{N}(0,\tau^{2,(q)}_{b^{D}})$ and
$b^{Y,(q)}_{j} \sim \mathcal{N}(0,\tau^{2,(q)}_{b^{Y}})$.
Then, for each individual $i=1,\ldots,n_j$, we generate $M$ Monte Carlo draws of the treatment receipt under the specified treatment regimes, denoted by $D_{ji}^{(q,m)} \sim F_{D_{ji}\mid Z_j=z,\mathbf{X}_{ji},n_j,b^{D,(q)}_{j}}, 
\mathbf{D}_{j(-i)}^{(q,m)} \sim F_{\mathbf{D}_{.(-i)}\mid Z_j=z_{*},\mathbf{X}_{ji},n_j,b^{D,(q)}_{j}},
 m = 1, \ldots, M.$
The predicted probability of the outcome for the $m^{th}$ Monte Carlo sample is $\widehat{y}_{ji}^{(q,m)} = \Phi\Big(f^{Y,(q)}\big(z, D_{ji}^{(q,m)}, \bar{D}_{j(-i)}^{(q,m)}, \mathbf{X}_{ji}, n_j, \hat{\pi}^{\bar{D}}_{j(-i)}(z_*)\big) + b^{Y,(q)}_{j}\Big),$ where $\bar{D}_{j(-i)}^{(q,m)} = \frac{1}{n_j -1}\sum_{\ell \neq i}D_{j\ell}^{(q,m)}$. This yields the inner integral approximation 
$\widehat{\Phi}_{ji}^{ITT,(q)} = \frac{1}{M}\sum_{m=1}^M \widehat{y}_{ji}^{(q,m)},$
and the corresponding cluster-level functional
$\widehat{\Phi}_{j}^{ITT,(q)}(\mathbf{X}_j,n_j) = \frac{1}{n_j}\sum_{i=1}^{n_j} \widehat{\Phi}_{ji}^{ITT,(q)}.$  

The final approximation of the ITT estimand in iteration $q$ is therefore $\widehat\Theta_{(z,z_{*})}^{ITT,(q)}
= \sum_{j=1}^J p_j^{(q)} \,\widehat{\Phi}_{j}^{ITT,(q)}(\mathbf{X}_j,n_j),$
where the BB weights $p_j^{(q)}$ directly define the contribution of each observed cluster. The posterior distribution of $\Theta_{(z,z_{*})}^{ITT}$ is then obtained from $\{\widehat{\Theta}_{(z,z_{*})}^{ITT,(q)}\}_{q=1}^Q$.

\paragraph{Remark.} We integrate the cluster random effects out rather than conditioning on the fitted effects of the observed clusters. 
Combined with the Bayesian bootstrap distribution of $(\mathbf{X}_j,n_j)$, this yields super-population causal estimands. Within each posterior draw, we use the same simulated cluster effects across all $(z,z_*)$ combinations so that causal contrasts are evaluated under a common realization of cluster-level heterogeneity.

For clarity, the full approximation of the ITT estimand in iteration $q$ is $\widehat\Theta_{(z,z_{*})}^{ITT,(q)}
\;\approx\; \sum_{j=1}^{J} p_j^{(q)}
\Big[\frac{1}{n_{j}}\sum_{i=1}^{n_{j}} \Big\{ \frac{1}{M}\sum_{m=1}^{M}\widehat{y}_{ji}^{(q,m)}\Big\}  \Big]$ and the details are given in Algorithm~\ref{Algo1}.
The G-computation algorithm for the cluster-LATE estimand (Algorithm~\ref{Algo2}) is provided in Appendix~\ref{AppendixC}.

\begin{algorithm}[htbp]
\caption{G-computation algorithm for $\Theta_{(z,z_{*})}^{ITT,(q)}$ under the super-population framework }\label{Algo1}
\begin{algorithmic}[1]
\State Compute the clever covariate $\hat{\pi}^{\bar{D}}_{j(-i)}(z_*)$ via Equation~\eqref{Eq::CleverCovariate} for $z_* \in \{0,1\}$ and append $\hat{\pi}^{\bar{D}}_{j(-i)}(Z_j)$ as a fixed column of the outcome design matrix.
\State Fit the BART models specified in Section~\ref{Sec5}. Discard the first $B$ MCMC iterations as burn-in, and retain $Q$ posterior draws.
\For{posterior draw $q = 1,\dots,Q$}
    \State Sample weights $(p^{(q)}_1,\ldots,p^{(q)}_J) \sim \mathrm{Dirichlet}(1,\ldots,1)$.
    \State For each cluster $j = 1,\dots,J$, draw the cluster random effects $b^{Y,(q)}_{j} \sim \mathcal{N}\big(0, \tau^{2,(q)}_{b^{Y}}\big)$ and $b^{D,(q)}_{j} \sim \mathcal{N}\big(0, \tau^{2,(q)}_{b^{D}}\big)$, where $\tau^{2,(q)}_{b^{Y}}$ and $\tau^{2,(q)}_{b^{D}}$ are the $q$-th posterior draws of the variance components.
    \For{cluster $j = 1,\dots,J$}
        \For{individual $i = 1,\dots,n_j$}
            \For{Monte Carlo replicate $m = 1,\dots,M$}
                \State Sample $D_{ji}^{(q,m)} \sim F_{D_{ji}\mid Z_j=z,\mathbf{X}_{ji},n_j,b^{D,(q)}_{j}}$.
                \State Sample $\mathbf{D}_{j(-i)}^{(q,m)} \sim F_{\mathbf{D}_{.(-i)}\mid Z_j=z_{*},\mathbf{X}_{ji},n_j,b^{D,(q)}_{j}}$.
                \State Compute $\widehat{y}_{ji}^{(q,m)} = \Phi\Big(f^{Y,(q)}\big(z, D_{ji}^{(q,m)}, \bar{D}_{j(-i)}^{(q,m)}, \mathbf{X}_{ji}, n_j, \hat{\pi}^{\bar{D}}_{j(-i)}(z_*)\big) + b^{Y,(q)}_{j}\Big)$, where  $\bar{D}_{j(-i)}^{(q,m)} = \frac{1}{n_j -1}\sum_{\ell \neq i}D_{j\ell}^{(q,m)}$, and $\hat{\pi}^{\bar{D}}_{j(-i)}(z_*)$ is the clever covariate evaluated at the peer assignment $z_*$.
            \EndFor
            \State $\widehat{\Phi}_{ji}^{ITT,(q)} \gets \frac{1}{M}\sum_{m=1}^M \widehat{y}_{ji}^{(q,m)}$.
        \EndFor
        \State $\widehat{\Phi}_{j}^{ITT,(q)}(\mathbf{X}_j,n_j) \gets \frac{1}{n_j}\sum_{i=1}^{n_j} \widehat{\Phi}_{ji}^{ITT,(q)}$.
    \EndFor
    \State $\widehat{\Theta}_{(z,z_{*})}^{ITT,(q)} \gets \sum_{j=1}^J p_j^{(q)} \,\widehat{\Phi}_{j}^{ITT,(q)}(\mathbf{X}_j,n_j)$.
\EndFor
\State Posterior samples: $\{\widehat{\Theta}_{(z,z_{*})}^{ITT,(q)}\}_{q=1}^Q$.
\end{algorithmic}
\end{algorithm}

\section{Simulation studies}\label{Sec8}

We conduct two simulation studies to evaluate the proposed approach. Section~\ref{Sec81} compares effect estimation using BART-RE and a parametric probit model under PSDP-calibrated data-generating mechanisms (DGMs). Section~\ref{Sec82} examines the contribution of the peer-treatment clever covariate through paired comparisons under varying spillover strengths.

\subsection{Comparison with a parametric probit model}\label{Sec81}
In this study, we compare effect estimation using BART-RE and a parametric probit model under PSDP-calibrated linear and nonlinear DGMs. Specifically, we calibrate the DGMs using posterior estimates from a parametric probit model fitted to the PSDP data and evaluate estimation of direct, spillover, and total effects under correct model specification and nonlinear misspecification at two sample sizes ($N=1{,}500$ with $J=50$ clusters and $N=2{,}250$ with $J=75$ clusters). Details of the simulation design and data generation are provided in Appendix~\ref{AppendixFa}.

Tables~\ref{TableSim1} and~\ref{TableSim2} report the simulation results, showing a clear trade-off between the two approaches. Under the correctly specified DGM (Table~\ref{TableSim1}), the parametric model performs well, with low bias and close-to-nominal coverage across effects. BART-RE shows greater bias for the direct and spillover components, while recovering the total effect with little bias and near-nominal coverage. This difference is expected when the parametric model is correctly specified, as its stronger functional-form assumptions provide efficiency that a more flexible model cannot generally match in moderate samples.

The results also suggest that BART-RE has difficulty separating direct and spillover effects when individual and peer treatment receipt are correlated, even though it recovers the total effect accurately. The undercoverage observed for the spillover effect under the correctly specified DGM likely arises from uncertainty in decomposing the overall effect, rather than from poor estimation of the underlying potential-outcome means. As shown in Table~\ref{TabSuppS1}, coverage for all three potential-outcome means ($\Theta_{(1,1)}$, $\Theta_{(0,0)}$, and $\Theta_{(1,0)}$) remains closer to the nominal $95\%$ level. Increasing the sample size from $N=1{,}500$ to $N=2{,}250$ further improves coverage for the BART-RE spillover effect, from approximately $88\%$--$89\%$ to $92\%$ for both the cluster-ITT and cluster-LATE estimands.

The advantage of BART-RE becomes more apparent under parametric model misspecification (Table~\ref{TableSim2}). When nonlinear covariate effects and treatment-effect heterogeneity are present, the parametric model exhibits substantial bias and poor coverage. For example, coverage for the cluster-ITT spillover effect is only $8.2\%$ and $0.6\%$ for $N=1{,}500$ and $N=2{,}250$, respectively, while coverage for the cluster-LATE direct effect is $0\%$ at both sample sizes. In contrast, BART-RE maintains substantially smaller bias and coverage close to or marginally above the nominal level across both sample sizes and estimands. Thus, while the parametric model can be preferable when its functional form is correctly specified, BART-RE provides greater protection against misspecification while retaining reasonable performance for the total effect. These results support the use of BART-RE when the functional forms governing treatment receipt and outcomes are uncertain, as is typically the case in real applications.

\begin{table}[htbp]
\caption{Simulation results under the correctly specified DGM: bias, RMSE, and $95\%$ credible interval coverage across $500$ replicated datasets.}
\label{TableSim1}
\scalebox{0.8}{%
\begin{tabular}{@{}lcccccccccc@{}}
\hline
Effect type & $N$ & $J$ & Effects & Ground Truth & \multicolumn{3}{c}{Parametric} &\multicolumn{3}{c}{BART-RE} \\
\cline{6-8}\cline{9-11}%
 &  &   &  &   & Bias & Cov. & RMSE & Bias & Cov. & RMSE \\
\hline
      & 1500 & 50 & {Direct}   & $0.14$ & $\phantom{-}0.0051$ & $95.0\%$ & $0.055$ & $-0.084$ & $95.0\%$ & $0.094$ \\
          & & & {Spillover} & $-0.059$ & $-0.013$ & $94.6\%$ & $0.028$ & $\phantom{-}0.091$ & $87.8\%$ & $0.095$ \\
{cluster-ITT}    & & & {Total}    & $\phantom{-}0.085$ & $-0.0074$ & $94.0\%$ & $0.051$ & $\phantom{-}0.0068$ & $94.8\%$ & $0.054$ \\
\cline{2-11}
{effects}    & 2250 & 75 & {Direct}   & $\phantom{-}0.14$ & $\phantom{-}0.0058$ & $94.8\%$ & $0.044$ & $-0.080$ & $96.8\%$ & $0.088$ \\
            & & & {Spillover}  & $-0.059$ & $-0.017$ & $90.2\%$ & $0.028$ & $\phantom{-}0.087$ & $92.0\%$ & $0.091$ \\
            & & & {Total}     & $\phantom{-}0.085$ & $-0.011$ & $94.4\%$ & $0.042$ & $\phantom{-}0.0064$ & $95.8\%$ & $0.045$ \\
\cline{1-11}
     & 1500 & 50 & {Direct}   & $\phantom{-}0.0034$ & $\phantom{-}0.0006$ & $94.2\%$ & $0.063$ & $-0.093$ & $93.0\%$ & $0.10$ \\
             & & & {Spillover}  & $-0.062$ & $-0.0078$ & $96.8\%$ & $0.026$ & $\phantom{-}0.091$ & $89.2\%$ & $0.095$ \\
{cluster-LATE}    & & & {Total}   & $-0.059$ & $-0.0071$ & $93.8\%$ & $0.055$ & $-0.0015$ & $95.4\%$ & $0.056$ \\
\cline{2-11}
 {effects}   & 2250 & 75 & {Direct}   & $\phantom{-}0.0034$ & $\phantom{-}0.0026$ & $93.2\%$ & $0.050$ & $-0.089$ & $93.0\%$ & $0.098$ \\
                   & & & {Spillover}  & $-0.062$ & $-0.012$ & $91.2\%$ & $0.025$ & $\phantom{-}0.087$ & $91.6\%$ & $0.091$ \\
                   & & & {Total}   & $-0.059$ & $-0.0098$ & $94.8\%$ & $0.046$ & $-0.0016$ & $94.6\%$ & $0.046$ \\
\hline
\end{tabular}}
\end{table}

\begin{table}[htbp]
\caption{Simulation results under the misspecified DGM: bias, RMSE, and $95\%$ credible interval coverage across $500$ replicated datasets. The parametric probit model omits the nonlinear terms in the true DGM.}
\label{TableSim2}
\scalebox{0.8}{%
\begin{tabular}{@{}lcccccccccc@{}}
\hline
Effect type & $N$ & $J$ & Effects & Ground Truth & \multicolumn{3}{c}{Parametric} &\multicolumn{3}{c}{BART-RE} \\
\cline{6-8}\cline{9-11}%
 &  &   &  &   & Bias & Cov. & RMSE & Bias & Cov. & RMSE \\
\hline
      & 1500 & 50 & {Direct}   & $\phantom{-}0.083$ & $\phantom{-}0.050$ & $76.4\%$ & $0.060$ & $-0.024$ & $\phantom{0}99.2\%$ & $0.062$ \\
          & & & {Spillover} & $-0.026$ & $-0.077$ & $\phantom{0}8.2\%$ & $0.080$ & $\phantom{-}0.028$ & $\phantom{0}99.8\%$ & $0.055$ \\
{cluster-ITT}    & & & {Total}    & $\phantom{-}0.057$ & $-0.028$ & $91.4\%$ & $0.045$ & $\phantom{-}0.0040$ & $\phantom{0}97.2\%$ & $0.035$ \\
\cline{2-11}
{effects}    & 2250 & 75 & {Direct}   & $\phantom{-}0.083$ & $\phantom{-}0.050$ & $59.6\%$ & $0.057$ & $-0.026$ & $\phantom{0}99.4\%$ & $0.058$ \\
            & & & {Spillover}  & $-0.026$ & $-0.079$ & $\phantom{0}0.6\%$ & $0.081$ & $\phantom{-}0.029$ & $\phantom{0}99.8\%$ & $0.054$ \\
            & & & {Total}     & $\phantom{-}0.057$ & $-0.030$ & $84.3\%$ & $0.043$ & $\phantom{-}0.0025$ & $\phantom{0}95.2\%$ & $0.029$ \\
\cline{1-11}
     & 1500 & 50 & {Direct}   & $-0.019$ & $\phantom{-}0.30$ & $\phantom{0}0.0\%$ & $0.30$ & $\phantom{-}0.017$ & $\phantom{0}99.2\%$ & $0.061$ \\
             & & & {Spillover}  & $-0.026$ & $-0.060$ & $\phantom{0}9.8\%$ & $0.063$ & $\phantom{-}0.023$ & $\phantom{0}99.8\%$ & $0.050$ \\
{cluster-LATE}    & & & {Total}   & $-0.045$ & $\phantom{-}0.24$ & $\phantom{0}0.0\%$ & $0.25$ & $\phantom{-}0.040$ & $\phantom{0}88.2\%$ & $0.056$ \\
\cline{2-11}
 {effects}   & 2250 & 75 & {Direct}   & $-0.019$ & $\phantom{-}0.300$ & $\phantom{0}0.0\%$ & $0.302$ & $\phantom{-}0.0038$ & $\phantom{0}99.2\%$ & $0.053$ \\
                   & & & {Spillover}  & $-0.026$ & $-0.062$ & $\phantom{0}1.2\%$ & $0.064$ & $\phantom{-}0.024$ & $\phantom{0}99.8\%$ & $0.051$ \\
                   & & & {Total}   & $-0.045$ & $\phantom{-}0.238$ & $\phantom{0}0.0\%$ & $0.241$ & $\phantom{-}0.028$ & $\phantom{0}87.7\%$ & $0.043$ \\
\hline
\end{tabular}}
\end{table}

Tables~\ref{TabSuppS3} and~\ref{TabSuppS4} in Appendix~\ref{AppendixF2} show that, under the correctly specified DGM, the parametric model produces narrower intervals for the direct and spillover effects, whereas BART-RE yields wider intervals while maintaining comparable uncertainty for the total effect. Under misspecification, however, the narrow intervals from the parametric model are misleading as they become even narrower despite the substantial bias and severe undercoverage observed for the direct and spillover effects, particularly for the cluster-LATE estimand. BART-RE appropriately expresses greater uncertainty about these decomposed effects, while its intervals for the total effect remain similar in width to those from the parametric model. Taken together, the apparent precision of the parametric model under misspecification is not supported by its inferential performance, whereas BART-RE provides a more reliable assessment of uncertainty for the direct and spillover effects.

\subsection{Contribution of the peer-treatment clever covariate}\label{Sec82}

In this study, we assess the contribution of the clever covariate by fitting otherwise identical BART-RE models with and without $\hat{\pi}^{\bar D}_{j(-i)}$ in simulated datasets. Specifically, we analyze $500$ datasets for each spillover strength using $50$ trees. The two contrasting model fits (\emph{with} versus \emph{without} clever covariate) share treatment-receipt posterior draws and BB weights. Each dataset contains $49$ clusters of $48$ individuals, with $24$ clusters assigned to treatment. Noncompliance is one-sided, with average uptake of $82\%$ in the treated arm. We hold the individual receipt coefficient fixed and vary the peer receipt coefficient over $\gamma_{\bar D}\in\{0,-0.25,-0.5,-0.75,-1\}$. This changes the true spillover effect while leaving the direct effect fixed, allowing us to assess whether the models distinguish these pathways. Appendix~\ref{AppendixFb} describes the DGM, the calculation of the true effects, the MCMC settings, and the procedure for integrating over treatment receipt.

The design also highlights a difficulty relevant to the PSDP trial. Estimating $\Theta_{(1,0)}$ requires predictions for individuals whose own receipt is generated under treatment assignment but whose peers are untreated. Treated clusters have high uptake, so the corresponding receipt configurations are not represented in the simulated data. The decomposition therefore depends on the outcome model's extrapolation.

Figure~\ref{FigSimCC} shows that accurate estimation of the total effect can conceal errors in its decomposition. When we omit the clever covariate from the outcome model, it increasingly assigns the benefit of peer treatment to individual receipt as spillover strengthens. The resulting component biases have opposite signs and largely cancel in their sum. The clever covariate improves this decomposition, reducing the MSE of both components by more than half under moderate or strong spillover, with little change in total-effect estimation. Its benefit is conditional because with weak spillover, the gains are mixed, and without spillover, it increases component error. The evidence therefore supports adding the clever covariate for improving the decomposition when peer treatment actually contributes to the outcome. Table~\ref{TableSim3} in Appendix~\ref{AppendixFb} provides a tabulation of results in Figure~\ref{FigSimCC}.

\begin{figure}[htbp]
\centering
\includegraphics[width=\textwidth]{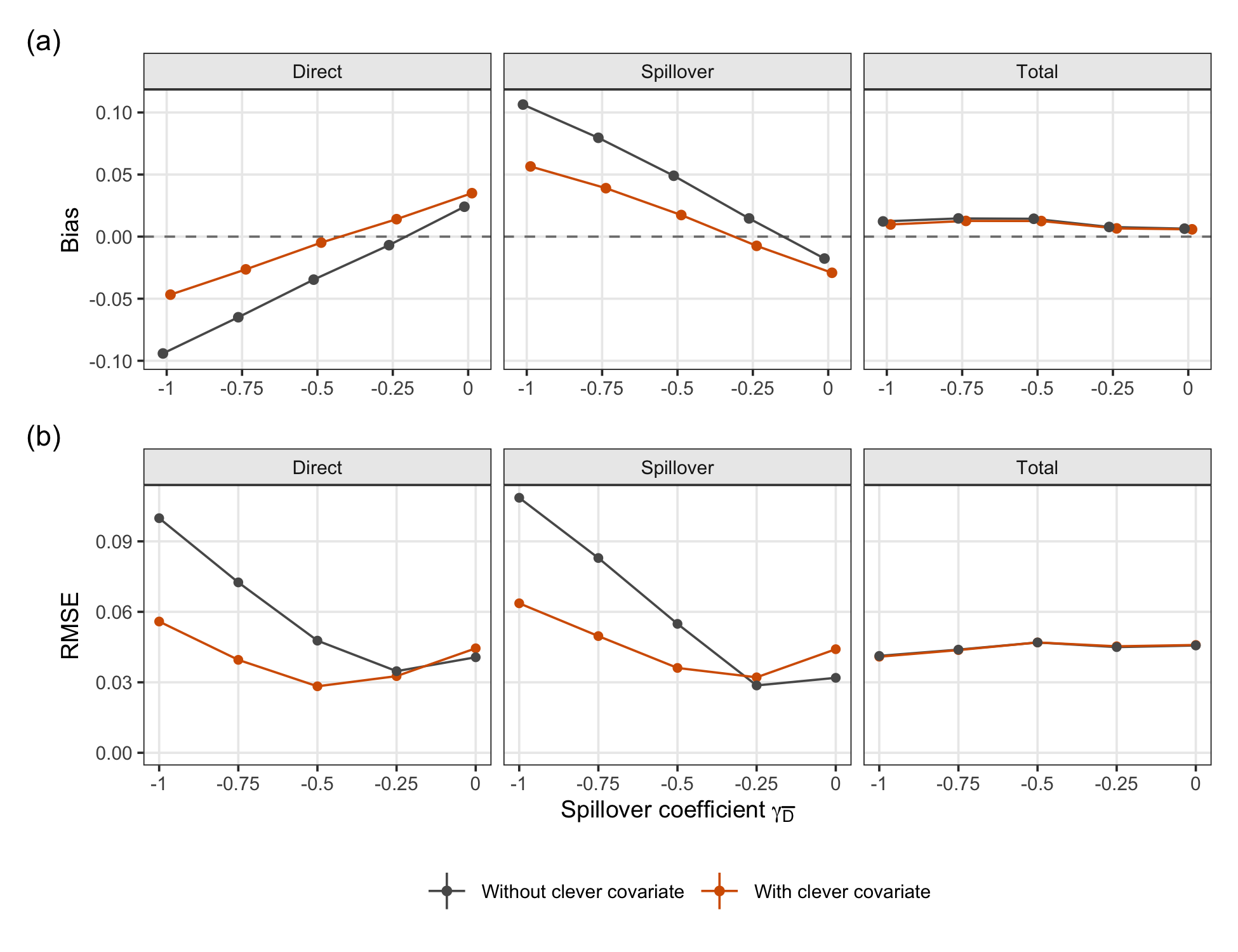}
\caption{Clever-covariate study ($J=49$, $n_j=48$): (a) bias ($\pm2$ Monte Carlo standard errors) and (b) RMSE for the cluster-ITT direct, spillover, and total effects, with spillover coefficient $\gamma_{\bar D}$ based on $500$ replicated datasets.}
\label{FigSimCC}
\end{figure}

The results in Appendix~\ref{AppendixFb} show that including the clever covariate mainly improves estimation of $\Theta_{(1,0)}$. Estimates of $\Theta_{(1,1)}$ and $\Theta_{(0,0)}$ remain similar with and without the clever covariate. Although the true $\Theta_{(1,0)}$ is constant across spillover strengths, its estimates change with the treated-arm mean. Including the clever covariate reduces this dependence, improving the direct and spillover estimates with little change in the total-effect estimate. However, the clever covariate varies little within each assignment arm in this design, so these results provide limited evidence about its performance when cluster covariates produce substantial differences in predicted peer uptake.

We note that better point estimates do not yield precise component inference. Both models produce wide, conservative direct and spillover intervals, reflecting substantial posterior uncertainty about the mixed mean. The covariate does not resolve the missing receipt support. An independent study with smaller clusters reproduces the component improvements under strong spillover (Appendix~\ref{AppendixFc}). Together, the studies support the clever covariate as a useful regularization device for effect decomposition, with spillover strength and receipt support as important determinants of its performance.

\section{Application to PSDP trial}\label{Sec9}
We apply the proposed BART-RE methodology to the PSDP analysis sample described in Section~\ref{Sec2}, comprising $N = 2{,}383$ students nested within $J = 49$ schools. The outcome and treatment models are each specified using $50$ regression trees. The outcome model includes the clever covariate $\hat{\pi}^{\bar{D}}_{j(-i)}$ described in Section~\ref{Sec5}, computed from a parametric probit pre-fit of the treatment receipt model. Posterior inference is based on $10{,}000$ burn-in iterations, followed by $1{,}000$ retained draws (thinned every 5 iterations), with $M = 10{,}000$ Monte Carlo samples per draw for the G-computation step. Table~\ref{Table::CaseStudyResults} summarizes posterior inference for the cluster-level ITT and LATE estimands. Because the outcome is coded so that $Y_{ji}=1$ indicates infection, negative values correspond to reductions in infection risk.

The results suggest that the reduction in infection risk is driven primarily by students' own treatment receipt. For example, the ITT total effect estimate shows that assignment to a treatment school reduces the infection probability by an estimated $0.150$ (95\% CrI: $-0.277,-0.029$). We see that most of this effect is attributable to the direct component, estimated at $-0.123$ (95\% CrI: $-0.272, 0.016$). In contrast, the estimated spillover effect is smaller, with limited posterior evidence for a beneficial effect of the treatment.

The corresponding LATE estimates show a similar pattern, with somewhat larger direct effects among compliers. The estimated spillover effect remains small and uncertain. The larger magnitude of the LATE relative to the ITT is consistent with attenuation of the ITT by treatment noncompliance, since the ITT averages treatment assignment effects over students who do and do not comply. Overall, the strongest evidence concerns a reduction in the total infection risk, with most of the estimated reduction attributable to students' own treatment receipt. We find comparatively little evidence for an additional spillover effect.

The sensitivity analyses preserve the main finding of reduced total infection risk and the negative posterior means of the direct effects (Appendices~\ref{AppendixD} and~\ref{AppendixE}). Under the calibrated departure from conditional expected homogeneity, the estimated total effects remain unchanged because the adjustment enters the direct and spillover components with opposite signs. The adjusted direct-effect means remain negative, although uncertainty increases as the ITT interval widens to $(-0.583,0.350)$ and the LATE interval to $(-0.652,0.292)$, both including zero. Varying latent cross-world receipt dependence produces negligible changes in posterior means over the specified grid of sensitivity parameters. This stability is expected under structural one-sided noncompliance, where control-world receipt is constant, and provides limited evidence about sensitivity to more general receipt dependence. Thus, the sensitivity analyses preserve our overall conclusion while indicating greater uncertainty about how the total effect is divided between individual and peer treatment receipt.

\begin{table}[htbp]
\caption{Posterior summaries of causal effect estimates from the PSDP case study using BART-RE. Negative estimates indicate reductions in infection probability. $\Pr(<0)$ denotes the posterior probability of a negative effect.}
\label{Table::CaseStudyResults}
\centering
\begin{tabular}{@{}llcccc@{}}
\hline
 &  & Post.\ Mean & $2.5\%$ & $97.5\%$ & $\Pr(<0)$ \\
\hline
\textit{ITT effects} & Direct  & $-0.123$ & $-0.272$ & $\phantom{-}0.016$ & $0.959$ \\
          & Spillover  & $-0.027$ & $-0.160$ & $0.098$ & $0.662$ \\
          & Total   & $-0.150$ & $-0.277$ & $-0.029$ & $0.990$ \\[4pt]
\textit{LATE effects} & Direct  & $-0.175$ & $-0.320$ & $-0.034$ & $0.991$ \\
          & Spillover  & $-0.027$ & $-0.161$ & $\phantom{-}0.092$ & $0.664$ \\
          & Total   & $-0.202$ & $-0.324$ & $-0.080$ & $0.999$ \\
\hline
\end{tabular}
\end{table}

\section{Discussion}\label{Sec10}
We developed a Bayesian framework for separating the effects of individual and peer treatment receipt in CRTs with noncompliance and within-cluster interference.  Our BART-RE models combine flexible adjustment for individual- and cluster-level covariates with random intercepts for residual within-cluster dependence. We also introduce a peer-treatment clever covariate to mitigate regularization-induced bias of the direct and spillover components when individual and peer receipt are correlated. Together, the framework provides a flexible approach to estimate interpretable direct and spillover effects in CRTs with noncompliance.


Our framework complements the principal-stratification approach of \citet{cheng2025identification}, who obtain point identification under extended principal ignorability and allow treatment assignment to affect outcomes beyond treatment receipt. We retain network exclusion so that the total effect decomposition describes effects through individual and peer treatment receipt. We also define restrictions on cross-world outcome means and receipt dependence, which do not generally imply principal ignorability. In addition, our direct effect holds peer receipt at its control-assignment level, while their individual compliance effect holds it at its treatment-assignment level. This means the two decompositions can differ when individual and peer receipt interact. Our contribution therefore lies in the receipt-based targets, our identifying assumptions, and the Bayesian inference framework.

Our simulation studies show that BART-RE reduces bias under model misspecification. Comparisons of BART-RE models fitted with and without the clever covariate further show that accurate total-effect estimation can conceal errors in its decomposition between individual and peer receipt. The clever covariate improves the decomposition under moderate or strong spillover, although benefits are mixed under weak spillover. These findings support its role as a regularization tool, while the conservative credible intervals highlight the difficulty of separating individual and peer-level effects when the data contain few or no treated individuals whose peers are all untreated. In the PSDP, we find evidence of reduced infection risk, with a larger estimated contribution from individual treatment receipt. 



Separating the direct and spillover effects requires predicting outcomes for individuals who receive treatment while their peers remain untreated. When this treatment pattern is absent or rare in the observed data, the decomposition relies on extrapolation from the fitted outcome model. Assessing sensitivity to receipt dependence is also limited in the PSDP trial because one-sided noncompliance fixes receipt under control assignment at zero. These limitations motivate further evaluation under two-sided noncompliance, where receipt can vary under both assignments. Future work can also further examine the contribution of the clever covariate and compare the peer treatment proportions used in g-computation with their observed ranges among treated and untreated individuals. These extensions would strengthen the framework's use across CRTs with different compliance patterns.

\section*{Code availability}

All code for the data analysis and simulation studies is available at \url{https://github.com/SBstats/BART-RE-causal-inference-CRT-noncompliance-spillover}.

\section*{Acknowledgements}
Bhandari and Daniels were partially supported by NIH grants R01 HL158963 and R01 HL166324. Li was supported by the Patient-Centered Outcomes Research Institute\textsuperscript{\textregistered} (PCORI\textsuperscript{\textregistered} Award ME-2023C1-31350). The statements presented in this article are solely the responsibility of the authors and do not necessarily represent the views of PCORI\textsuperscript{\textregistered}, its Board of Governors or Methodology Committee.

\bibliographystyle{apalike}
\bibliography{ref}

\newpage

\appendix

\section{Non-parametric identification proofs}\label{AppendixA}
\allowdisplaybreaks
This appendix presents the proof of Proposition~\ref{Prop1:ITTidenditification}. Throughout this section, let $\textbf{b}_j = \big\{ b_j^{Y}, b_j^{D}\big\}$ denote a vector of random effects for cluster $j$ for conciseness.
\subsection{Identification for ITT effects}\label{AppendixA1}
To identify the ITT direct and spillover effects from the observed data distributions, we must be able to identify the following expectations:
\begin{align*}
    & \mathbf{E}\Big[\frac{1}{n_j} \sum_{i=1}^{n_j} \big\{Y_{ji}\big(D_{ji}(1), \textbf{D}_{j(-i)}(1)\big)\big\}\Big], 
    \\
    &  \mathbf{E}\Big[\frac{1}{n_j} \sum_{i=1}^{n_j} \big\{Y_{ji}\big(D_{ji}(0), \textbf{D}_{j(-i)}(0)\big)\big\}\Big], \text{ and }
    \\
    & \mathbf{E}\Big[\frac{1}{n_j} \sum_{i=1}^{n_j} \big\{Y_{ji}\big(D_{ji}(1), \textbf{D}_{j(-i)}(0)\big)\big\}\Big].
\end{align*}

Recall that these expectations are taken over the joint distribution of the cluster-level random variables $\mathbf{W}_j$ and the cluster size $n_j$, denoted 
\[
\mathcal{P}_{W,n} = \mathcal{P}_{W|n} \times \mathcal{P}_{n},
\] 
as implied by the super-population framework. For notational simplicity, let $n$ denote an arbitrary realization drawn from the unknown distribution $\mathcal{P}_{n}$. The expectations above can then be rewritten in terms of this single realization of the cluster size as:

\begin{align*}
    & \mathbf{E}\Big[\frac{1}{n} \sum_{i=1}^{n} \big\{Y_{.i}\big(D_{.i}(1), \textbf{D}_{.(-i)}(1)\big)\big\}\Big],
    \\
    & \mathbf{E}\Big[\frac{1}{n} \sum_{i=1}^{n} \big\{Y_{.i}\big(D_{.i}(0), \textbf{D}_{.(-i)}(0)\big)\big\}\Big], \text{ and }
    \\
    & \mathbf{E}\Big[\frac{1}{n} \sum_{i=1}^{n} \big\{Y_{.i}\big(D_{.i}(1), \textbf{D}_{.(-i)}(0)\big)\big\}\Big].
\end{align*}

In this representation, the notations $Y_{.i}, D_{.i}$ and $\textbf{D}_{.(-i)}$ indicate that the random variables belong to an arbitrary cluster of size $n$. The notation $\textbf{b}_{.} = \big\{ b_{.}^{Y}, b_{.}^{D}\big\}$ denotes the vector of random effects for this cluster. The identification steps for the expectations involved in the direct, spillover, and overall ITT effects are outlined below.

1) Identification steps for $\mathbf{E}\Big[\frac{1}{n} \sum_{i=1}^{n} \big\{Y_{.i}\big(D_{.i}(1), \textbf{D}_{.(-i)}(1)\big)\big\}\Big]$:
    \begin{equation*}
        \begin{aligned}
            & \mathbf{E}\Big[\frac{1}{n} \sum_{i=1}^{n} \big\{Y_{.i}\big(D_{.i}(1), \textbf{D}_{.(-i)}(1)\big)\big\}\Big]
            \\&
            = \mathbf{E}_{\textbf{X},n}\Big[\mathbf{E}\Big[\frac{1}{n} \sum_{i=1}^{n} \big\{Y_{.i}\big(D_{.i}(1), \textbf{D}_{.(-i)}(1)\big)\big\}\Big|\textbf{X}_{.i},n \Big]\Big] \quad \quad \big\{ \text{By the  law of iterated expectations}\big\}
            \\&
            =  \mathbf{E}_{\textbf{X},n}\Big[\int_{d_{.i}}\int_{\textbf{d}_{.(-i)}} \mathbf{E}\Big[\frac{1}{n}\sum_{i=1}^{n} \big\{ Y_{.i}\big(d_{.i}, \textbf{d}_{.(-i)}\big)\big\}\big| D_{.i}(1) = d_{.i}, \textbf{D}_{.(-i)}(1) = \textbf{d}_{.(-i)},   \textbf{X}_{.i} = \textbf{x}_{.i}, n\Big]  \\&  dF_{\textbf{D}_{.(-i)}(1)| D_{.i}(1), \textbf{X}_{.i},n}(\textbf{d}_{.(-i)}) dF_{D_{.i}(1) |  \textbf{X}_{.i},n}(d_{.i}) \Big] 
            \quad \quad \big\{ \text{By the  law of iterated expectations}\big\}
            \\& 
            =  \mathbf{E}_{\textbf{X},n}\Big[\int_{d_{.i}}\int_{\textbf{d}_{.(-i)}} \mathbf{E}\Big[\frac{1}{n}\sum_{i=1}^{n} \big\{ Y_{.i}\big(d_{.i}, \textbf{d}_{.(-i)}\big)\big\}\big| Z = 1, D_{.i}(1) = d_{.i}, \textbf{D}_{.(-i)}(1) = \textbf{d}_{.(-i)},   \textbf{X}_{.i} = \textbf{x}_{.i},n\Big]  \\&  dF_{\textbf{D}_{.(-i)}(1)| Z=1, D_{.i}(1), \textbf{X}_{.i},n}(\textbf{d}_{.(-i)}) dF_{D_{.i}(1) | Z=1,  \textbf{X}_{.i},n}(d_{.i}) \Big] 
            \quad \quad \big\{ \text{By the randomization of treatment assignment}\big\}
            \\& 
            =  \mathbf{E}_{\textbf{X},n}\Big[\frac{1}{n}\sum_{i=1}^{n}\int_{d_{.i}}\int_{\textbf{d}_{.(-i)}} \mathbf{E}\Big[  Y_{.i}\big(d_{.i}, \textbf{d}_{.(-i)}\big)\big| Z = 1, D_{.i}(1) = d_{.i}, \textbf{D}_{.(-i)}(1) = \textbf{d}_{.(-i)},   \textbf{X}_{.i} = \textbf{x}_{.i},n\Big]  \\&  dF_{\textbf{D}_{.(-i)}(1)| Z=1, D_{.i}(1), \textbf{X}_{.i},n}(\textbf{d}_{.(-i)}) dF_{D_{.i}(1) | Z=1,  \textbf{X}_{.i},n}(d_{.i}) \Big] 
            \\& \big\{ \text{By the linearity of expectation and switching the order of summation and integration}\big\}
            \\& 
            =  \mathbf{E}_{\textbf{X},n}\Big[\frac{1}{n}\sum_{i=1}^{n}\int_{d_{.i}}\int_{\textbf{d}_{.(-i)}} \mathbf{E}\Big[  Y_{.i}\big(d_{.i}, \textbf{d}_{.(-i)}\big)\big| Z = 1, D_{.i} = d_{.i}, \textbf{D}_{.(-i)} = \textbf{d}_{.(-i)},   \textbf{X}_{.i} = \textbf{x}_{.i},n\Big]  \\&  dF_{\textbf{D}_{.(-i)}| Z=1, D_{.i}, \textbf{X}_{.i},n}(\textbf{d}_{.(-i)}) dF_{D_{.i} | Z=1,  \textbf{X}_{.i},n}(d_{.i}) \Big]  
            \quad \quad \big\{ \text{By consistency assumptions 1 and 2}\big\}
            \\& 
            =  \mathbf{E}_{\textbf{X},n}\Big[\frac{1}{n}\sum_{i=1}^{n}\int_{d_{.i}}\int_{\textbf{d}_{.(-i)}} \mathbf{E}\Big[  Y_{.i}\big| Z = 1, D_{.i} = d_{.i}, \textbf{D}_{.(-i)} = \textbf{d}_{.(-i)},   \textbf{X}_{.i} = \textbf{x}_{.i},n\Big]  \\&  dF_{\textbf{D}_{.(-i)}| Z=1, D_{.i}, \textbf{X}_{.i},n}(\textbf{d}_{.(-i)}) dF_{D_{.i} | Z=1,  \textbf{X}_{.i},n}(d_{.i}) \Big]
            \quad \quad \big\{ \text{By consistency assumption 3}\big\}
            \\& 
            =  \int_{\textbf{b}_{.}}\mathbf{E}_{\textbf{X},n}\Big[\frac{1}{n}\sum_{i=1}^{n}\int_{d_{.i}}\int_{\textbf{d}_{.(-i)}} \mathbf{E}\Big[  Y_{.i}\big| Z = 1, D_{.i} = d_{.i}, \textbf{D}_{.(-i)} = \textbf{d}_{.(-i)},   \textbf{X}_{.i} = \textbf{x}_{.i},n, b_{.}^{Y}\Big]  \\&  dF_{\textbf{D}_{.(-i)}| Z=1,  \textbf{X}_{.i},n, b_{.}^{D}}(\textbf{d}_{.(-i)}) dF_{D_{.i} | Z=1,  \textbf{X}_{.i},n, b_{.}^{D}}(d_{.i}) \Big] \times  dF_{\textbf{b}_{.}}
        \\& \big\{ \text{By the law of iterated expectations for random effects and assuming } \textbf{D}_{.(-i)} \indep D_{.i} \text{ given } Z, \textbf{X}_{.i},n, b_{.}^{D}\big\}
        \end{aligned}
    \end{equation*}

2)  Identification steps for $\mathbf{E}\Big[\frac{1}{n} \sum_{i=1}^{n} \big\{Y_{.i}\big(D_{.i}(0), \textbf{D}_{.(-i)}(0)\big)\big\}\Big]$: \\
    Following the arguments in the identification steps of $\mathbf{E}\Big[\frac{1}{n} \sum_{i=1}^{n} \big\{Y_{.i}\big(D_{.i}(1), \textbf{D}_{.(-i)}(1)\big)\big\}\Big]$, we get:
    \begin{equation*}
        \begin{aligned}
            & \mathbf{E}\Big[\frac{1}{n} \sum_{i=1}^{n} \big\{Y_{.i}\big(D_{.i}(0), \textbf{D}_{.(-i)}(0)\big)\big\}\Big]
            \\& 
            =  \int_{\textbf{b}_{.}}\mathbf{E}_{\textbf{X},n}\Big[\frac{1}{n}\sum_{i=1}^{n}\int_{d_{.i}}\int_{\textbf{d}_{.(-i)}} \mathbf{E}\Big[  Y_{.i}\big| Z = 0, D_{.i} = d_{.i}, \textbf{D}_{.(-i)} = \textbf{d}_{.(-i)},   \textbf{X}_{.i} = \textbf{x}_{.i},n, b_{.}^{Y}\Big]   dF_{\textbf{D}_{.(-i)}| Z=0,  \textbf{X}_{.i},n, b_{.}^{D}}(\textbf{d}_{.(-i)}) \\&
            dF_{D_{.i} | Z=0,  \textbf{X}_{.i},n, b_{.}^{D}}(d_{.i}) \Big] \times  dF_{\textbf{b}_{.}} 
        \end{aligned}
    \end{equation*}

3)  Identification steps for $\mathbf{E}\Big[\frac{1}{n} \sum_{i=1}^{n} \big\{Y_{.i}\big(D_{.i}(1), \textbf{D}_{.(-i)}(0)\big)\big\}\Big]$:

        \begin{align*}
            & \mathbf{E}\Big[\frac{1}{n} \sum_{i=1}^{n} \big\{Y_{.i}\big(D_{.i}(1), \textbf{D}_{.(-i)}(0)\big)\big\}\Big]
            \\& 
            = \mathbf{E}_{\textbf{X},n}\Big[\mathbf{E}\Big[\frac{1}{n} \sum_{i=1}^{n} \big\{Y_{.i}\big(D_{.i}(1), \textbf{D}_{.(-i)}(0)\big)\big\}\Big|\textbf{X}_{.i},n \Big]\Big] \quad \quad \big\{ \text{By the  law of iterated expectations}\big\}
            \\&
            =  \mathbf{E}_{\textbf{X},n}\Big[\int_{d_{.i}}\int_{\textbf{d}_{.(-i)}} \mathbf{E}\Big[\frac{1}{n}\sum_{i=1}^{n} \big\{ Y_{.i}\big(d_{.i}, \textbf{d}_{.(-i)}\big)\big\}\big| D_{.i}(1) = d_{.i}, \textbf{D}_{.(-i)}(0) = \textbf{d}_{.(-i)},   \textbf{X}_{.i} = \textbf{x}_{.i}, n\Big]  \\&  dF_{\textbf{D}_{.(-i)}(0)| D_{.i}(1), \textbf{X}_{.i},n}(\textbf{d}_{.(-i)}) dF_{D_{.i}(1) |\textbf{X}_{.i},n}(d_{.i}) \Big] 
            \quad \quad \big\{ \text{By the  law of iterated expectations}\big\}
            \\&
            =  \mathbf{E}_{\textbf{X},n}\Big[\int_{d_{.i}}\int_{\textbf{d}_{.(-i)}} \mathbf{E}\Big[\frac{1}{n}\sum_{i=1}^{n} \big\{ Y_{.i}\big(d_{.i}, \textbf{d}_{.(-i)}\big)\big\}\big| D_{.i}(1) = d_{.i}, \textbf{D}_{.(-i)}(0) = \textbf{d}_{.(-i)},   \textbf{X}_{.i} = \textbf{x}_{.i}, n\Big]  \\&   dF_{\textbf{D}_{.(-i)}(0)|  \textbf{X}_{.i},n}(\textbf{d}_{.(-i)}) dF_{D_{.i}(1) |\textbf{X}_{.i},n}(d_{.i})\Big] 
            \quad \quad \big\{ \text{By Assumption 7}\big\}
            \\&
            =  \mathbf{E}_{\textbf{X},n}\Big[\int_{d_{.i}}\int_{\textbf{d}_{.(-i)}} \mathbf{E}\Big[\frac{1}{n}\sum_{i=1}^{n} \big\{ Y_{.i}\big(d_{.i}, \textbf{d}_{.(-i)}\big)\big\}\big| Z=1, D_{.i}(1) = d_{.i}, \textbf{D}_{.(-i)}(0) = \textbf{d}_{.(-i)},   \textbf{X}_{.i} = \textbf{x}_{.i}, n\Big]  \\&   dF_{\textbf{D}_{.(-i)}(0)|  Z=1, \textbf{X}_{.i},n}(\textbf{d}_{.(-i)}) dF_{D_{.i}(1) |Z=1, \textbf{X}_{.i},n}(d_{.i})\Big]  
            \quad \quad \big\{ \text{By the randomization of treatment assignment (Assumption 2)}\big\}
            \\&
            =  \mathbf{E}_{\textbf{X},n}\Big[\int_{d_{.i}}\int_{\textbf{d}_{.(-i)}} \mathbf{E}\Big[\frac{1}{n}\sum_{i=1}^{n} \big\{ Y_{.i}\big(d_{.i}, \textbf{d}_{.(-i)}\big)\big\}\big| Z=1, D_{.i}(1) = d_{.i}, \textbf{D}_{.(-i)}(1) = \textbf{d}_{.(-i)},   \textbf{X}_{.i} = \textbf{x}_{.i}, n\Big]  \\&   dF_{\textbf{D}_{.(-i)}(0)|  Z=1, \textbf{X}_{.i},n}(\textbf{d}_{.(-i)}) dF_{D_{.i}(1) |Z=1, \textbf{X}_{.i},n}(d_{.i})\Big] \quad \quad
           \big\{ \text{By Assumption 6}\big\}
           \\&
            =  \mathbf{E}_{\textbf{X},n}\Big[\int_{d_{.i}}\int_{\textbf{d}_{.(-i)}} \mathbf{E}\Big[\frac{1}{n}\sum_{i=1}^{n} \big\{ Y_{.i}\big(d_{.i}, \textbf{d}_{.(-i)}\big)\big\}\big| Z=1, D_{.i}(1) = d_{.i}, \textbf{D}_{.(-i)}(1) = \textbf{d}_{.(-i)},   \textbf{X}_{.i} = \textbf{x}_{.i}, n\Big]  \\&   dF_{\textbf{D}_{.(-i)}(0)|  Z=0, \textbf{X}_{.i},n}(\textbf{d}_{.(-i)}) dF_{D_{.i}(1) |Z=1, \textbf{X}_{.i},n}(d_{.i})\Big] \quad \quad
           \big\{ \text{By the randomization of treatment assignment (Assumption 2)}\big\}
            \\&
            =  \mathbf{E}_{\textbf{X},n}\Big[\frac{1}{n}\sum_{i=1}^{n}\int_{d_{.i}}\int_{\textbf{d}_{.(-i)}}  \mathbf{E}\Big[  Y_{.i}\big(d_{.i}, \textbf{d}_{.(-i)}\big)\big| Z=1, D_{.i}(1) = d_{.i}, \textbf{D}_{.(-i)}(1) = \textbf{d}_{.(-i)},  \textbf{X}_{.i} = \textbf{x}_{.i}, n\Big]  \\&  dF_{\textbf{D}_{.(-i)}(0)|  Z=0, \textbf{X}_{.i},n}(\textbf{d}_{.(-i)}) dF_{D_{.i}(1) |Z=1, \textbf{X}_{.i},n}(d_{.i})\Big] \\&
           \big\{ \text{By the linearity of expectation and switching the order of summation and integration}\big\}
            \\&
            =  \mathbf{E}_{\textbf{X},n}\Big[\frac{1}{n}\sum_{i=1}^{n}\int_{d_{.i}}\int_{\textbf{d}_{.(-i)}}  \mathbf{E}\Big[  Y_{.i}\big| Z=1, D_{.i} = d_{.i}, \textbf{D}_{.(-i)} = \textbf{d}_{.(-i)},  \textbf{X}_{.i} = \textbf{x}_{.i}, n\Big]  \\&  dF_{\textbf{D}_{.(-i)}|  Z=0, \textbf{X}_{.i},n}(\textbf{d}_{.(-i)}) dF_{D_{.i} |Z=1, \textbf{X}_{.i},n}(d_{.i})\Big]
             \quad \quad \big\{ \text{By sequentially applying consistency assumptions 1, 2 and 3}  \big\}
             \\& 
             =  \int_{\textbf{b}_{.}}\mathbf{E}_{\textbf{X},n}\Big[\frac{1}{n}\sum_{i=1}^{n}\int_{d_{.i}}\int_{\textbf{d}_{.(-i)}}  \mathbf{E}\Big[  Y_{.i}\big| Z=1, D_{.i} = d_{.i}, \textbf{D}_{.(-i)} = \textbf{d}_{.(-i)},  \textbf{X}_{.i} = \textbf{x}_{.i}, n, b_{.}^{Y}\Big]  \\&  dF_{\textbf{D}_{.(-i)}|  Z=0, \textbf{X}_{.i},n, b_{.}^{D}}(\textbf{d}_{.(-i)}) dF_{D_{.i} |Z=1, \textbf{X}_{.i},n, b_{.}^{D}}(d_{.i})\Big] \times  dF_{\textbf{b}_{.}}
        \\& \big\{ \text{By the law of iterated expectations for random effects and assuming } \textbf{D}_{.(-i)} \indep D_{.i} \text{ given } Z, \textbf{X}_{.i},n, b_{.}^{D}\big\}
        \end{align*}

Therefore, the cluster-ITT direct effects can be identified as:
\begin{align}
    \text{ITT}^C_{DE} &= \mathbf{E}\Big[\frac{1}{n_j} \sum_{i=1}^{n_j} \big\{Y_{ji}\big(D_{ji}(1), \textbf{D}_{j(-i)}(0)\big) - Y_{ji}\big(D_{ji}(0), \textbf{D}_{j(-i)}(0)\big)\big\}\Big]
    \nonumber\\&
    = \int_{\textbf{b}_{.}}\mathbf{E}_{\textbf{X},n}\Big[\frac{1}{n}\sum_{i=1}^{n}\int_{d_{.i}}\int_{\textbf{d}_{.(-i)}}  \mathbf{E}\Big[  Y_{.i}\big| Z=1, D_{.i} = d_{.i}, \textbf{D}_{.(-i)} = \textbf{d}_{.(-i)},  \textbf{X}_{.i} = \textbf{x}_{.i}, n, b_{.}^{Y}\Big]   dF_{\textbf{D}_{.(-i)}|  Z=0, \textbf{X}_{.i},n, b_{.}^{D}}(\textbf{d}_{.(-i)})\nonumber\\&  dF_{D_{.i} |Z=1, \textbf{X}_{.i},n, b_{.}^{D}}(d_{.i})\Big] \times  dF_{\textbf{b}_{.}}
    \nonumber\\&\quad-
    \nonumber\\& \int_{\textbf{b}_{.}}\mathbf{E}_{\textbf{X},n}\Big[\frac{1}{n}\sum_{i=1}^{n}\int_{d_{.i}}\int_{\textbf{d}_{.(-i)}}  \mathbf{E}\Big[  Y_{.i}\big| Z=0, D_{.i} = d_{.i}, \textbf{D}_{.(-i)} = \textbf{d}_{.(-i)},  \textbf{X}_{.i} = \textbf{x}_{.i}, n, b_{.}^{Y}\Big]   dF_{\textbf{D}_{.(-i)}|  Z=0, \textbf{X}_{.i},n, b_{.}^{D}}(\textbf{d}_{.(-i)})\nonumber\\&  dF_{D_{.i} |Z=0, \textbf{X}_{.i},n, b_{.}^{D}}(d_{.i})\Big] \times  dF_{\textbf{b}_{.}}
\end{align}

Similarly, the cluster-ITT spillover effects can be identified as:
\begin{align}
    \text{ITT}^C_{SE} &= \mathbf{E}\Big[\frac{1}{n_j} \sum_{i=1}^{n_j} \big\{Y_{ji}\big(D_{ji}(1), \textbf{D}_{j(-i)}(1)\big) - Y_{ji}\big(D_{ji}(1), \textbf{D}_{j(-i)}(0)\big)\big\}\Big]
    \nonumber\\&
    = \int_{\textbf{b}_{.}}\mathbf{E}_{\textbf{X},n}\Big[\frac{1}{n}\sum_{i=1}^{n}\int_{d_{.i}}\int_{\textbf{d}_{.(-i)}}  \mathbf{E}\Big[  Y_{.i}\big| Z=1, D_{.i} = d_{.i}, \textbf{D}_{.(-i)} = \textbf{d}_{.(-i)},  \textbf{X}_{.i} = \textbf{x}_{.i}, n, b_{.}^{Y}\Big]   dF_{\textbf{D}_{.(-i)}|  Z=1, \textbf{X}_{.i},n, b_{.}^{D}}(\textbf{d}_{.(-i)})\nonumber\\&  dF_{D_{.i} |Z=1, \textbf{X}_{.i},n, b_{.}^{D}}(d_{.i})\Big] \times  dF_{\textbf{b}_{.}}
    \nonumber\\&\quad-
    \nonumber\\&       \int_{\textbf{b}_{.}}\mathbf{E}_{\textbf{X},n}\Big[\frac{1}{n}\sum_{i=1}^{n}\int_{d_{.i}}\int_{\textbf{d}_{.(-i)}}  \mathbf{E}\Big[  Y_{.i}\big| Z=1, D_{.i} = d_{.i}, \textbf{D}_{.(-i)} = \textbf{d}_{.(-i)},  \textbf{X}_{.i} = \textbf{x}_{.i}, n, b_{.}^{Y}\Big]   dF_{\textbf{D}_{.(-i)}|  Z=0, \textbf{X}_{.i},n, b_{.}^{D}}(\textbf{d}_{.(-i)})\nonumber\\&  dF_{D_{.i} |Z=1, \textbf{X}_{.i},n, b_{.}^{D}}(d_{.i})\Big] \times  dF_{\textbf{b}_{.}}
\end{align}

These two effects decompose the cluster-ITT overall effects, which can be identified as:
\begin{align}
    \text{ITT}^C &= \mathbf{E}\Big[\frac{1}{n_j} \sum_{i=1}^{n_j} \big\{Y_{ji}\big(D_{ji}(1), \textbf{D}_{j(-i)}(1)\big) - Y_{ji}\big(D_{ji}(0), \textbf{D}_{j(-i)}(0)\big)\big\}\Big]
    \nonumber\\&
    = \int_{\textbf{b}_{.}}\mathbf{E}_{\textbf{X},n}\Big[\frac{1}{n}\sum_{i=1}^{n}\int_{d_{.i}}\int_{\textbf{d}_{.(-i)}}  \mathbf{E}\Big[  Y_{.i}\big| Z=1, D_{.i} = d_{.i}, \textbf{D}_{.(-i)} = \textbf{d}_{.(-i)},  \textbf{X}_{.i} = \textbf{x}_{.i}, n, b_{.}^{Y}\Big]   dF_{\textbf{D}_{.(-i)}|  Z=1, \textbf{X}_{.i},n, b_{.}^{D}}(\textbf{d}_{.(-i)})\nonumber\\&  dF_{D_{.i} |Z=1, \textbf{X}_{.i},n, b_{.}^{D}}(d_{.i})\Big] \times  dF_{\textbf{b}_{.}}
    \nonumber\\&\quad-
    \nonumber\\& \int_{\textbf{b}_{.}}\mathbf{E}_{\textbf{X},n}\Big[\frac{1}{n}\sum_{i=1}^{n}\int_{d_{.i}}\int_{\textbf{d}_{.(-i)}}  \mathbf{E}\Big[  Y_{.i}\big| Z=0, D_{.i} = d_{.i}, \textbf{D}_{.(-i)} = \textbf{d}_{.(-i)},  \textbf{X}_{.i} = \textbf{x}_{.i}, n, b_{.}^{Y}\Big]   dF_{\textbf{D}_{.(-i)}|  Z=0, \textbf{X}_{.i},n, b_{.}^{D}}(\textbf{d}_{.(-i)})\nonumber\\&  dF_{D_{.i} |Z=0, \textbf{X}_{.i},n, b_{.}^{D}}(d_{.i})\Big] \times  dF_{\textbf{b}_{.}}
\end{align}


\subsection{Identification for LATE effects}\label{AppendixA2}
To identify the LATE direct and spillover effects from the observed data distributions, we must be able to identify the following expectations:
\begin{equation*}
    \begin{aligned}
        \mathbf{E}\Big[\frac{1}{n_j} \sum_{i=1}^{n_j} \big\{Y_{ji}\big(D_{ji}(1), \textbf{D}_{j(-i)}(1)\big)\big\}\big| D_{ji}(1) =1, D_{ji}(0) = 0\Big],
    \end{aligned}
\end{equation*}
\begin{equation*}
    \begin{aligned}
        \mathbf{E}\Big[\frac{1}{n_j} \sum_{i=1}^{n_j} \big\{Y_{ji}\big(D_{ji}(0), \textbf{D}_{j(-i)}(0)\big)\big\}\big| D_{ji}(1) =1, D_{ji}(0) = 0\Big],
    \end{aligned}
\end{equation*}
and 
\begin{equation*}
    \begin{aligned}
        \mathbf{E}\Big[\frac{1}{n_j} \sum_{i=1}^{n_j} \big\{Y_{ji}\big(D_{ji}(1), \textbf{D}_{j(-i)}(0)\big)\big\}\big| D_{ji}(1) =1, D_{ji}(0) = 0\Big].
    \end{aligned}
\end{equation*}

Like before, these expectations are taken over the joint distribution of the cluster-level random variables $\mathbf{W}_j$ and the cluster size $n_j$, denoted 
\[
\mathcal{P}_{W,n} = \mathcal{P}_{W|n} \times \mathcal{P}_{n},
\] 
as implied by the super-population framework. Again, for notational simplicity, let $n$ denote an arbitrary realization drawn from the unknown distribution $\mathcal{P}_{n}$. The expectations above can then be rewritten in terms of this single realization of the cluster size as:

\begin{equation*}
    \begin{aligned}
        \mathbf{E}\Big[\frac{1}{n} \sum_{i=1}^{n} \big\{Y_{.i}\big(D_{.i}(1), \textbf{D}_{.(-i)}(1)\big)\big\}\big| D_{.i}(1) =1, D_{.i}(0) = 0\Big],
    \end{aligned}
\end{equation*}
\begin{equation*}
    \begin{aligned}
        \mathbf{E}\Big[\frac{1}{n} \sum_{i=1}^{n} \big\{Y_{.i}\big(D_{.i}(0), \textbf{D}_{.(-i)}(0)\big)\big\}\big| D_{.i}(1) =1, D_{.i}(0) = 0\Big],
    \end{aligned}
\end{equation*}
and 
\begin{equation*}
    \begin{aligned}
        \mathbf{E}\Big[\frac{1}{n} \sum_{i=1}^{n} \big\{Y_{.i}\big(D_{.i}(1), \textbf{D}_{.(-i)}(0)\big)\big\}\big| D_{.i}(1) =1, D_{.i}(0) = 0\Big].
    \end{aligned}
\end{equation*}
In this representation, the notations $Y_{.i}, D_{.i}$ and $\textbf{D}_{.(-i)}$ indicate that the random variables belong to an arbitrary cluster of size $n$. Similarly, the notation $\textbf{b}_{.} = \big\{ b_{.}^{Y}, b_{.}^{D}\big\}$ denotes the vector of random effects for this cluster.
The identification steps for the expectations involved in the direct, spillover, and overall LATE effects are outlined below.


1)  Identification steps for $\mathbf{E}\Big[\frac{1}{n} \sum_{i=1}^{n} \big\{Y_{.i}\big(D_{.i}(1), \textbf{D}_{.(-i)}(1)\big)\big\}\big| D_{.i}(1) =1, D_{.i}(0) = 0\Big]$:
    \begin{align*}
            & \mathbf{E}\Big[\frac{1}{n} \sum_{i=1}^{n} \big\{Y_{.i}\big(D_{.i}(1), \textbf{D}_{.(-i)}(1)\big)\big\}\big| D_{.i}(1) =1, D_{.i}(0) = 0\Big]
            \\&
            = \mathbf{E}_{\textbf{X},n}\Big[\mathbf{E}\Big[\frac{1}{n} \sum_{i=1}^{n} \big\{Y_{.i}\big(D_{.i}(1), \textbf{D}_{.(-i)}(1)\big)\big\}\Big|D_{.i}(1) =1, D_{.i}(0) = 0, \textbf{X}_{.i},n \Big]\Big] \quad \quad \big\{ \text{By the  law of iterated expectations}\big\}
            \\& 
            =  \mathbf{E}_{\textbf{X},n}\Big[ \int_{\textbf{d}_{.(-i)}} \mathbf{E}\Big[\frac{1}{n}\sum_{i=1}^{n} \big\{ Y_{.i}\big(D_{.i}(1) = 1, \textbf{D}_{.(-i)}(1) = \textbf{d}_{.(-i)}\big)\big\}\big| D_{.i}(1) = 1, D_{.i}(0) = 0, \textbf{D}_{.(-i)}(1) = \textbf{d}_{.(-i)},   \textbf{X}_{.i} = \textbf{x}_{.i}, n\Big]  \\&  dF_{\textbf{D}_{.(-i)}(1)| D_{.i}(1) =1, D_{.i}(0) =0, \textbf{X}_{.i},n}(\textbf{d}_{.(-i)})\Big] 
            \quad \quad \big\{ \text{By the  law of iterated expectations}\big\}
            \\&
            =  \mathbf{E}_{\textbf{X},n}\Big[ \int_{\textbf{d}_{.(-i)}} \mathbf{E}\Big[\frac{1}{n}\sum_{i=1}^{n} \big\{ Y_{.i}\big(D_{.i}(1) = 1, \textbf{D}_{.(-i)}(1) = \textbf{d}_{.(-i)}\big)\big\}\big| Z =1, D_{.i}(1) = 1, D_{.i}(0) = 0, \textbf{D}_{.(-i)}(1) = \textbf{d}_{.(-i)}, \\&  \textbf{X}_{.i} = \textbf{x}_{.i}, n\Big]    dF_{\textbf{D}_{.(-i)}(1)|  Z =1, D_{.i}(1) =1, D_{.i}(0) =0, \textbf{X}_{.i},n}(\textbf{d}_{.(-i)})\Big] 
            \quad \quad \big\{ \text{By the randomization of treatment assignment}\big\}
            \\&
            =  \mathbf{E}_{\textbf{X},n}\Big[ \int_{\textbf{d}_{.(-i)}} \mathbf{E}\Big[\frac{1}{n}\sum_{i=1}^{n} \big\{ Y_{.i}\big(1, \textbf{d}_{.(-i)}\big)\big\}\big| Z =1, D_{.i} = 1,  \textbf{D}_{.(-i)}(1) = \textbf{d}_{.(-i)},   \textbf{X}_{.i} = \textbf{x}_{.i}, n\Big]  \\&  dF_{\textbf{D}_{.(-i)}(1)|  Z =1, D_{.i} =1, \textbf{X}_{.i},n}(\textbf{d}_{.(-i)})\Big] 
            \\& \big\{  \text{Given } Z = 1, \text{ for complier stratum } (D_{.i}(1) =1,  D_{.i}(0) =0), \text{ we have } D_{.i} = Z\times D_{.i}(1) + (1-Z) \times D_{.i}(0)  =1 \big\}
            \\&
            =  \mathbf{E}_{\textbf{X},n}\Big[ \frac{1}{n}\sum_{i=1}^{n}\int_{\textbf{d}_{.(-i)}} \mathbf{E}\Big[  Y_{.i}\big(1, \textbf{d}_{.(-i)}\big)\big| Z =1, D_{.i} = 1,  \textbf{D}_{.(-i)}(1) = \textbf{d}_{.(-i)},   \textbf{X}_{.i} = \textbf{x}_{.i}, n\Big]  \\&  dF_{\textbf{D}_{.(-i)}(1)|  Z =1, D_{.i} =1, \textbf{X}_{.i},n}(\textbf{d}_{.(-i)})\Big] 
            \\& \big\{  \text{By the linearity of expectation and switching the order of summation and integration}\big\}
            \\&
            =  \mathbf{E}_{\textbf{X},n}\Big[ \frac{1}{n}\sum_{i=1}^{n}\int_{\textbf{d}_{.(-i)}} \mathbf{E}\Big[  Y_{.i}\big(1, \textbf{d}_{.(-i)}\big)\big| Z =1, D_{.i} = 1,  \textbf{D}_{.(-i)} = \textbf{d}_{.(-i)},   \textbf{X}_{.i} = \textbf{x}_{.i}, n\Big]  \\&  dF_{\textbf{D}_{.(-i)}|  Z =1, D_{.i} =1, \textbf{X}_{.i},n}(\textbf{d}_{.(-i)})\Big] 
            \quad \quad \big\{  \text{By consistency assumption 2}\big\}
            \\&
            =  \mathbf{E}_{\textbf{X},n}\Big[ \frac{1}{n}\sum_{i=1}^{n}\int_{\textbf{d}_{.(-i)}} \mathbf{E}\Big[  Y_{.i}\big| Z =1, D_{.i} = 1,  \textbf{D}_{.(-i)} = \textbf{d}_{.(-i)},   \textbf{X}_{.i} = \textbf{x}_{.i}, n\Big]    dF_{\textbf{D}_{.(-i)}|  Z =1, D_{.i} =1, \textbf{X}_{.i},n}(\textbf{d}_{.(-i)})\Big] 
            \\& \big\{  \text{By consistency assumption 3}\big\}
            \\& 
            =  \int_{\textbf{b}_{.}}\mathbf{E}_{\textbf{X},n}\Big[\frac{1}{n}\sum_{i=1}^{n} \int_{\textbf{d}_{.(-i)}} \mathbf{E}\Big[  Y_{.i}\big| Z = 1, D_{.i} = 1, \textbf{D}_{.(-i)} = \textbf{d}_{.(-i)},   \textbf{X}_{.i} = \textbf{x}_{.i},n, b_{.}^{Y}\Big]    dF_{\textbf{D}_{.(-i)}| Z = 1, \textbf{X}_{.i},n, b_{.}^{D}}(\textbf{d}_{.(-i)})\Big] \times  dF_{\textbf{b}_{.}}
        \\& \big\{ \text{By the law of iterated expectations for random effects and assuming } \textbf{D}_{.(-i)} \indep D_{.i} \text{ given } Z, \textbf{X}_{.i},n, b_{.}^{D}\big\}
    \end{align*}

2) Identification steps for $\mathbf{E}\Big[\frac{1}{n} \sum_{i=1}^{n} \big\{Y_{.i}\big(D_{.i}(0), \textbf{D}_{.(-i)}(0)\big)\big\}\big| D_{.i}(1) =1, D_{.i}(0) = 0\Big]$: \\
    Following the arguments in the identification steps of $\mathbf{E}\Big[\frac{1}{n} \sum_{i=1}^{n} \big\{Y_{.i}\big(D_{.i}(1), \textbf{D}_{.(-i)}(1)\big)\big\}\big| D_{.i}(1) =1, D_{.i}(0) = 0\Big]$ with $Z =0$, we get:
\begin{equation*}
        \begin{aligned}
            & \mathbf{E}\Big[\frac{1}{n} \sum_{i=1}^{n} \big\{Y_{.i}\big(D_{.i}(0), \textbf{D}_{.(-i)}(0)\big)\big\}\big| D_{.i}(1) =1, D_{.i}(0) = 0\Big]
            \\&
            =  \int_{\textbf{b}_{.}}\mathbf{E}_{\textbf{X},n}\Big[\frac{1}{n}\sum_{i=1}^{n} \int_{\textbf{d}_{.(-i)}} \mathbf{E}\Big[  Y_{.i}\big| Z = 0, D_{.i} = 0, \textbf{D}_{.(-i)} = \textbf{d}_{.(-i)},   \textbf{X}_{.i} = \textbf{x}_{.i},n, b_{.}^{Y}\Big]    dF_{\textbf{D}_{.(-i)}| Z = 0, \textbf{X}_{.i},n, b_{.}^{D}}(\textbf{d}_{.(-i)})\Big] \times  dF_{\textbf{b}_{.}}
        \end{aligned}
    \end{equation*}

 3) Identification steps for $\mathbf{E}\Big[\frac{1}{n} \sum_{i=1}^{n} \big\{Y_{.i}\big(D_{.i}(1), \textbf{D}_{.(-i)}(0)\big)\big\}\big| D_{.i}(1) =1, D_{.i}(0) = 0\Big]$:
    \begin{align*}
            & \mathbf{E}\Big[\frac{1}{n} \sum_{i=1}^{n} \big\{Y_{.i}\big(D_{.i}(1), \textbf{D}_{.(-i)}(0)\big)\big\}\big| D_{.i}(1) =1, D_{.i}(0) = 0\Big]
            \\&
            = \mathbf{E}_{\textbf{X},n}\Big[\mathbf{E}\Big[\frac{1}{n} \sum_{i=1}^{n} \big\{Y_{.i}\big(D_{.i}(1), \textbf{D}_{.(-i)}(0)\big)\big\}\Big|D_{.i}(1) =1, D_{.i}(0) = 0, \textbf{X}_{.i},n \Big]\Big] \quad \quad \big\{ \text{By the  law of iterated expectations}\big\}
            \\& 
            =  \mathbf{E}_{\textbf{X},n}\Big[ \int_{\textbf{d}_{.(-i)}} \mathbf{E}\Big[\frac{1}{n}\sum_{i=1}^{n} \big\{ Y_{.i}\big(D_{.i}(1) = 1, \textbf{D}_{.(-i)}(0) = \textbf{d}_{.(-i)}\big)\big\}\big| D_{.i}(1) = 1, D_{.i}(0) = 0, \textbf{D}_{.(-i)}(0) = \textbf{d}_{.(-i)},  \\&  \textbf{X}_{.i} = \textbf{x}_{.i}, n\Big]   dF_{\textbf{D}_{.(-i)}(0)| D_{.i}(1) =1, D_{.i}(0) =0, \textbf{X}_{.i},n}(\textbf{d}_{.(-i)})\Big] 
            \quad \quad \big\{ \text{By the  law of iterated expectations}\big\}
            \\& 
            =  \mathbf{E}_{\textbf{X},n}\Big[ \int_{\textbf{d}_{.(-i)}} \mathbf{E}\Big[\frac{1}{n}\sum_{i=1}^{n} \big\{ Y_{.i}\big( 1,  \textbf{d}_{.(-i)}\big)\big\}\big| D_{.i}(1) = 1, D_{.i}(0) = 0, \textbf{D}_{.(-i)}(0) = \textbf{d}_{.(-i)},    \textbf{X}_{.i} = \textbf{x}_{.i}, n\Big]   \\& dF_{\textbf{D}_{.(-i)}(0)|  D_{.i}(0) =0, \textbf{X}_{.i},n}(\textbf{d}_{.(-i)})\Big] 
            \quad \quad \big\{ \text{By Assumption 7}\big\}
            \\&
            =  \mathbf{E}_{\textbf{X},n}\Big[ \int_{\textbf{d}_{.(-i)}} \mathbf{E}\Big[\frac{1}{n}\sum_{i=1}^{n} \big\{ Y_{.i}\big(1, \textbf{d}_{.(-i)}\big)\big\}\big|Z=1, D_{.i}(1) = 1, D_{.i}(0) = 0, \textbf{D}_{.(-i)}(0) = \textbf{d}_{.(-i)},   \textbf{X}_{.i} = \textbf{x}_{.i}, n\Big]  \\&  dF_{\textbf{D}_{.(-i)}(0)|Z=1,  D_{.i}(0) =0, \textbf{X}_{.i},n}(\textbf{d}_{.(-i)})\Big]
            \quad \quad \big\{ \text{By the  randomization of treatment assignment (Assumption 2)}\big\}
            \\&
            =  \mathbf{E}_{\textbf{X},n}\Big[ \int_{\textbf{d}_{.(-i)}} \mathbf{E}\Big[\frac{1}{n}\sum_{i=1}^{n} \big\{ Y_{.i}\big( 1, \textbf{d}_{.(-i)}\big)\big\}\big|Z=1, D_{.i} = 1,  \textbf{D}_{.(-i)}(0) = \textbf{d}_{.(-i)},    \textbf{X}_{.i} = \textbf{x}_{.i}, n\Big] \\&  dF_{\textbf{D}_{.(-i)}(0)|Z=1,  D_{.i}(0) =0, \textbf{X}_{.i},n}(\textbf{d}_{.(-i)})\Big]
            \\& \big\{   \text{Given } Z = 1, \text{ for complier stratum } (D_{.i}(1) =1,  D_{.i}(0) =0), \text{ we have } D_{.i} = Z\times D_{.i}(1) + (1-Z) \times D_{.i}(0)  =1\big\}
            \\&
            =  \mathbf{E}_{\textbf{X},n}\Big[ \int_{\textbf{d}_{.(-i)}} \mathbf{E}\Big[\frac{1}{n}\sum_{i=1}^{n} \big\{ Y_{.i}\big(1, \textbf{d}_{.(-i)}\big)\big\}\big|Z=1, D_{.i} = 1,  \textbf{D}_{.(-i)}(1) = \textbf{d}_{.(-i)},    \textbf{X}_{.i} = \textbf{x}_{.i}, n\Big] \\&  dF_{\textbf{D}_{.(-i)}(0)|Z=1,  D_{.i}(0) =0, \textbf{X}_{.i},n}(\textbf{d}_{.(-i)})\Big]
            \quad \quad \big\{  \text{By Assumption 6}\big\}
            \\& 
            =  \mathbf{E}_{\textbf{X},n}\Big[ \int_{\textbf{d}_{.(-i)}} \mathbf{E}\Big[\frac{1}{n}\sum_{i=1}^{n} \big\{ Y_{.i}\big( 1,  \textbf{d}_{.(-i)}\big)\big\}\big|Z=1, D_{.i} = 1,  \textbf{D}_{.(-i)}(1) = \textbf{d}_{.(-i)},    \textbf{X}_{.i} = \textbf{x}_{.i}, n\Big]  \\& dF_{\textbf{D}_{.(-i)}(0)|Z=0,  D_{.i}(0) =0, \textbf{X}_{.i},n}(\textbf{d}_{.(-i)})\Big] 
            \quad  \big\{ \text{By the randomization of treatment assignment (Assumption 2)}\big\}
            \\&
            =  \mathbf{E}_{\textbf{X},n}\Big[ \frac{1}{n}\sum_{i=1}^{n}\int_{\textbf{d}_{.(-i)}} \mathbf{E}\Big[  Y_{.i}\big(1,  \textbf{d}_{.(-i)}\big)\big| Z=1, D_{.i} = 1,  \textbf{D}_{.(-i)}(1) = \textbf{d}_{.(-i)},    \textbf{X}_{.i} = \textbf{x}_{.i}, n\Big]   \\& dF_{\textbf{D}_{.(-i)}(0)|Z=0,  D_{.i}(0) =0, \textbf{X}_{.i},n}(\textbf{d}_{.(-i)})\Big] 
            \\& \big\{ \text{By the linearity of expectation and switching the order of summation and integration}  \big\}
            \\&
            =  \mathbf{E}_{\textbf{X},n}\Big[ \frac{1}{n}\sum_{i=1}^{n}\int_{\textbf{d}_{.(-i)}} \mathbf{E}\Big[  Y_{.i}\big(1,  \textbf{d}_{.(-i)}\big)\big| Z =1,  D_{.i} = 1, \textbf{D}_{.(-i)} = \textbf{d}_{.(-i)},  \textbf{X}_{.i} = \textbf{x}_{.i}, n\Big]   dF_{\textbf{D}_{.(-i)}|Z =0, D_{.i} =0,  \textbf{X}_{.i},n}  (\textbf{d}_{.(-i)})\Big] 
            \\&    \big\{ \text{By consistency assumption 1 }  \big\}
            \\&
            =  \mathbf{E}_{\textbf{X},n}\Big[ \frac{1}{n}\sum_{i=1}^{n}\int_{\textbf{d}_{.(-i)}} \mathbf{E}\Big[  Y_{.i}\big| Z =1,  D_{.i} = 1, \textbf{D}_{.(-i)} = \textbf{d}_{.(-i)},  \textbf{X}_{.i} = \textbf{x}_{.i}, n\Big]     dF_{\textbf{D}_{.(-i)}|Z =0, D_{.i} =0,  \textbf{X}_{.i},n}  (\textbf{d}_{.(-i)})\Big] 
            \\& \big\{ \text{By consistency assumption 2 }  \big\}
            \\& 
            =  \int_{\textbf{b}_{.}}\mathbf{E}_{\textbf{X},n}\Big[\frac{1}{n}\sum_{i=1}^{n} \int_{\textbf{d}_{.(-i)}} \mathbf{E}\Big[  Y_{.i}\big| Z = 1, D_{.i} = 1, \textbf{D}_{.(-i)} = \textbf{d}_{.(-i)},   \textbf{X}_{.i} = \textbf{x}_{.i},n, b_{.}^{Y}\Big]    dF_{\textbf{D}_{.(-i)}| Z = 0, \textbf{X}_{.i},n, b_{.}^{D}}(\textbf{d}_{.(-i)})\Big] \times  dF_{\textbf{b}_{.}}
        \\& \big\{ \text{By the law of iterated expectations for random effects and assuming } \textbf{D}_{.(-i)} \indep D_{.i} \text{ given } Z, \textbf{X}_{.i},n, b_{.}^{D}\big\}
    \end{align*}

Therefore, the cluster-LATE direct effects can be identified as:
{\footnotemathsize
\begin{align}
    \text{LATE}^C_{DE} &= \mathbf{E}\Big[\frac{1}{n_j} \sum_{i=1}^{n_j} \big\{Y_{ji}\big(D_{ji}(1), \textbf{D}_{j(-i)}(0)\big) - Y_{ji}\big(D_{ji}(0), \textbf{D}_{j(-i)}(0)\big)\big\} \big| D_{ji}(1) =1, D_{ji}(0) = 0\Big]
    \nonumber\\&
    = \int_{\textbf{b}_{.}}\mathbf{E}_{\textbf{X},n}\Big[\frac{1}{n}\sum_{i=1}^{n} \int_{\textbf{d}_{.(-i)}} \mathbf{E}\Big[  Y_{.i}\big| Z = 1, D_{.i} = 1, \textbf{D}_{.(-i)} = \textbf{d}_{.(-i)},   \textbf{X}_{.i} = \textbf{x}_{.i},n, b_{.}^{Y}\Big]    dF_{\textbf{D}_{.(-i)}| Z = 0, \textbf{X}_{.i},n, b_{.}^{D}}(\textbf{d}_{.(-i)})\Big] \times  dF_{\textbf{b}_{.}}
    \nonumber\\&\quad-
    \nonumber\\& \int_{\textbf{b}_{.}}\mathbf{E}_{\textbf{X},n}\Big[\frac{1}{n}\sum_{i=1}^{n} \int_{\textbf{d}_{.(-i)}} \mathbf{E}\Big[  Y_{.i}\big| Z = 0, D_{.i} = 0, \textbf{D}_{.(-i)} = \textbf{d}_{.(-i)},   \textbf{X}_{.i} = \textbf{x}_{.i},n, b_{.}^{Y}\Big]    dF_{\textbf{D}_{.(-i)}| Z = 0, \textbf{X}_{.i},n, b_{.}^{D}}(\textbf{d}_{.(-i)})\Big] \times  dF_{\textbf{b}_{.}}
\end{align}}

Similarly, the cluster-LATE spillover effects can be identified as:
{\footnotemathsize
\begin{align}
    \text{LATE}^C_{SE} &= \mathbf{E}\Big[\frac{1}{n_j} \sum_{i=1}^{n_j} \big\{Y_{ji}\big(D_{ji}(1), \textbf{D}_{j(-i)}(1)\big) - Y_{ji}\big(D_{ji}(1), \textbf{D}_{j(-i)}(0)\big)\big\} \big| D_{ji}(1) =1, D_{ji}(0) = 0\Big]
    \nonumber\\&
    = \int_{\textbf{b}_{.}}\mathbf{E}_{\textbf{X},n}\Big[\frac{1}{n}\sum_{i=1}^{n} \int_{\textbf{d}_{.(-i)}} \mathbf{E}\Big[  Y_{.i}\big| Z = 1, D_{.i} = 1, \textbf{D}_{.(-i)} = \textbf{d}_{.(-i)},   \textbf{X}_{.i} = \textbf{x}_{.i},n, b_{.}^{Y}\Big]    dF_{\textbf{D}_{.(-i)}| Z = 1, \textbf{X}_{.i},n, b_{.}^{D}}(\textbf{d}_{.(-i)})\Big] \times  dF_{\textbf{b}_{.}}
    \nonumber\\&\quad-
    \nonumber\\& \int_{\textbf{b}_{.}}\mathbf{E}_{\textbf{X},n}\Big[\frac{1}{n}\sum_{i=1}^{n} \int_{\textbf{d}_{.(-i)}} \mathbf{E}\Big[  Y_{.i}\big| Z = 1, D_{.i} = 1, \textbf{D}_{.(-i)} = \textbf{d}_{.(-i)},   \textbf{X}_{.i} = \textbf{x}_{.i},n, b_{.}^{Y}\Big]    dF_{\textbf{D}_{.(-i)}| Z = 0, \textbf{X}_{.i},n, b_{.}^{D}}(\textbf{d}_{.(-i)})\Big] \times  dF_{\textbf{b}_{.}}
\end{align}}

These two effects decompose the cluster-LATE overall effects, which can be identified as:
{\footnotemathsize
\begin{align}
    \text{LATE}^C &= \mathbf{E}\Big[\frac{1}{n_j} \sum_{i=1}^{n_j} \big\{Y_{ji}\big(D_{ji}(1), \textbf{D}_{j(-i)}(1)\big) - Y_{ji}\big(D_{ji}(0), \textbf{D}_{j(-i)}(0)\big)\big\} \big| D_{ji}(1) =1, D_{ji}(0) = 0\Big]
    \nonumber\\&
    = \int_{\textbf{b}_{.}}\mathbf{E}_{\textbf{X},n}\Big[\frac{1}{n}\sum_{i=1}^{n} \int_{\textbf{d}_{.(-i)}} \mathbf{E}\Big[  Y_{.i}\big| Z = 1, D_{.i} = 1, \textbf{D}_{.(-i)} = \textbf{d}_{.(-i)},   \textbf{X}_{.i} = \textbf{x}_{.i},n, b_{.}^{Y}\Big]    dF_{\textbf{D}_{.(-i)}| Z = 1, \textbf{X}_{.i},n, b_{.}^{D}}(\textbf{d}_{.(-i)})\Big] \times  dF_{\textbf{b}_{.}}
    \nonumber\\&\quad-
    \nonumber\\& \int_{\textbf{b}_{.}}\mathbf{E}_{\textbf{X},n}\Big[\frac{1}{n}\sum_{i=1}^{n} \int_{\textbf{d}_{.(-i)}} \mathbf{E}\Big[  Y_{.i}\big| Z = 0, D_{.i} = 0, \textbf{D}_{.(-i)} = \textbf{d}_{.(-i)},   \textbf{X}_{.i} = \textbf{x}_{.i},n, b_{.}^{Y}\Big]    dF_{\textbf{D}_{.(-i)}| Z = 0, \textbf{X}_{.i},n, b_{.}^{D}}(\textbf{d}_{.(-i)})\Big] \times  dF_{\textbf{b}_{.}}
\end{align}}

\section{Details on MCMC sampling framework for BART with cluster-level random effects}\label{AppendixB}
In this section, we provide details on the Markov chain Monte Carlo (MCMC) algorithm for posterior inference under the models described in Section~\ref{Sec5}. The sampler fits (i) a probit-BART model for the binary outcome with an additive cluster-level Gaussian random effect $b^{Y}_j$, and (ii) a probit-BART model for treatment receipt with an additive cluster-level Gaussian random effect $b^{D}_j$. Both models share the same Albert--Chib latent-variable augmentation structure. We first restate the models and priors, then provide the full conditional distributions and the blocked Gibbs sampler.

\subsection{Outcome model (binary data):}
\begin{align}
Y_{ji} \mid f^Y, b_j^{Y} &\sim \operatorname{Bernoulli}\big(\Phi\{f^Y(Z_j, D_{ji}, \bar{D}_{j(-i)}, \mathbf{X}_{ji}, n_j, \hat{\pi}^{\bar{D}}_{j(-i)}) + b_j^{Y}\}\big),
\label{eq:outcome}
\end{align}
where $f^Y(\cdot)$ is represented by a BART ensemble of $T_Y$ regression trees:
\[
f^Y(\cdot) = \sum_{t =1}^{T_Y} g_{t}^Y\big(\cdot;\mathcal{T}^Y_t,\mathcal{M}^Y_t \big),
\]
with $\mathcal{T}^Y_t$ denoting the tree structure and $\mathcal{M}^Y_t$ denoting the terminal-node parameters for tree $t$. The cluster effect $b_j^{Y}$ enters additively outside $f^Y(\cdot)$, ensuring separation between fixed-effect (BART) and random-effect components.

\subsection*{Priors:}
\begin{align*}
 \mathcal{T}^Y_t,\mathcal{M}^Y_t &\text{: standard probit-BART priors (tree depth/birth/death probabilities and } \mu\text{-prior)}, \\
b_j^{Y} &\overset{\text{i.i.d.}}{\sim}\mathcal{N}(0,\tau_{b^Y}^2),\qquad
\tau_{b^Y}^2 \sim \pi_{\tau_Y}(\cdot).
\end{align*}

\subsection*{Latent-variable augmentation for probit-BART (outcome model):}

Introduce latent $U^Y_{ji}$ such that
\[
U^Y_{ji} = f^Y(Z_j, D_{ji}, \bar{D}_{j(-i)}, \mathbf{X}_{ji}, n_j, \hat{\pi}^{\bar{D}}_{j(-i)}) + b_j^{Y} + \varepsilon^{Y}_{ji},\quad \varepsilon^{Y}_{ji}\sim \mathcal{N}(0,1),
\]
with $Y_{ji}=1$ iff $U^Y_{ji}>0$. Conditional on $Y_{ji}$, each $U^Y_{ji}$ is drawn from a truncated normal:
\[
U^Y_{ji}\mid Y_{ji}=1 \sim \mathcal{TN}_{(0,\infty)}\big(f^Y(\cdot)+b_j^Y,1\big), \quad
U^Y_{ji}\mid Y_{ji}=0 \sim \mathcal{TN}_{(-\infty,0]}\big(f^Y(\cdot)+b_j^Y,1\big).
\]

\subsection*{Full posterior:}

Denote by $\mathcal{D}^Y=\{Y_{ji},Z_j,D_{ji},\bar D_{j(-i)},\mathbf{X}_{ji},n_j,\hat{\pi}^{\bar{D}}_{j(-i)}\}$ the observed data for the outcome model. The posterior, augmented with latent utilities $\{U^Y_{ji}\}$, is given by:
\[
p\Big(\{\mathcal{T}^Y_t,\mathcal{M}^Y_t\},\{U^Y_{ji}\},\{b_j^{Y}\},\tau_{b^Y}^2 \mid \mathcal{D}^Y \Big).
\]
Conditional on the latent utilities $U^Y_{ji}$ and cluster-level effects $b_j^Y$, the BART trees are updated using standard Metropolis--Hastings moves with partial residuals $R^Y_{ji,(t)} = U^Y_{ji} - b_j^Y - \sum_{s \ne t} g_s^Y(\cdot)$.

\subsection{Treatment receipt model (binary data):}
\begin{align}
D_{ji} \mid f^D, b_j^{D} &\sim \operatorname{Bernoulli}\big(\Phi\{f^D(Z_j,\mathbf{X}_{ji}, n_j) + b_j^{D}\}\big),
\label{eq:receipt}
\end{align}
with
\[
f^D(\cdot) = \sum_{t =1}^{T_D} g_{t}^D\big(\cdot;\mathcal{T}^D_t,\mathcal{M}^D_t \big).
\]
Here, $b_j^{D}$ also enters additively outside $f^D(\cdot)$.

\subsection*{Priors:} 
\begin{align*}
\mathcal{T}^D_t,\mathcal{M}^D_t &\text{: standard probit-BART priors}, \\
b_j^{D} &\overset{\text{i.i.d.}}{\sim}\mathcal{N}(0,\tau_{b^D}^2),\qquad
\tau_{b^D}^2 \sim \pi_{\tau_D}(\cdot).
\end{align*}

\subsection*{Full posterior:}

Denote by $\mathcal{D}^D=\{D_{ji},Z_j,\mathbf{X}_{ji},n_j\}$ the observed data for the treatment receipt status model. The posterior, augmented with latent utilities $\{U^D_{ji}\}$, is given by:
\[
p\Big(\{\mathcal{T}^D_t,\mathcal{M}^D_t\},\{U^D_{ji}\},\{b_j^{D}\},\tau_{b^D}^2 \mid \mathcal{D}^D \Big).
\]
Conditional on latent utilities $U^D_{ji}$ and cluster-level effects $b_j^D$, each BART tree is updated via standard Metropolis--Hastings steps with partial residuals $R^D_{ji,(t)} = U^D_{ji} - b_j^D - \sum_{s \ne t} g_s^D(\cdot)$.

\subsection*{Latent-variable augmentation for probit-BART (treatment receipt model):}

Introduce latent $U^D_{ji}$ such that
\[
U^D_{ji} = f^D(Z_j,\mathbf{X}_{ji},n_j) + b_j^{D} + \varepsilon^{D}_{ji},\quad \varepsilon^{D}_{ji}\sim \mathcal{N}(0,1),
\]
with $D_{ji}=1$ iff $U^D_{ji}>0$. Conditional on $D_{ji}$, each $U^D_{ji}$ is drawn from a truncated normal:
\[
U^D_{ji}\mid D_{ji}=1 \sim \mathcal{TN}_{(0,\infty)}\big(f^D(\cdot)+b_j^D,1\big), \quad
U^D_{ji}\mid D_{ji}=0 \sim \mathcal{TN}_{(-\infty,0]}\big(f^D(\cdot)+b_j^D,1\big).
\]

\subsection{MCMC components}

\subsubsection*{1. BART tree updates}
Conditional on all other parameters (including the current values of cluster-level predictors $b_j^Y$ or $b_j^D$), update each tree $t=1,\dots,T$ in the ensemble using standard BART moves:
\begin{itemize}
\item Grow / Prune / Change / Swap proposals for $\mathcal{T}_t$,
\item Compute marginal likelihood for affected terminal nodes using residuals/partial residuals,
\item Accept/reject via Metropolis--Hastings ratio.
\end{itemize}

For each model independently, conditional on current latent utilities and cluster effects, update each tree $t$ via standard BART Metropolis--Hastings using partial residuals:
\[
R^Y_{ji,(t)} = U^Y_{ji} - b_j^Y - \sum_{s\ne t} g_s^Y(\cdot), \quad
R^D_{ji,(t)} = U^D_{ji} - b_j^D - \sum_{s\ne t} g_s^D(\cdot).
\]
Since both models operate on Gaussian latent utilities with unit variance, the tree updates for $f^Y$ and $f^D$ follow the standard continuous-BART Metropolis--Hastings procedure applied to the pseudo-responses $U^Y_{ji}$ and $U^D_{ji}$, respectively.

\subsubsection*{2. Latent utilities $U^Y_{ji}$ and $U^D_{ji}$}
For the outcome model, conditional on the current tree fit $f^Y_{ji}$ and cluster effect $b_j^Y$, draw
\[
U^Y_{ji}\mid Y_{ji}, f^Y_{ji}, b_j^Y \;\sim\;
\begin{cases}
\mathcal{TN}_{(0,\infty)}\big(f^Y_{ji}+b_j^Y,\,1\big), & Y_{ji}=1,\\[4pt]
\mathcal{TN}_{(-\infty,0]}\big(f^Y_{ji}+b_j^Y,\,1\big), & Y_{ji}=0.
\end{cases}
\]
For the treatment receipt model, conditional on $f^D_{ji}$ and $b_j^D$, draw
\[
U^D_{ji}\mid D_{ji}, f^D_{ji}, b_j^D \;\sim\;
\begin{cases}
\mathcal{TN}_{(0,\infty)}\big(f^D_{ji}+b_j^D,\,1\big), & D_{ji}=1,\\[4pt]
\mathcal{TN}_{(-\infty,0]}\big(f^D_{ji}+b_j^D,\,1\big), & D_{ji}=0,
\end{cases}
\]
where \(\mathcal{TN}_{(a,b)}(\mu,\sigma^2)\) denotes a normal distribution with mean \(\mu\) and variance \(\sigma^2\) truncated to the interval \((a,b)\).

\subsubsection*{3. Conjugate updates for cluster-level random effects \(b_j^Y\) and \(b_j^D\)}

The full conditionals of random effects are Gaussian and admit closed-form Gibbs updates (conditioning on the current values of the BART fits and hyperparameters). Denote the current fitted tree contributions by
\[
f^Y_{ji} \equiv f^Y(Z_j,D_{ji},\bar D_{j(-i)},\mathbf{X}_{ji},n_j,\hat{\pi}^{\bar{D}}_{j(-i)}),\qquad
f^D_{ji} \equiv f^D(Z_j,\mathbf{X}_{ji},n_j),
\]
evaluated at the current tree parameters.

\paragraph{Outcome model (probit).}
Using the latent augmentation $U^Y_{ji}$ (sampled from truncated normals as described above), the likelihood for cluster $j$ is
\[
U^Y_{ji}\mid f^Y_{ji}, b_j^Y \overset{\text{ind}}{\sim} \mathcal{N}(f^Y_{ji}+b_j^Y,\,1),\qquad i=1,\dots,n_j,
\]
with prior $b_j^Y\sim\mathcal{N}(0,\tau_{b^Y}^2)$. Define
\[
S_j^Y \;=\; \sum_{i=1}^{n_j}\big(U^Y_{ji}-f^Y_{ji}\big).
\]
The posterior precision and variance are
\[
\lambda_{j}^Y \;=\; n_j + \frac{1}{\tau_{b^Y}^2}, \qquad
\mathrm{Var}(b_j^Y\mid\text{rest}) \;=\; (\lambda_{j}^Y)^{-1}.
\]
The posterior mean is
\[
\mathbf{E}[b_j^Y\mid\text{rest}] \;=\; (\lambda_{j}^Y)^{-1}\,S_j^Y.
\]
Hence the Gibbs update is
\[
b_j^Y \mid \text{rest} \;\sim\; \mathcal{N}\!\left( (\lambda_{j}^Y)^{-1}S_j^Y,\; (\lambda_{j}^Y)^{-1} \right).
\]

\paragraph{Treatment receipt status (probit) model.}  
Using the latent augmentation $U^D_{ji}$ (sampled from truncated normals as described above), the likelihood for cluster $j$ is
\[
U^D_{ji}\mid f^D_{ji}, b_j^D \overset{\text{ind}}{\sim} \mathcal{N}(f^D_{ji}+b_j^D,\,1),\qquad i=1,\dots,n_j,
\]
with prior \(b_j^D\sim\mathcal{N}(0,\tau_{b^D}^2)\). Define
\[
S_j^D \;=\; \sum_{i=1}^{n_j}\big(U^D_{ji}-f^D_{ji}\big).
\]
The posterior precision and variance are
\[
\lambda_{j}^D \;=\; n_j + \frac{1}{\tau_{b^D}^2}, \qquad
\mathrm{Var}(b_j^D\mid\text{rest}) \;=\; (\lambda_{j}^D)^{-1}.
\]
The posterior mean is
\[
\mathbf{E}[b_j^D\mid\text{rest}] \;=\; (\lambda_{j}^D)^{-1}\,S_j^D,
\]
and the Gibbs update is
\[
b_j^D \mid \text{rest} \;\sim\; \mathcal{N}\!\left( (\lambda_{j}^D)^{-1}S_j^D,\; (\lambda_{j}^D)^{-1} \right).
\]

Note that the quantities \(f^Y_{ji}\) and \(f^D_{ji}\) are evaluated at the current BART tree parameters when computing \(S_j^Y\) and \(S_j^D\).

\subsubsection*{4. Hyperpriors $\tau_{b^Y}^2$ and $\tau_{b^D}^2$}
If conjugate inverse-gamma prior:
\[
\tau_{b^Y}^2 \mid \{b_j^Y\} \sim \text{Inv-Gamma}\Big(a_{\tau_Y}+\frac{J}{2},\, b_{\tau_Y}+\frac{1}{2}\sum_{j=1}^J (b_j^Y)^2\Big),
\]
and analogously for $\tau_{b^D}^2$.

The above steps are repeated for a large number of iterations, with the first $B$ iterations discarded as burn-in, thinning as needed, and $Q$ iterations retained to collect posterior draws and posterior predictive samples for each model independently.

\section{Additional modeling and estimation details}\label{AppendixC}

\subsection{Super-population framework}\label{AppendixC1}

We view the $J$ observed clusters, each containing $n_j$ individuals, as independent draws from a super-population of clusters with finite sizes. We assume that $n_j \sim \mathcal{P}_n$, where $\mathcal{P}_n$ is an unknown distribution with finite support on $\mathbf{N}^{+}$. Under independent cluster-level randomization, we factor the joint distribution as 
\begin{equation*}
    \mathcal{P}_{W,Z,n}
=
\mathcal{P}_{Y\mid D,X,Z,n}
\times \mathcal{P}_{D\mid X,Z,n}
\times \mathcal{P}_{X\mid n}
\times \mathcal{P}_{n}
\times \mathcal{P}_{Z}.
\end{equation*}

This formulation accommodates variation in cluster size and the dimension of $\mathbf{W}_j$, with a common marginal distribution obtained by mixing over $\mathcal{P}_n$. We adopt this framework to define average cluster effects in a target population beyond the clusters enrolled in the trial (Section~\ref{Sec53}), assigning equal weight to each cluster. We express these estimands as expectations of cluster-level potential-outcome summaries over the joint distribution of cluster sizes and baseline covariates, and use iterated expectations to develop their identification. Similar super-population frameworks appear in \citet{wang2024model} and \citet{ohnishi2026bayesian}.

\subsection{G-computation algorithm for cluster-LATE effects}

The cluster-LATE estimand $\Theta_{(z,z_{*})}^{LATE}$ is defined in Equation~\eqref{Eq::GcompEstimand}, and Proposition~\ref{Prop1:ITTidenditification} expresses the cluster-LATE direct, spillover, and total effects as contrasts of $\Theta_{(z,z_{*})}^{LATE}$ (Equation~\eqref{Eq::Prop2LATE}).

The difference between the G-computation algorithm for the cluster-ITT estimand (Algorithm~\ref{Algo1}) and that for the cluster-LATE effects (Algorithm~\ref{Algo2}) lies in the treatment receipt modeling. Specifically, for LATE, we do not sample from the posterior predictive model for the subject’s treatment receipt $D_{ji}$. Instead, we fix $D_{ji}$ deterministically to $Z_j$, thereby restricting to the complier stratum where treatment assignment equals treatment receipt.

\begin{algorithm}[htbp]
\caption{G-computation algorithm for $\Theta_{(z,z_{*})}^{LATE,(q)}$ under the super-population framework }\label{Algo2}
\begin{algorithmic}[1]
\State Compute the clever covariate $\hat{\pi}^{\bar{D}}_{j(-i)}(z_*)$ via Equation~\eqref{Eq::CleverCovariate} for $z_* \in \{0,1\}$ and append $\hat{\pi}^{\bar{D}}_{j(-i)}(Z_j)$ as a fixed column of the outcome design matrix.
\State Fit the BART models specified in Section~\ref{Sec5}. Discard the first $B$ MCMC iterations as burn-in, and retain $Q$ posterior draws.
\For{posterior draw $q = 1,\dots,Q$}
    \State Sample weights $(p^{(q)}_1,\ldots,p^{(q)}_J) \sim \mathrm{Dirichlet}(1,\ldots,1)$.
    \State For each cluster $j = 1,\dots,J$, draw the cluster random effects $b^{Y,(q)}_{j} \sim \mathcal{N}\big(0, \tau^{2,(q)}_{b^{Y}}\big)$ and $b^{D,(q)}_{j} \sim \mathcal{N}\big(0, \tau^{2,(q)}_{b^{D}}\big)$, where $\tau^{2,(q)}_{b^{Y}}$ and $\tau^{2,(q)}_{b^{D}}$ are the $q$-th posterior draws of the variance components.
    \For{cluster $j = 1,\dots,J$}
        \For{individual $i = 1,\dots,n_j$}
            \For{Monte Carlo replicate $m = 1,\dots,M$}
                \State Fix $D_{ji}^{(q,m)} = z$ if $Z_j = z$ and $D_{ji}^{(q,m)} = z_{*}$ if $Z_j = z_{*}$ (complier stratum).
                \State Sample $\mathbf{D}_{j(-i)}^{(q,m)} \sim F_{\mathbf{D}_{j(-i)}\mid Z_j=z_{*},\mathbf{X}_{ji}, n_j, b^{D,(q)}_{j}}$.
                \State Compute $\widehat{y}_{ji}^{(q,m)} = \Phi\Big(f^{Y,(q)}\big(z, D_{ji}^{(q,m)}, \bar{D}_{j(-i)}^{(q,m)}, \mathbf{X}_{ji}, n_j, \hat{\pi}^{\bar{D}}_{j(-i)}(z_*)\big) + b^{Y,(q)}_{j}\Big)$, where  $\bar{D}_{j(-i)}^{(q,m)} = \frac{1}{n_j -1}\sum_{\ell \neq i}D_{j\ell}^{(q,m)}$, and $\hat{\pi}^{\bar{D}}_{j(-i)}(z_*)$ is the clever covariate evaluated at the peer assignment $z_*$.
            \EndFor
            \State $\widehat{\Phi}_{ji}^{LATE,(q)} \gets \frac{1}{M}\sum_{m=1}^M \widehat{y}_{ji}^{(q,m)}$.
        \EndFor
        \State $\widehat{\Phi}_{j}^{LATE,(q)}(\mathbf{X}_j,n_j) \gets \frac{1}{n_j}\sum_{i=1}^{n_j} \widehat{\Phi}_{ji}^{LATE,(q)}$.
    \EndFor
    \State $\widehat{\Theta}_{(z,z_{*})}^{LATE,(q)} \gets \sum_{j=1}^J p_j^{(q)} \,\widehat{\Phi}_{j}^{LATE,(q)}(\mathbf{X}_j,n_j)$.
\EndFor
\State Posterior samples: $\{\widehat{\Theta}_{(z,z_{*})}^{LATE,(q)}\}_{q=1}^Q$.
\end{algorithmic}
\end{algorithm}

\subsection{Additional details on baseline covariates}\label{AppendixC3}
The PSDP analysis adjusts for $p=13$ baseline covariates. The nine individual-level covariates are sex, age, household latrine availability, livestock ownership, weight-for-age $z$-score, blood in stool, illness frequency, hand cleanliness, and malaria status. The four school-level covariates are the zone baseline infection rate, school examination score, distance to Lake Victoria, and latrines per pupil. Table~\ref{Table::PSDP_baseline} summarizes these covariates by treatment assignment.

\section{Sensitivity analysis for the violation of Assumption 6}\label{AppendixD}

Recall that Assumption 6 states the following:
\begin{equation}\label{Assump:6SA} 
    \begin{aligned}
        & \mathbf{E}\Big[\sum_{i=1}^{n} \big\{ Y_{.i}\big(D_{.i}(z) = d_{.i}, \textbf{D}_{.(-i)}(z_{*}) = \textbf{d}_{.(-i)}\big)\big\}\big| Z=z, D_{.i}(z) = d_{.i}, \textbf{D}_{.(-i)}(z_{*}) = \textbf{d}_{.(-i)},   \textbf{X}_{.i} = \textbf{x}_{.i}, n\Big]
        \\&
        = \mathbf{E}\Big[\sum_{i=1}^{n} \big\{ Y_{.i}\big(D_{.i}(z) = d_{.i}, \textbf{D}_{.(-i)}(z) = \textbf{d}_{.(-i)}\big)\big\}\big| Z=z, D_{.i}(z) = d_{.i}, \textbf{D}_{.(-i)}(z) = \textbf{d}_{.(-i)},   \textbf{X}_{.i} = \textbf{x}_{.i}, n\Big] 
    \end{aligned}
\end{equation}
for  $z \neq z_{*}$.

In this section, we evaluate the robustness of the estimated cluster-ITT and cluster-LATE effects to violations of Assumption 6, assuming all other identification assumptions hold. Our proposed sensitivity analysis framework is adapted from \cite{daniels2023bayesian}.

\subsection{Sensitivity parameter}
Suppose that all identification assumptions hold except for Assumption 6. Our sensitivity analysis is based on the following parameter measuring departure from (\ref{Assump:6SA}): 
\begin{align*}
    \rho\big(1,0\big) 
    &= \mathbf{E}\Big[\frac{1}{n}\sum_{i=1}^{n} \big\{ Y_{.i}\big(D_{.i}(1) = d_{.i}, \textbf{D}_{.(-i)}(0) = \textbf{d}_{.(-i)}\big)\big\}\big| Z=1, D_{.i}(1) = d_{.i}, \textbf{D}_{.(-i)}(0) = \textbf{d}_{.(-i)},   \textbf{X}_{.i} = \textbf{x}_{.i}, n\Big]
        \\&
        - \mathbf{E}\Big[\frac{1}{n}\sum_{i=1}^{n} \big\{ Y_{.i}\big(D_{.i}(1) = d_{.i}, \textbf{D}_{.(-i)}(1) = \textbf{d}_{.(-i)}\big)\big\}\big| Z=1, D_{.i}(1) = d_{.i}, \textbf{D}_{.(-i)}(1) = \textbf{d}_{.(-i)},   \textbf{X}_{.i} = \textbf{x}_{.i}, n\Big]
\end{align*}
Using Assumption 2 the sensitivity parameters  $\rho\big(1,0\big)$ can be simplified to:

\begin{equation}\label{Eq::SAReducedrho}
\begin{aligned}
    \rho\big(1,0\big) 
    &= \mathbf{E}\Big[\frac{1}{n}\sum_{i=1}^{n} \big\{ Y_{.i}\big(D_{.i}(1) = d_{.i}, \textbf{D}_{.(-i)}(0) = \textbf{d}_{.(-i)}\big)\big\}\big|  D_{.i}(1) = d_{.i}, \textbf{D}_{.(-i)}(0) = \textbf{d}_{.(-i)},   \textbf{X}_{.i} = \textbf{x}_{.i}, n\Big]
        \\&
        - \mathbf{E}\Big[\frac{1}{n}\sum_{i=1}^{n} \big\{ Y_{.i}\big(D_{.i}(1) = d_{.i}, \textbf{D}_{.(-i)}(1) = \textbf{d}_{.(-i)}\big)\big\}\big|  D_{.i}(1) = d_{.i}, \textbf{D}_{.(-i)}(1) = \textbf{d}_{.(-i)},   \textbf{X}_{.i} = \textbf{x}_{.i}, n\Big].
\end{aligned}
\end{equation}

\subsection{Calibrating sensitivity parameters}
To calibrate $\rho\big(1,0\big)$ using the observed data, we begin by computing the total amount of variability in the outcomes $Y_{ji}$ under $Z_j = 1$ that is explained by $d_{ji}$, $\textbf{d}_{j(-i)}$, $\textbf{x}_{ji}$ and $n_j$. Specifically, we compute the coefficient of determination from the regression of $Y$ on $\big(d_{ji},\textbf{d}_{j(-i)},\textbf{x}_{ji},n_j \big)$ for subjects with $Z_j =1$, denoted by $R_1$. 

Now, suppose $\kappa$ represents the percentage of total variance not explained by $d_{ji}$, $\textbf{d}_{j(-i)}$, $\textbf{x}_{ji}$ and $n_j$. Using $R_1$ and $\kappa$, we assume a bound on $\rho\big(1,0\big)$ of the form:
\begin{equation}\label{Eq::SArhobound}
    \big|\rho\big(1,0\big)\big| \leq \sqrt{\text{Var}(Y_{ji}\big|Z_j = 1)\times (1-R_1) \times \kappa}.
\end{equation}

Intuitively, this implies that any violation of Assumption 6 is expected to have an effect bounded by the variance of $Y_{ji}\big|Z_j = 1$ scaled by $100 \times \kappa\%$ of the variability not explained by $d_{ji}$, $\textbf{d}_{j(-i)}$, $\textbf{x}_{ji}$ and $n_j$. Thus, we define the sensitivity parameter $\kappa$ bounded in $[0,1]$. 

We place a $\text{Beta}(2,1)$ prior on $\kappa$, which has density $f(\kappa) = 2\kappa$ on $[0,1]$ and corresponds to a triangular distribution that increases linearly from zero to its mode at $\kappa = 1$. This choice reflects the view that, a priori, larger departures from Assumption~6 should receive at least as much posterior mass as smaller ones, thereby ensuring that the sensitivity analysis meaningfully stresses the identifying assumption. A $\text{Uniform}(0,1)$ prior, by contrast, assigns half of its mass to $\kappa \in [0, 0.5]$---the region in which the sensitivity bound (\ref{Eq::SArhobound}) is smallest---and would therefore concentrate the adjusted estimates near their unadjusted values, understating the potential impact of violations of Assumption~6.

\subsection{Estimation of causal effects under the violation of Assumption 6, assuming all other assumptions hold}
To conduct posterior inference for the cluster-ITT and cluster-LATE estimates, we begin by taking posterior samples from the usual difference in means to compute the effect estimates, assuming Assumption 6 holds. Next, we generate random draws, one for each MCMC iteration, from the distribution of the sensitivity parameter $\kappa$ and use these to compute the bound on $\rho\big(1,0\big)$ given in (\ref{Eq::SArhobound}). Finally, we adjust the causal effect estimates obtained under Assumption 6 by subtracting the bound values from the estimates. Mathematically, we can express 
the sensitivity-adjusted cluster-ITT direct and spillover effects  using $\rho\big(1,0\big)$ as:

\begin{align}
    \text{ITT}^C_{DE} &= \mathbf{E}\Big[\frac{1}{n_j} \sum_{i=1}^{n_j} \big\{Y_{ji}\big(D_{ji}(1), \textbf{D}_{j(-i)}(0)\big) - Y_{ji}\big(D_{ji}(0), \textbf{D}_{j(-i)}(0)\big)\big\}\Big]
    \nonumber\\&
    = \bigg\{ \rho\big(1,0\big) +\int_{\textbf{b}_{.}}\mathbf{E}_{\textbf{X},n}\Big[\frac{1}{n}\sum_{i=1}^{n}\int_{d_{.i}}\int_{\textbf{d}_{.(-i)}}   \mathbf{E}\Big[  Y_{.i}\big| Z=1, D_{.i} = d_{.i}, \textbf{D}_{.(-i)} = \textbf{d}_{.(-i)},  \textbf{X}_{.i} = \textbf{x}_{.i}, n, b_{.}^{Y}\Big]   \nonumber\\& dF_{\textbf{D}_{.(-i)}|  Z=0, \textbf{X}_{.i},n, b_{.}^{D}}(\textbf{d}_{.(-i)})  dF_{D_{.i} |Z=1, \textbf{X}_{.i},n, b_{.}^{D}}(d_{.i})\Big] \times  dF_{\textbf{b}_{.}}\bigg\}
    \nonumber\\&\quad-
    \nonumber\\& \int_{\textbf{b}_{.}}\mathbf{E}_{\textbf{X},n}\Big[\frac{1}{n}\sum_{i=1}^{n}\int_{d_{.i}}\int_{\textbf{d}_{.(-i)}}  \mathbf{E}\Big[  Y_{.i}\big| Z=0, D_{.i} = d_{.i}, \textbf{D}_{.(-i)} = \textbf{d}_{.(-i)},  \textbf{X}_{.i} = \textbf{x}_{.i}, n, b_{.}^{Y}\Big]   dF_{\textbf{D}_{.(-i)}|  Z=0, \textbf{X}_{.i},n, b_{.}^{D}}(\textbf{d}_{.(-i)})\nonumber\\&  dF_{D_{.i} |Z=0, \textbf{X}_{.i},n, b_{.}^{D}}(d_{.i})\Big] \times
    dF_{\textbf{b}_{.}}, \text{ and }
\end{align}
\begin{align}
    \text{ITT}^C_{SE} &= \mathbf{E}\Big[\frac{1}{n_j} \sum_{i=1}^{n_j} \big\{Y_{ji}\big(D_{ji}(1), \textbf{D}_{j(-i)}(1)\big) - Y_{ji}\big(D_{ji}(1), \textbf{D}_{j(-i)}(0)\big)\big\}\Big]
    \nonumber\\&
    = \int_{\textbf{b}_{.}}\mathbf{E}_{\textbf{X},n}\Big[\frac{1}{n}\sum_{i=1}^{n}\int_{d_{.i}}\int_{\textbf{d}_{.(-i)}}  \mathbf{E}\Big[  Y_{.i}\big| Z=1, D_{.i} = d_{.i}, \textbf{D}_{.(-i)} = \textbf{d}_{.(-i)},  \textbf{X}_{.i} = \textbf{x}_{.i}, n, b_{.}^{Y}\Big]   dF_{\textbf{D}_{.(-i)}|  Z=1, \textbf{X}_{.i},n, b_{.}^{D}}(\textbf{d}_{.(-i)})\nonumber\\&  dF_{D_{.i} |Z=1, \textbf{X}_{.i},n, b_{.}^{D}}(d_{.i})\Big] \times  dF_{\textbf{b}_{.}}
    \nonumber\\&\quad-
    \nonumber\\&       \bigg\{ \rho\big(1,0\big) + \int_{\textbf{b}_{.}}\mathbf{E}_{\textbf{X},n}\Big[\frac{1}{n}\sum_{i=1}^{n}\int_{d_{.i}}\int_{\textbf{d}_{.(-i)}}   \mathbf{E}\Big[  Y_{.i}\big| Z=1, D_{.i} = d_{.i}, \textbf{D}_{.(-i)} = \textbf{d}_{.(-i)},  \textbf{X}_{.i} = \textbf{x}_{.i}, n, b_{.}^{Y}\Big] \nonumber\\&  dF_{\textbf{D}_{.(-i)}|  Z=0, \textbf{X}_{.i},n, b_{.}^{D}}(\textbf{d}_{.(-i)})  dF_{D_{.i} |Z=1, \textbf{X}_{.i},n, b_{.}^{D}}(d_{.i})\Big] \times  dF_{\textbf{b}_{.}} \bigg\}.
\end{align}
Note that we do not require sensitivity analysis for the total effect estimates, as Assumption 6 is only necessary for the identification of direct and spillover effect parameters.

Similarly, we can express
the sensitivity-adjusted cluster-LATE direct and spillover effects  using $\rho\big(1,0\big)$ as:
\begin{align}
    \text{LATE}^C_{DE} &= \mathbf{E}\Big[\frac{1}{n_j} \sum_{i=1}^{n_j} \big\{Y_{ji}\big(D_{ji}(1), \textbf{D}_{j(-i)}(0)\big) - Y_{ji}\big(D_{ji}(0), \textbf{D}_{j(-i)}(0)\big)\big\} \big| D_{ji}(1) =1, D_{ji}(0) = 0\Big]
    \nonumber\\&
    = \bigg\{ \rho\big(1,0\big) + \int_{\textbf{b}_{.}}\mathbf{E}_{\textbf{X},n}\Big[\frac{1}{n}\sum_{i=1}^{n} \int_{\textbf{d}_{.(-i)}}  \mathbf{E}\Big[  Y_{.i}\big| Z = 1, D_{.i} = 1, \textbf{D}_{.(-i)} = \textbf{d}_{.(-i)},   \textbf{X}_{.i} = \textbf{x}_{.i},n, b_{.}^{Y}\Big] \nonumber\\&   dF_{\textbf{D}_{.(-i)}| Z = 0, \textbf{X}_{.i},n, b_{.}^{D}}(\textbf{d}_{.(-i)})\Big] \times  dF_{\textbf{b}_{.}} \bigg\}
    \nonumber\\&\quad-
    \nonumber\\& \int_{\textbf{b}_{.}}\mathbf{E}_{\textbf{X},n}\Big[\frac{1}{n}\sum_{i=1}^{n} \int_{\textbf{d}_{.(-i)}} \mathbf{E}\Big[  Y_{.i}\big| Z = 0, D_{.i} = 0, \textbf{D}_{.(-i)} = \textbf{d}_{.(-i)},   \textbf{X}_{.i} = \textbf{x}_{.i},n, b_{.}^{Y}\Big] \nonumber\\&   dF_{\textbf{D}_{.(-i)}| Z = 0, \textbf{X}_{.i},n, b_{.}^{D}}(\textbf{d}_{.(-i)})\Big] \times  dF_{\textbf{b}_{.}}, \text{ and }
\end{align}

\begin{align}
    \text{LATE}^C_{SE} &= \mathbf{E}\Big[\frac{1}{n_j} \sum_{i=1}^{n_j} \big\{Y_{ji}\big(D_{ji}(1), \textbf{D}_{j(-i)}(1)\big) - Y_{ji}\big(D_{ji}(1), \textbf{D}_{j(-i)}(0)\big)\big\} \big| D_{ji}(1) =1, D_{ji}(0) = 0\Big]
    \nonumber\\&
    = \int_{\textbf{b}_{.}}\mathbf{E}_{\textbf{X},n}\Big[\frac{1}{n}\sum_{i=1}^{n} \int_{\textbf{d}_{.(-i)}} \mathbf{E}\Big[  Y_{.i}\big| Z = 1, D_{.i} = 1, \textbf{D}_{.(-i)} = \textbf{d}_{.(-i)},   \textbf{X}_{.i} = \textbf{x}_{.i},n, b_{.}^{Y}\Big]  \nonumber\\&  dF_{\textbf{D}_{.(-i)}| Z = 1, \textbf{X}_{.i},n, b_{.}^{D}}(\textbf{d}_{.(-i)})\Big] \times  dF_{\textbf{b}_{.}}
    \nonumber\\&\quad-
    \nonumber\\& \bigg\{ \rho\big(1,0\big) + \int_{\textbf{b}_{.}}\mathbf{E}_{\textbf{X},n}\Big[\frac{1}{n}\sum_{i=1}^{n} \int_{\textbf{d}_{.(-i)}}  \mathbf{E}\Big[  Y_{.i}\big| Z = 1, D_{.i} = 1, \textbf{D}_{.(-i)} = \textbf{d}_{.(-i)},   \textbf{X}_{.i} = \textbf{x}_{.i},n, b_{.}^{Y}\Big]  \nonumber\\&   dF_{\textbf{D}_{.(-i)}| Z = 0, \textbf{X}_{.i},n, b_{.}^{D}}(\textbf{d}_{.(-i)})\Big] \times  dF_{\textbf{b}_{.}} \bigg\}.
\end{align}

\subsection{Sensitivity analysis results}
Table~\ref{Table::SenAnalysisResults} compares the unadjusted and sensitivity-adjusted posterior summaries for all ITT and LATE effects. The posterior means of the adjusted direct effects remain negative and close to their unadjusted counterparts. The sign direction is preserved across all six effect estimates, indicating that no plausible violation of Assumption~\ref{Assump:6}, as calibrated by our sensitivity framework, reverses the qualitative conclusions. The adjusted credible intervals are substantially wider than the unadjusted intervals, which is expected as the sensitivity parameter $\rho(1,0)$ introduces additional uncertainty that reflects our ignorance about the magnitude and direction of the cross-world bias. For instance, the adjusted $95\%$ credible interval for the ITT direct effect widens from $(-0.272, 0.016)$ to $(-0.583, 0.350)$, now covering zero. This widening reflects the conservative nature of the analysis rather than evidence against the effect itself as the Beta$(2,1)$ prior on $\kappa$ places substantial mass on large departures from Assumption~\ref{Assump:6}. The total effects are unaffected by the sensitivity adjustment, as expected from the decomposition, because $\rho(1,0)$ enters the direct and spillover components with opposite signs and cancels in the sum.

Overall, the sensitivity analysis provides reassurance that the main findings from Section~\ref{Sec9} are not an artifact of the cross-world homogeneity assumption. The direction and approximate magnitude of the estimated effects persist under a range of plausible departures from Assumption~\ref{Assump:6}.

\begin{table}[htbp]
\caption{Sensitivity analysis for violations of Assumption~\ref{Assump:6}: unadjusted versus adjusted posterior summaries from the PSDP case study. The sensitivity parameter $\kappa \sim \text{Beta}(2,1)$ with random sign; calibration yields $R^2_1 = 0.184$ and $\text{Var}(Y \mid Z=1) = 0.198$.}
\label{Table::SenAnalysisResults}
\centering
\begin{tabular}{@{}llcccccc@{}}
\hline
 &  & \multicolumn{3}{c}{Unadjusted} & \multicolumn{3}{c}{Adjusted} \\
\cmidrule(lr){3-5} \cmidrule(lr){6-8}
 &  & Mean & $2.5\%$ & $97.5\%$ & Mean & $2.5\%$ & $97.5\%$ \\
\hline
\textit{ITT} & Direct  & $-0.123$ & $-0.272$ & $\phantom{-}0.016$ & $-0.138$ & $-0.583$ & $0.350$ \\
          & Spillover  & $-0.027$ & $-0.160$ & $0.098$ & $-0.013$ & $-0.483$ & $0.441$ \\
          & Total   & $-0.150$ & $-0.277$ & $-0.029$ & $-0.150$ & $-0.277$ & $-0.029$ \\[4pt]
\textit{LATE} & Direct  & $-0.175$ & $-0.320$ & $-0.034$ & $-0.193$ & $-0.652$ & $0.292$ \\
          & Spillover  & $-0.027$ & $-0.161$ & $\phantom{-}0.092$ & $-0.009$ & $-0.477$ & $0.428$ \\
          & Total   & $-0.202$ & $-0.324$ & $-0.080$ & $-0.202$ & $-0.324$ & $-0.080$ \\
\hline
\end{tabular}
\end{table}

\section{Sensitivity analysis for the violation of Assumption 7}\label{AppendixE}

Assumption~\ref{Assump:7} requires that, conditional on baseline covariates and cluster size, the potential treatment receipt of individual $i$ under assignment $z$ is independent of the potential treatment receipt of a different individual $i'$ in the same cluster under a different assignment $z_*$. This cross-world independence is inherently untestable from the observed data, since the treatment receipt vectors $\textbf{D}_j(1)$ and $\textbf{D}_j(0)$ are never jointly observed for the same cluster. In this section, we develop a sensitivity analysis to examine how the estimated causal effects change as cross-world dependence is introduced. The approach adapts the Gaussian copula sensitivity framework of \citet{ohnishi2025identification} to the binary treatment receipt setting. Note that this sensitivity analysis is complementary to the one in Appendix~\ref{AppendixD}: Assumption~\ref{Assump:6} concerns cross-world exchangeability in the outcome model, whereas Assumption~\ref{Assump:7} concerns cross-world independence in the treatment receipt model. Violations of these two assumptions affect the causal estimates through distinct channels, and we conduct separate sensitivity analyses for each.

\subsection{Sensitivity parameter on the latent probit scale}

The treatment receipt model in Section~\ref{Sec5} specifies $D_{ji} = \mathbf{1}(U^D_{ji} > 0)$ with $U^D_{ji} = f^D(Z_j, \mathbf{X}_{ji}, n_j) + b_j^D + e_{ji}^D$ and $e_{ji}^D \sim \mathcal{N}(0,1)$. In the G-computation algorithm, counterfactual treatment receipt draws are generated conditional on $b_j^D$ (sampled from its posterior), so Assumption~\ref{Assump:7} is operationally imposed conditional on $(\mathbf{X}_{ji}, n_j, b_j^D)$. Note that the shared cluster random effect $b_j^D$ already induces within-cluster dependence in treatment receipt at the marginal level (integrating over $b_j^D$). The sensitivity analysis developed here targets residual cross-world dependence \emph{beyond} what is captured by $b_j^D$. Formally, we relax the conditional independence of $e_{ji}^D(z)$ and $e_{ji'}^D(z_*)$ for $i \neq i'$ and $z \neq z_*$, given $(\mathbf{X}_{ji}, n_j, b_j^D)$, by introducing a sensitivity parameter $\rho_7 \in (-\rho_{\max}, \rho_{\max})$ such that
\[
\text{Corr}\big(e_{ji}^D(z),\; e_{ji'}^D(z_*) \mid \mathbf{X}_{ji}, n_j, b_j^D\big) = \rho_7, \qquad i \neq i',\; z \neq z_*.
\]

When $\rho_7 = 0$, Assumption~\ref{Assump:7} holds conditional on $b_j^D$. When $\rho_7 > 0$, unobserved factors beyond the cluster random effect, such as community health attitudes or peer pressure dynamics not captured by $b_j^D$, induce positive cross-world dependence in treatment receipt across individuals within a cluster. Although $\rho_7$ is mathematically defined on $(-\rho_{\max}, \rho_{\max})$, in the PSDP trial context we expect $\rho_7$ to be non-negative, since the plausible unobserved mechanisms (shared community health attitudes, peer pressure, common exposure to outreach) all act in the same direction across individuals within a cluster.

\subsection{Correlation structure and positive-definiteness constraint}

 \citet{ohnishi2025identification} model the full joint distribution of cross-world intermediate variables within a cluster through a $2n_j \times 2n_j$ Gaussian copula correlation matrix with a $2 \times 2$ block structure,
\[
\boldsymbol{\Omega} = \begin{pmatrix} \mathbf{C}_{11} & \mathbf{C}_{10} \\ \mathbf{C}_{10}^{\top} & \mathbf{C}_{00} \end{pmatrix},
\]
where $\mathbf{C}_{11}$ and $\mathbf{C}_{00}$ are $n_j \times n_j$ equicorrelation matrices governing within-world dependence under treatment ($\gamma_1$) and control ($\gamma_0$), respectively, and the cross-world block $\mathbf{C}_{10}$ has diagonal entries $\rho$ (same individual, different worlds) and off-diagonal entries $\rho^{*} = (\gamma_0 + \gamma_1)\rho/2$ (different individuals, different worlds). In their setting, $\gamma_0$ and $\gamma_1$ are identifiable from observed data because within-cluster correlations of the intermediate variables can be estimated from each treatment arm separately.

In our framework, the G-computation algorithm draws counterfactual treatment receipt vectors \emph{conditional} on the cluster random effect $b_j^D$, which is sampled from its posterior at each MCMC iteration. Conditioning on $b_j^D$ absorbs the shared within-cluster component that drives same-world correlations, so that the treatment receipt model in Section~\ref{Sec5} implies conditional independence: $D_{ji}(z)$ and $D_{ji'}(z)$ are independent given $(Z_j, \mathbf{X}_j, n_j, b_j^D)$ for $i \neq i'$. On the latent probit scale, this corresponds to $\gamma_1 = \gamma_0 = 0$, and the diagonal blocks reduce to identity matrices: $\mathbf{C}_{11} = \mathbf{C}_{00} = \mathbf{I}_{n_j}$. As for the cross-world block $\mathbf{C}_{10}$, the same-individual cross-world correlation (i.e., $\text{Corr}(e_{ji}^D(1), e_{ji}^D(0))$, represented by the diagonal entries of $\mathbf{C}_{10}$) does not directly govern Assumption~\ref{Assump:7}, which concerns \emph{different} individuals. Our sensitivity parameter $\rho_7$ instead targets the off-diagonal entries of $\mathbf{C}_{10}$: the cross-world correlation $\text{Corr}(e_{ji}^D(z), e_{ji'}^D(z_*))$ for $i \neq i'$, which is precisely the quantity that Assumption~\ref{Assump:7} sets to zero. We treat the same-individual cross-world correlation as a nuisance that does not affect the G-computation for the estimands of interest, since the focal individual's treatment receipt is drawn from its own marginal (ITT) or fixed deterministically (LATE).

The G-computation for $\Theta_{(1,0)}$ requires, for each focal individual $i$, a joint draw of $D_{ji}(1)$ (the focal individual's treatment receipt under $z = 1$) and $\{D_{ji'}(0)\}_{i' \neq i}$ (all peers' treatment receipts under $z_* = 0$). This is a subvector of length $n_j$ extracted from the full $2n_j$-dimensional latent vector. On the latent probit scale, define the $n_j$-dimensional error vector
\[
\mathbf{e}^{(i)}_j = \big(e_{ji}^D(1),\; e_{j,i'_1}^D(0), \ldots, e_{j,i'_{n_j-1}}^D(0)\big)^{\top},
\]
where the first component is the focal individual's error under $z = 1$ and the remaining $n_j - 1$ components are the peers' errors under $z_* = 0$. Under the simplified copula ($\gamma_0 = \gamma_1 = 0$, $\rho^* = 0$), the correlation matrix of $\mathbf{e}^{(i)}_j$ is
\[
\boldsymbol{\Sigma}_j(\rho_7) =
\begin{pmatrix}
1 & \rho_7 & \rho_7 & \cdots & \rho_7 \\
\rho_7 & 1 & 0 & \cdots & 0 \\
\rho_7 & 0 & 1 & \cdots & 0 \\
\vdots & \vdots & \vdots & \ddots & \vdots \\
\rho_7 & 0 & 0 & \cdots & 1
\end{pmatrix},
\]
where the $(1,k)$ and $(k,1)$ entries equal $\rho_7$ for $k = 2, \ldots, n_j$, and all remaining off-diagonal entries are zero. The structure reflects three features: (i) each peer's error under $z_* = 0$ is correlated with the focal individual's error under $z = 1$ at strength $\rho_7$ (the parameter governed by Assumption~\ref{Assump:7}); (ii) peers' same-world errors are mutually uncorrelated conditional on $b_j^D$ (from the model specification); and (iii) the common-$\rho_7$ structure across all focal--peer pairs imposes exchangeability, yielding a parsimonious single-parameter sensitivity analysis. We note that, by symmetry of the relabeling, $\boldsymbol{\Sigma}_j(\rho_7)$ has the same structure regardless of which individual serves as the focal unit.

 The matrix $\boldsymbol{\Sigma}_j(\rho_7)$ can be written as a rank-$2$ perturbation of the identity:
\[
\boldsymbol{\Sigma}_j(\rho_7) = \mathbf{I}_{n_j} + \rho_7 \big(\mathbf{e}_1 \mathbf{v}^{\top} + \mathbf{v}\mathbf{e}_1^{\top}\big),
\]
where $\mathbf{e}_1 = (1, 0, \ldots, 0)^{\top} \in \mathbb{R}^{n_j}$ and $\mathbf{v} = (0, 1, 1, \ldots, 1)^{\top} \in \mathbb{R}^{n_j}$, with $\|\mathbf{v}\| = \sqrt{n_j - 1}$. The perturbation matrix $\mathbf{B} = \mathbf{e}_1 \mathbf{v}^{\top} + \mathbf{v}\mathbf{e}_1^{\top}$ has rank $2$, since its column space is $\operatorname{span}\{\mathbf{e}_1, \mathbf{v}\}$. Any eigenvector of $\boldsymbol{\Sigma}_j$ orthogonal to both $\mathbf{e}_1$ and $\mathbf{v}$ is an eigenvector of $\mathbf{I}_{n_j}$ with eigenvalue $1$. There are $n_j - 2$ such linearly independent eigenvectors. The remaining two eigenvalues are found by restricting to $\operatorname{span}\{\mathbf{e}_1, \mathbf{v}\}$. In this two-dimensional subspace, $\mathbf{B}$ has the representation
\[
\mathbf{B}\big|_{\operatorname{span}\{\mathbf{e}_1, \mathbf{v}\}} = \begin{pmatrix} 0 & \|\mathbf{v}\| \\ \|\mathbf{v}\| & 0 \end{pmatrix},
\]
whose eigenvalues are $\pm \|\mathbf{v}\| = \pm \sqrt{n_j - 1}$. Therefore, the eigenvalues of $\boldsymbol{\Sigma}_j(\rho_7) = \mathbf{I} + \rho_7 \mathbf{B}$ are
\[
\lambda_1 = 1 + \rho_7\sqrt{n_j - 1}, \qquad \lambda_2 = 1 - \rho_7\sqrt{n_j - 1}, \qquad \lambda_3 = \cdots = \lambda_{n_j} = 1.
\]
Positive definiteness requires all eigenvalues to be strictly positive. The binding constraint is $\lambda_2 > 0$ when $\rho_7 > 0$ (or $\lambda_1 > 0$ when $\rho_7 < 0$), giving
\begin{equation}\label{Eq::A7_PD_constraint}
|\rho_7| < \frac{1}{\sqrt{n_j - 1}} \eqqcolon \rho_{\max}(n_j).
\end{equation}
For the PSDP dataset, the largest cluster has $n_j = 80$ students, yielding $\rho_{\max}(80) = 1/\sqrt{79} \approx 0.112$. When a common $\rho_7$ is used across all clusters, the binding constraint is determined by the largest cluster. Smaller clusters admit larger values of $|\rho_7|$. For instance, a cluster of size $n_j = 25$ gives $\rho_{\max}(25) \approx 0.204$. This constraint has an important practical implication: in settings with large clusters, even modest departures from Assumption~\ref{Assump:7} on the latent scale are limited by the positive-definiteness requirement, so the feasible range of $\rho_7$ is inherently narrow. This provides a structural form of robustness, as the degree to which Assumption~\ref{Assump:7} can be violated is bounded by the cluster size.

\subsection{Modified G-computation under the relaxed assumption}

The sensitivity analysis modifies only the G-computation for $\Theta_{(1,0)}^{ITT}$ and $\Theta_{(1,0)}^{LATE}$, since these are the only estimand components that involve cross-world treatment receipt draws. The estimands $\Theta_{(1,1)}$ and $\Theta_{(0,0)}$ are unaffected because all treatment receipt draws occur under the same assignment. Consequently, the total effects $\text{ITT}^C$ and $\text{LATE}^C$ are invariant to Assumption~\ref{Assump:7}, and only the direct and spillover decompositions are sensitive.

For a given value of $\rho_7$, the modified G-computation for the ITT estimand $\Theta_{(1,0)}^{ITT}$ proceeds as follows. At each retained MCMC iteration $q$, for each cluster $j$ and focal individual $i$:

\begin{enumerate}
\item Compute marginal treatment probabilities from the stored BART posterior:
\[
p_{ji}(1) = \Phi\big(f^{D,(q)}(1, \mathbf{X}_{ji}, n_j) + b_j^{D,(q)}\big), \qquad p_{ji'}(0) = \Phi\big(f^{D,(q)}(0, \mathbf{X}_{ji'}, n_j) + b_j^{D,(q)}\big).
\]
\item Compute the Cholesky factor $\mathbf{L}_j$ of $\boldsymbol{\Sigma}_j(\rho_7)$, which needs to be computed only once per unique cluster size $n_j$ and value of $\rho_7$.
\item For each Monte Carlo replicate $m = 1, \ldots, M$:
\begin{enumerate}
\item[(a)] Draw $\boldsymbol{\eta}^{(m)} \sim \mathcal{N}(\mathbf{0}, \mathbf{I}_{n_j})$ and compute $\mathbf{W}^{(m)} = \mathbf{L}_j \boldsymbol{\eta}^{(m)}$.
\item[(b)] Set $D_{ji}^{(q,m)}(1) = \mathbf{1}\big\{\Phi(W_1^{(m)}) < p_{ji}(1)\big\}$ and $D_{ji'_k}^{(q,m)}(0) = \mathbf{1}\big\{\Phi(W_{k+1}^{(m)}) < p_{ji'_k}(0)\big\}$ for $k = 1, \ldots, n_j - 1$.
\item[(c)] Compute $\bar{D}_{j(-i)}^{(q,m)} = \frac{1}{n_j - 1}\sum_{k=1}^{n_j-1} D_{ji'_k}^{(q,m)}(0)$.
\item[(d)] Predict the outcome probability: $\widehat{y}_{ji}^{(q,m)} = \Phi\big(f^{Y,(q)}(1, D_{ji}^{(q,m)}, \bar{D}_{j(-i)}^{(q,m)}, \mathbf{X}_{ji}, n_j) + b_j^{Y,(q)}\big)$.
\end{enumerate}
\item Compute the individual-level average $\widehat{\Phi}_{ji}^{(q)} = \frac{1}{M}\sum_{m=1}^M \widehat{y}_{ji}^{(q,m)}$ and aggregate to the cluster level.
\end{enumerate}

For the LATE estimand $\Theta_{(1,0)}^{LATE}$, the algorithm is identical except that in step 3(b), the focal individual's treatment receipt is fixed at $D_{ji}^{(q,m)}(1) = 1$ (the complier stratum value) rather than being drawn stochastically. The peers' treatment receipt draws $D_{ji'_k}^{(q,m)}(0)$ are generated from the $(n_j - 1)$-dimensional marginal of $\mathbf{W}^{(m)}$ conditional on $W_1^{(m)}$ exceeding the complier threshold. In practice, this conditioning can be implemented by rejection sampling: retain only those Monte Carlo replicates for which $\Phi(W_1^{(m)}) < p_{ji}(1)$ (i.e., $D_{ji}^{(q,m)}(1) = 1$), and discard the rest.

The sensitivity-adjusted $\widehat{\Theta}_{(1,0)}^{SA,(q)}(\rho_7)$ is then computed as a Bayesian Bootstrap--weighted average over clusters, exactly as in Algorithms~\ref{Algo1} and~\ref{Algo2}. The adjusted direct and spillover effects are
\[
\text{ITT}^{C,SA}_{DE}(\rho_7) = \widehat{\Theta}_{(1,0)}^{SA}(\rho_7) - \widehat{\Theta}_{(0,0)}, \qquad \text{ITT}^{C,SA}_{SE}(\rho_7) = \widehat{\Theta}_{(1,1)} - \widehat{\Theta}_{(1,0)}^{SA}(\rho_7),
\]
with analogous expressions for the LATE effects. The total effects remain unchanged.

\subsection{Sensitivity analysis results}

We evaluate the sensitivity of the PSDP causal effect estimates to violations of Assumption~\ref{Assump:7} over a grid of four values: $\rho_7 \in \{0.025, 0.05, 0.075, 0.10\}$, spanning nearly the entire feasible range ($\rho_{\max}(80) \approx 0.1125$). The sensitivity analysis is computed using the same BART-RE with clever covariate tree ensemble and the same Bayesian Bootstrap weights as the main G-computation at each retained iteration.

Figure~\ref{Fig::SA_A7} displays the sensitivity-adjusted posterior means and $95\%$ credible intervals alongside the unadjusted estimates. The results exhibit near-complete insensitivity to the cross-world correlation parameter. Specifically, the maximum absolute deviation of any sensitivity-adjusted posterior mean from its unadjusted counterpart across the entire $\rho_7$ grid is less than $2 \times 10^{-6}$ on the probability scale. No sign direction changes occur at any value of $\rho_7$, and the posterior means show no visible trend as $\rho_7$ increases from $0.025$ to $0.10$. This insensitivity has a structural explanation. The positive-definiteness constraint~\eqref{Eq::A7_PD_constraint} restricts $|\rho_7| < 1/\sqrt{n_j - 1}$, which for the largest cluster in the PSDP ($n_j = 80$) limits the cross-world correlation to a narrow band of width approximately $0.1125$. Within this band, the latent correlation between the focal individual's treatment receipt error under $z = 1$ and a peer's error under $z_* = 0$ is too weak to materially alter the distribution of peer treatment receipt draws, and therefore has negligible impact on the G-computation integrals. In settings with large clusters, this constraint provides an inherent structural safeguard against violations of Assumption~\ref{Assump:7}.

\begin{figure}[htbp]
\centering
\includegraphics[width=0.95\textwidth]{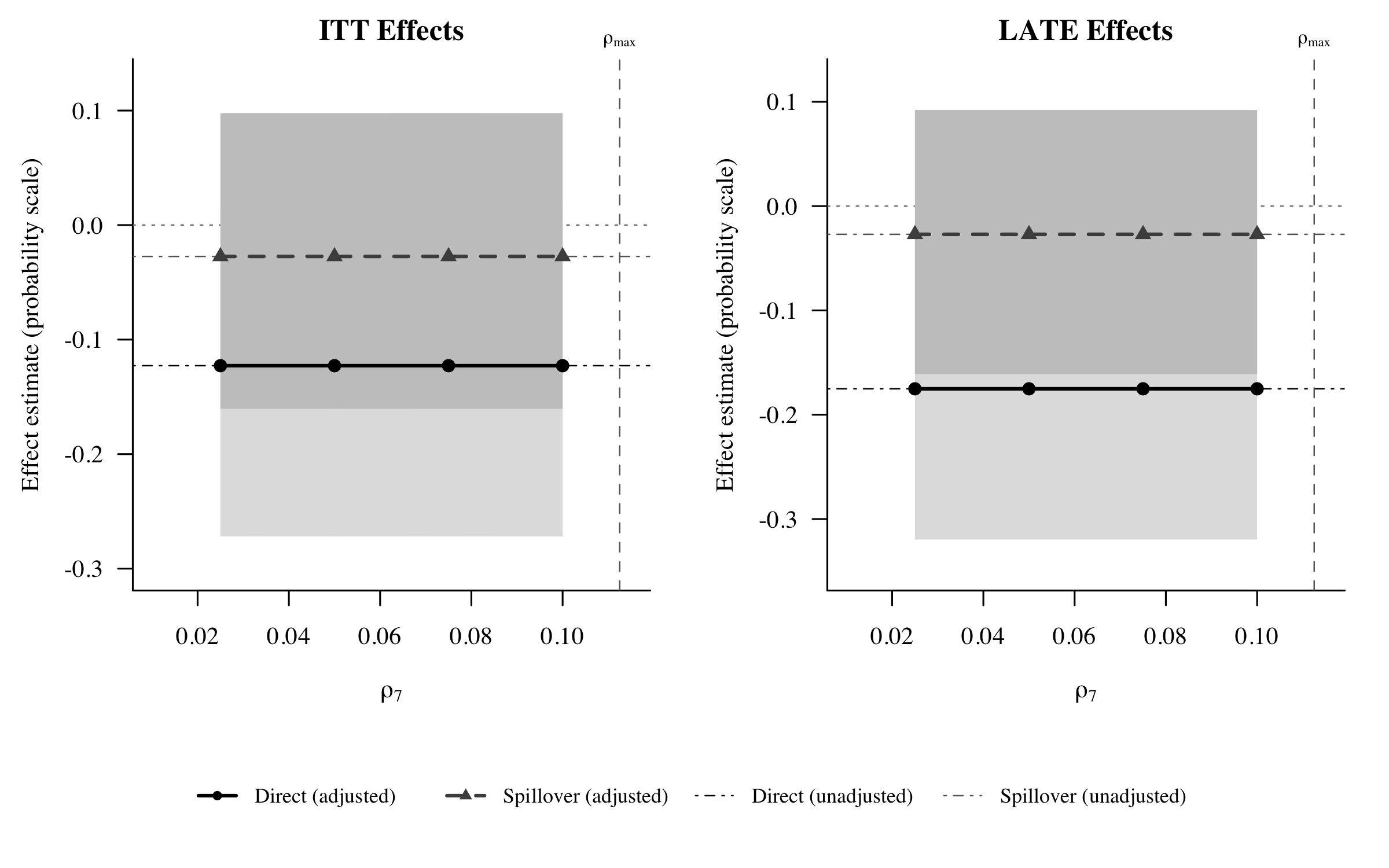}
\caption{Sensitivity of direct and spillover effect estimates to the cross-world correlation parameter $\rho_7$ (Assumption~\ref{Assump:7}), computed with $M = 10{,}000$ Monte Carlo draws. Left panel: ITT effects. Right panel: LATE effects. Solid lines with filled circles denote the direct effect; dashed lines with triangles denote the spillover effect. Shaded bands represent $95\%$ credible intervals. The dot-dash horizontal lines indicate the unadjusted estimates from the main analysis. The vertical dashed line marks $\rho_{\max} \approx 0.112$, the positive-definiteness bound for the largest cluster. The adjusted and unadjusted estimates are visually indistinguishable in this implementation. With structural one-sided noncompliance, control-world receipt is constant, so this stability does not assess general violations of receipt independence.}
\label{Fig::SA_A7}
\end{figure}

\section{Details on simulation study}\label{AppendixF}

This appendix provides further details on the simulation designs and estimation procedures. Appendix~\ref{AppendixFa} gives the parametric comparison, Appendix~\ref{AppendixFb} gives the paired clever-covariate study and diagnostics of the effect decomposition, and Appendix~\ref{AppendixFc} examines strong spillover with smaller clusters.

\subsection{Comparison with a parametric probit model}\label{AppendixFa}

\subsubsection{Data generating mechanism}\label{AppendixF1}

We first fit the following parametric probit model to the PSDP analysis sample:
\begin{equation}\label{Sim::ParamModel}
    \begin{aligned}
    Y_{ji}\mid D_{ji},\bar{D}_{j(-i)},Z_j,\mathbf{X}_{ji},b_j^Y
    &\sim \mathrm{Bern}\left\{
    \Phi\left(\gamma_0+\gamma_DD_{ji}
    +\gamma_{\bar D}\bar D_{j(-i)}
    +\gamma_ZZ_j+\boldsymbol{\gamma}_X^\top\mathbf{X}_{ji}
    +b_j^Y\right)\right\},\\
    D_{ji}\mid Z_j,\mathbf{X}_{ji},b_j^D
    &\sim \mathrm{Bern}\left\{
    \Phi\left(\xi_0+\xi_ZZ_j
    +\boldsymbol{\xi}_X^\top\mathbf{X}_{ji}+b_j^D\right)\right\},\\
    b_j^Y&\sim\mathcal{N}(0,\tau_Y^2),\qquad
    b_j^D\sim\mathcal{N}(0,\tau_D^2),
    \end{aligned}
\end{equation}
where $\mathbf{X}_{ji}=(X_{1,ji},X_{2,ji},X_{3,j})^\top$, $\Phi(\cdot)$ is the standard normal CDF, and the cluster-level random effects capture residual between-cluster heterogeneity.  We consider three continuous baseline covariates: age ($X_1$), weight-for-age $z$-score ($X_2$), and the zone baseline infection rate ($X_3$). The posterior mean estimates from this fitted model provide the parameter values for the DGM and, together with Algorithms~\ref{Algo1} and~\ref{Algo2}, define the corresponding ground-truth causal effects. We then consider both this correctly specified linear DGM and a misspecified nonlinear DGM that violates the functional-form assumptions of the parametric model. For each DGM, we evaluate recovery of the direct, spillover, and total effects at two sample sizes.

The fitted outcome-model variance is $\tau_Y^2=0.244$ (95\% credible interval $[0.157,0.367]$), corresponding to a latent-scale ICC of approximately $0.20$. Similarly, the fitted treatment-receipt variance is $\tau_D^2=0.254$, yielding a similar ICC.

The parametric probit model in \eqref{Sim::ParamModel} is fitted using \texttt{Nimble} \citep{de2017programming}, running four MCMC chains for 7,500 iterations each, discarding the first 5,000 iterations as burn-in, and retaining every tenth draw. This yields 1,000 posterior draws across chains for inference. We fit BART-RE using 50 trees for each of the outcome and treatment-receipt models, with 5,000 burn-in iterations and 1,000 retained posterior draws after thinning.

At each retained posterior draw, we apply Algorithms~\ref{Algo1} and~\ref{Algo2} using $M=10{,}000$ Monte Carlo draws for treatment-receipt integration. For each replicate, we obtain posterior means and 95\% credible intervals for the cluster-level ITT and LATE direct, spillover, and total effects. We evaluate performance using bias, root mean squared error (RMSE), and 95\% coverage.

We consider two simulation settings. In the first, each replicated dataset contains $J=50$ clusters with a total sample size of $N=1{,}500$, corresponding to an average cluster size of 30, with 25 clusters assigned to treatment and 25 to control. In the second, we increase to $J=75$ clusters and $N=2{,}250$, maintaining the same average cluster size of 30, with 37 or 38 clusters assigned to each arm. Cluster sizes are sampled from a discrete uniform distribution centered at $N/J$ with a range of $\pm5$ and then rescaled to achieve the target sample size while maintaining a minimum cluster size of two. Individual-level covariates are sampled with replacement from their empirical PSDP distributions, and cluster-level infection rates are sampled from the corresponding empirical distribution. We generate cluster random effects from their fitted normal distributions. Finally, we generate treatment receipt and the binary outcome from \eqref{Sim::ParamModel}.

We made two modifications to the PSDP-calibrated treatment-receipt model to ensure sufficient overlap for identifying separate direct and spillover effects. First, the original one-sided noncompliance produces almost no treatment receipt under control and therefore leaves little overlap in the peer-treatment proportion. To address this, we set $\xi_0=0$, $\xi_Z=1.0$, and $\tau_D^2=1.0$, which yields approximately
\[
P(D=1\mid Z=0)\approx0.22,\qquad
P(D=1\mid Z=1)\approx0.51,
\]
and reduces the correlation between $Z_j$ and $\bar D_{j(-i)}$ to approximately $0.34$--$0.41$. Second, because individual treatment receipt and peer treatment proportion remain correlated through the shared cluster random effect, we introduce moderate individual-level heterogeneity in compliance:
\[
e_{ji}\overset{\mathrm{i.i.d.}}{\sim}\mathcal{N}(0,1),
\]
so that the treatment-receipt model becomes
\[
D_{ji}\mid Z_j,\mathbf{X}_{ji},b_j^D,e_{ji}
\sim
\mathrm{Bern}\left\{
\Phi\left(\xi_0+\xi_ZZ_j+
\boldsymbol{\xi}_X^\top\mathbf{X}_{ji}+b_j^D+e_{ji}\right)
\right\}.
\]
Because the additional term is Gaussian, the marginal treatment-receipt model remains probit after integrating over $e_{ji}$, with a corresponding rescaling of the regression coefficients. The same treatment-receipt mechanism is used when computing the ground-truth effects.

To assess robustness to functional-form misspecification, we construct a nonlinear DGM by augmenting both the outcome and treatment-receipt linear predictors with
\[
\gamma_{X_1^2}\tilde X_{1,ji}^2+
\gamma_{\log X_2}\log(|\tilde X_{2,ji}|+1)+
\gamma_{X_1\times X_2}\tilde X_{1,ji}\tilde X_{2,ji},
\]
where $\tilde X_k$ denotes the standardized version of $X_k$ based on the observed PSDP mean and standard deviation. The corresponding coefficients are $(1.5,0.75,-2.5)$ for both models. The outcome model additionally includes the interaction
\[
\gamma_{D\times X_1}D_{ji}\tilde X_{1,ji},
\qquad \gamma_{D\times X_1}=1.75,
\]
allowing the treatment effect to vary with age. The parametric probit model in \eqref{Sim::ParamModel} omits these nonlinear terms and the treatment--covariate interaction and is therefore misspecified under this DGM. Finally, for both models, we use an $\mathrm{InvGamma}(3,0.6)$ prior for $\tau_Y^2$ and an $\mathrm{InvGamma}(3,2)$ prior for $\tau_D^2$.

\subsubsection{Additional results}\label{AppendixF2}

Tables~\ref{TabSuppS1} and~\ref{TabSuppS2} report performance for potential-outcome means using the estimator with the clever covariate. Tables~\ref{TabSuppS3} and~\ref{TabSuppS4} report average credible interval widths for the causal contrasts.

\begin{table}[htbp]
\caption{Simulation results for potential-outcome means under the correctly specified  DGM, with the clever covariate: bias, RMSE, and $95\%$ credible interval coverage across $500$ replicated datasets. Potential outcomes are $\Theta_{(1,1)}$ (all clusters assigned to treatment), $\Theta_{(0,0)}$ (all clusters assigned to control), and $\Theta_{(1,0)}$ (individual assigned to treatment, peers assigned to control).}
\label{TabSuppS1}
\scalebox{0.88}{%
\begin{tabular}{@{}llcccccccccc@{}}
\hline
Estimand & $N$ & $J$ & Potential outcome & Ground Truth & \multicolumn{3}{c}{Parametric} & \multicolumn{3}{c}{BART-RE} \\
\cline{6-8}\cline{9-11}
 & & & & & Bias & Cov. & RMSE & Bias & Cov. & RMSE \\
\hline
       & 1500 & 50 & $\Theta_{(1,1)}$ & $0.49$ & $-0.045$ & $96.8\%$ & $0.082$ & $\phantom{-}0.0052$ & $97.8\%$ & $0.045$ \\
       &      &    & $\Theta_{(0,0)}$ & $0.40$ & $-0.038$ & $95.6\%$ & $0.070$ & $-0.0016$ & $98.6\%$ & $0.043$ \\
ITT    &      &    & $\Theta_{(1,0)}$ & $0.55$ & $-0.033$ & $97.6\%$ & $0.073$ & $-0.086$ & $96.8\%$ & $0.096$ \\
\cline{2-11}
       & 2250 & 75 & $\Theta_{(1,1)}$ & $0.49$ & $-0.044$ & $92.4\%$ & $0.073$ & $\phantom{-}0.0044$ & $99.4\%$ & $0.036$ \\
       &      &    & $\Theta_{(0,0)}$ & $0.40$ & $-0.033$ & $93.4\%$ & $0.062$ & $-0.0020$ & $98.6\%$ & $0.037$ \\
       &      &    & $\Theta_{(1,0)}$ & $0.55$ & $-0.027$ & $96.2\%$ & $0.062$ & $-0.082$ & $97.2\%$ & $0.090$ \\
\hline
       & 1500 & 50 & $\Theta_{(1,1)}$ & $0.38$ & $-0.034$ & $95.8\%$ & $0.067$ & $-0.0008$ & $97.6\%$ & $0.047$ \\
       &      &    & $\Theta_{(0,0)}$ & $0.44$ & $-0.027$ & $96.6\%$ & $0.059$ & $\phantom{-}0.0006$ & $98.8\%$ & $0.043$ \\
LATE   &      &    & $\Theta_{(1,0)}$ & $0.44$ & $-0.026$ & $97.0\%$ & $0.064$ & $-0.092$ & $94.2\%$ & $0.10$ \\
\cline{2-11}
       & 2250 & 75 & $\Theta_{(1,1)}$ & $0.38$ & $-0.033$ & $94.0\%$ & $0.058$ & $-0.0020$ & $98.4\%$ & $0.036$ \\
       &      &    & $\Theta_{(0,0)}$ & $0.44$ & $-0.023$ & $94.8\%$ & $0.053$ & $-0.0004$ & $98.0\%$ & $0.037$ \\
       &      &    & $\Theta_{(1,0)}$ & $0.44$ & $-0.021$ & $96.6\%$ & $0.053$ & $-0.089$ & $94.4\%$ & $0.097$ \\
\hline
\end{tabular}}
\end{table}

\begin{table}[htbp]
\caption{Simulation results for potential-outcome means under the misspecified  DGM, with the clever covariate: bias, RMSE, and $95\%$ credible interval coverage across $500$ replicated datasets.}
\label{TabSuppS2}
\scalebox{0.88}{%
\begin{tabular}{@{}llcccccccccc@{}}
\hline
Estimand & $N$ & $J$ & Potential outcome & Ground Truth & \multicolumn{3}{c}{Parametric} & \multicolumn{3}{c}{BART-RE} \\
\cline{6-8}\cline{9-11}
 & & & & & Bias & Cov. & RMSE & Bias & Cov. & RMSE \\
\hline
       & 1500 & 50 & $\Theta_{(1,1)}$ & $0.65$ & $\phantom{-}0.023$ & $97.8\%$ & $0.041$ & $-0.0017$ & $98.0\%$ & $0.028$ \\
       &      &    & $\Theta_{(0,0)}$ & $0.59$ & $\phantom{-}0.050$ & $91.2\%$ & $0.065$ & $-0.0057$ & $98.2\%$ & $0.029$ \\
ITT    &      &    & $\Theta_{(1,0)}$ & $0.68$ & $\phantom{-}0.100$ & $48.0\%$ & $0.11$ & $-0.030$ & $99.4\%$ & $0.061$ \\
\cline{2-11}
       & 2250 & 75 & $\Theta_{(1,1)}$ & $0.65$ & $\phantom{-}0.024$ & $95.4\%$ & $0.037$ & $-0.0024$ & $99.2\%$ & $0.022$ \\
       &      &    & $\Theta_{(0,0)}$ & $0.59$ & $\phantom{-}0.054$ & $78.7\%$ & $0.063$ & $-0.0048$ & $98.2\%$ & $0.024$ \\
       &      &    & $\Theta_{(1,0)}$ & $0.68$ & $\phantom{-}0.103$ & $22.1\%$ & $0.11$ & $-0.031$ & $99.4\%$ & $0.058$ \\
\hline
       & 1500 & 50 & $\Theta_{(1,1)}$ & $0.61$ & $\phantom{-}0.156$ & $12.0\%$ & $0.16$ & $\phantom{-}0.011$ & $97.6\%$ & $0.031$ \\
       &      &    & $\Theta_{(0,0)}$ & $0.66$ & $-0.085$ & $79.8\%$ & $0.10$ & $-0.028$ & $94.0\%$ & $0.042$ \\
LATE   &      &    & $\Theta_{(1,0)}$ & $0.64$ & $\phantom{-}0.216$ & $\phantom{0}0.2\%$ & $0.22$ & $-0.012$ & $99.8\%$ & $0.053$ \\
\cline{2-11}
       & 2250 & 75 & $\Theta_{(1,1)}$ & $0.61$ & $\phantom{-}0.159$ & $\phantom{0}1.4\%$ & $0.16$ & $\phantom{-}0.0084$ & $97.2\%$ & $0.023$ \\
       &      &    & $\Theta_{(0,0)}$ & $0.66$ & $-0.079$ & $69.4\%$ & $0.092$ & $-0.020$ & $93.8\%$ & $0.033$ \\
       &      &    & $\Theta_{(1,0)}$ & $0.64$ & $\phantom{-}0.221$ & $\phantom{0}0.0\%$ & $0.22$ & $-0.016$ & $99.8\%$ & $0.050$ \\
\hline
\end{tabular}}
\end{table}

\begin{table}[htbp]
\caption{Average $95\%$ credible interval widths for causal effect estimates under the correctly specified DGM across $500$ replicated datasets.}
\label{TabSuppS3}
\scalebox{0.9}{%
\begin{tabular}{@{}lcccccc@{}}
\hline
Effect type & $N$ & $J$ & Effects & \multicolumn{2}{c}{Avg.\ CI Width} \\
\cline{5-6}
 &  &  &  & Parametric & BART-RE \\
\hline
      & 1500 & 50 & Direct   & $0.209$ & $0.319$ \\
          &  &  & Spillover & $0.106$ & $0.296$ \\
cluster-ITT    &  &  & Total    & $0.198$ & $0.216$ \\
\cline{2-6}
effects    & 2250 & 75 & Direct   & $0.171$ & $0.307$ \\
            &  &  & Spillover  & $0.088$ & $0.295$ \\
            &  &  & Total     & $0.162$ & $0.179$ \\
\hline
     & 1500 & 50 & Direct   & $0.236$ & $0.331$ \\
             &  &  & Spillover  & $0.105$ & $0.295$ \\
cluster-LATE    &  &  & Total   & $0.212$ & $0.227$ \\
\cline{2-6}
effects   & 2250 & 75 & Direct   & $0.193$ & $0.315$ \\
                   &  &  & Spillover  & $0.087$ & $0.293$ \\
                   &  &  & Total   & $0.173$ & $0.188$ \\
\hline
\end{tabular}}
\end{table}

\begin{table}[htbp]
\caption{Average $95\%$ credible interval widths for causal effect estimates under the misspecified DGM across $500$ replicated datasets.}
\label{TabSuppS4}
\scalebox{0.9}{%
\begin{tabular}{@{}lcccccc@{}}
\hline
Effect type & $N$ & $J$ & Effects & \multicolumn{2}{c}{Avg.\ CI Width} \\
\cline{5-6}
 &  &  &  & Parametric & BART-RE \\
\hline
      & 1500 & 50 & Direct   & $0.148$ & $0.435$ \\
          &  &  & Spillover & $0.096$ & $0.423$ \\
cluster-ITT    &  &  & Total    & $0.159$ & $0.151$ \\
\cline{2-6}
effects    & 2250 & 75 & Direct   & $0.117$ & $0.398$ \\
            &  &  & Spillover  & $0.075$ & $0.389$ \\
            &  &  & Total     & $0.125$ & $0.118$ \\
\hline
     & 1500 & 50 & Direct   & $0.183$ & $0.436$ \\
             &  &  & Spillover  & $0.089$ & $0.412$ \\
cluster-LATE    &  &  & Total   & $0.187$ & $0.172$ \\
\cline{2-6}
effects   & 2250 & 75 & Direct   & $0.145$ & $0.394$ \\
                   &  &  & Spillover  & $0.070$ & $0.377$ \\
                   &  &  & Total   & $0.148$ & $0.133$ \\
\hline
\end{tabular}}
\end{table}

\clearpage
\subsection{Contribution of the clever covariate as a function of strength of spillover}\label{AppendixFb}

\subsubsection{Additional details on data generating mechanism and models fit}\label{AppendixFb1}

We generate $J=49$ clusters of $n_j=48$ individuals and assign $24$ clusters to treatment by complete randomization. Individual covariates are independently distributed as $\mathbf{X}_{ji}\sim\mathcal{N}(\mathbf{0},\mathbf{I}_2)$, with no cluster-level covariates. Treatment receipt and the outcome follow
\begin{equation}\label{Sim::CCModel}
    \begin{aligned}
    D_{ji}(1)\mid\mathbf{X}_{ji}
    &\sim \mathrm{Bern}\left\{\Phi\left(\xi_0+\boldsymbol{\xi}_X^\top\mathbf{X}_{ji}\right)\right\},\\
    Y_{ji}\mid D_{ji},\bar{D}_{j(-i)},\mathbf{X}_{ji},b_j^Y
    &\sim \mathrm{Bern}\left\{
    \Phi\left(\gamma_DD_{ji}
    +\gamma_{\bar D}\bar D_{j(-i)}
    +\boldsymbol{\gamma}_X^\top\mathbf{X}_{ji}
    +b_j^Y\right)\right\},\qquad
    b_j^Y\sim\mathcal{N}(0,\tau_Y^2),
    \end{aligned}
\end{equation}
where $\boldsymbol{\xi}_X=(1,0.5)^\top$, $\boldsymbol{\gamma}_X=(0.4,0.4)^\top$, $\gamma_D=-0.5$, and $\tau_Y^2=0.25$. We set $\xi_0=\Phi^{-1}(0.82)(1+\|\boldsymbol{\xi}_X\|^2)^{1/2}=1.373$ to obtain marginal treated-arm uptake of $0.82$, and set observed receipt to $D_{ji}=Z_jD_{ji}(1)$. We vary $\gamma_{\bar D}$ over $\{0,-0.25,-0.5,-0.75,-1\}$, yielding outcome prevalence between $0.33$ and $0.43$. The DGM has no treatment receipt random effect, although both fitted models include one.

Under one-sided noncompliance, all peers remain untreated under control assignment. The true mean $\Theta_{(1,0)}$ therefore remains constant as spillover strength changes, allowing us to check whether its estimates also remain constant. The 10th and 90th percentiles of treated-cluster uptake are approximately $0.75$ and $0.89$. These datasets contain no treated individuals whose peers are all untreated, so estimating their outcomes requires extrapolation. The clever covariate varies little within each assignment arm and is highly correlated with observed peer receipt (correlation at least $0.995$). We examine whether including it reduces bias in the direct and spillover estimates despite this need for extrapolation.

For each spillover coefficient, we calculate the true potential-outcome means using $100{,}000$ independently generated clusters. We integrate over $b_j^Y$ using a closed-form expression and sum over all possible values of individual receipt and the number of treated peers. Monte Carlo standard errors for the effects are below $10^{-4}$. The true values are $\Theta_{(0,0)}=0.500$ and $\Theta_{(1,0)}=0.368$. Thus, the cluster-ITT direct effect remains $-0.132$, while the spillover effect ranges from $0$ to $-0.210$.

We generate $500$ datasets for each value of the peer receipt coefficient $\gamma_{\bar D}$ and fit both models using $50$ trees, $5{,}000$ burn-in iterations, and $1{,}000$ retained draws after thinning by $5$. Both outcome models include $(Z_j,D_{ji},\bar D_{j(-i)},\mathbf{X}_{ji})$. The model with the clever covariate additionally includes $\hat{\pi}^{\bar D}_{j(-i)}$. Both models use the priors in Appendix~\ref{AppendixF1}: $\mathrm{InvGamma}(3,0.6)$ for $\tau_Y^2$ and $\mathrm{InvGamma}(3,2)$ for $\tau_D^2$.

We fit the treatment-receipt BART-RE model and the preliminary parametric probit model only in treated clusters, setting control receipt probabilities to zero. Following Section~\ref{Sec51}, we include $\hat\pi^{\bar D}_{j(-i)}(Z_j)$ (clever covariate), calculated using each cluster's observed assignment, when fitting the outcome model. To predict outcomes for $\Theta_{(z,z_*)}$, we calculate the clever covariate using the peer assignment $z_*$. It equals zero when $z_*=0$ and the predicted mean peer uptake under treatment when $z_*=1$. For each simulated dataset, we fit the treatment-receipt model once. At each posterior draw, we use the same estimated receipt probabilities, receipt random-effect variance, and Bayesian-bootstrap weights in the g-computation analyses with and without the clever covariate. We also use the same random-number seed for the two outcome-model MCMC runs. The outcome models differ only in the inclusion of the clever covariate.

We sum over the possible numbers of treated peers using the Poisson-binomial distribution and integrate over the treatment receipt random effect using Gauss--Hermite quadrature. This replaces the Monte Carlo integration over receipt in Algorithm~\ref{Algo1}. We use $10$--$60$ quadrature nodes depending on the posterior draw of $\tau_D^2$, keeping the integration error of order $10^{-4}$ or smaller. We compare bias, RMSE, and coverage of $95\%$ credible intervals for the cluster-ITT effects and their potential-outcome means. We also calculate the ratio of MSE with the clever covariate to MSE without it, using the same simulated datasets for both estimates. To quantify uncertainty in this ratio, we use the delta method to approximate the standard error of its logarithm, accounting for the dependence between the two MSE estimates. We construct a $95\%$ confidence interval on the log scale and exponentiate its limits to ensure positive endpoints.

\subsubsection{Additional results}\label{AppendixFb2}

Table~\ref{TableSim3} provides the values from Figure~\ref{FigSimCC}, together with the true effects and credible interval coverage. Table~\ref{TabSuppS5} reports estimation results for $\Theta_{(1,1)}$, $\Theta_{(1,0)}$, and $\Theta_{(0,0)}$. Table~\ref{TabSuppS6} provides additional summaries of posterior uncertainty, including credible interval widths.

When spillover is present, the model without the clever covariate underestimates the contribution of peer receipt, with increasing bias as spillover strengthens. Including the clever covariate reduces this bias under moderate or strong spillover. With the strongest spillover, the estimated proportion of the total effect attributable to peer receipt increases from $32\%$ without the clever covariate to $46\%$ with it, although it remains below the true proportion of $62\%$. When spillover is absent, including the clever covariate increases the estimated proportion from $14\%$ to $23\%$, incorrectly attributing more of the total effect to peer receipt. Under weak spillover, it reduces direct-effect MSE but increases spillover-effect MSE.

Both models estimate $\Theta_{(1,1)}$ and $\Theta_{(0,0)}$ with small bias and approximately nominal credible interval coverage. Their estimates differ mainly for $\Theta_{(1,0)}$, the mean outcome when individual receipt is generated under treatment assignment and peer receipt under control assignment. Although its true value is constant across spillover strengths, its estimated value changes with the treated-arm mean. Across the five scenarios, the slope relating the average estimate of $\Theta_{(1,0)}$ to the true $\Theta_{(1,1)}$ is $0.56$ without the clever covariate and $0.38$ with it, compared with a true slope of zero. Including the clever covariate reduces this association but does not eliminate it. Since the covariate varies little within each assignment arm, these results suggest that it changes how BART attributes the difference between treatment and control outcomes to individual and peer receipt.

The $95\%$ credible intervals for the direct and spillover effects cover the true effects in every replicate. Their average widths are approximately $0.45$--$0.54$, several times the magnitude of the true direct effect. The ratio of average posterior standard deviation to the empirical standard deviation of posterior means ranges from $3.6$ to $5.3$ for these effects, compared with approximately one for the total effect. Including the clever covariate generally widens the direct and spillover intervals. Both models estimate the mean outcomes under treatment and control assignment relatively well, but must extrapolate to predict outcomes for treated individuals whose peers are all untreated. Uncertainty in these predictions contributes to the wide intervals for the decomposed effects. Thus, including the clever covariate reduces MSE for the point estimates without improving their precision.

For moderate and strong spillover ($\gamma_{\bar D}\leq-0.5$), the clever covariate reduces the direct-effect MSE by approximately $65\%$--$70\%$ and spillover-effect MSE by $57\%$--$66\%$. It also gives smaller absolute errors in $74\%$--$99.8\%$ of the corresponding paired datasets. Total-effect MSE ratios remain close to one across all scenarios. These comparisons demonstrate the benefits described in Section~\ref{Sec82}.

\begin{table}[htbp]
\caption{Cluster-ITT effects in the clever-covariate study ($J=49$, $n_j=48$): true value, bias, Monte Carlo standard error (MCSE) of the bias, $95\%$ credible interval coverage, and RMSE of BART-RE without and with the clever covariate, and the MSE ratio with/without the clever covariate with its $95\%$ confidence interval, across $500$ replicated datasets for each spillover coefficient $\gamma_{\bar D}$.}
\label{TableSim3}
\scalebox{0.74}{%
\begin{tabular}{@{}lccccccccccc@{}}
\hline
$\gamma_{\bar D}$ & Effects & Truth & \multicolumn{4}{c}{Without clever covariate} & \multicolumn{4}{c}{With clever covariate} & MSE ratio \\
\cline{4-7}\cline{8-11}
 & & & Bias & MCSE & Cov. & RMSE & Bias & MCSE & Cov. & RMSE & (95\% CI) \\
\hline
$\phantom{-}0$ & Direct & $-0.132$ & $0.0241$ & $0.0015$ & $100.0\%$ & $0.0407$ & $0.0350$ & $0.0012$ & $100.0\%$ & $0.0445$ & $1.196$ $(1.120, 1.277)$ \\
 & Spillover & $\phantom{-}0.000$ & $-0.0177$ & $0.0012$ & $100.0\%$ & $0.0319$ & $-0.0291$ & $0.0015$ & $100.0\%$ & $0.0441$ & $1.910$ $(1.784, 2.045)$ \\
 & Total & $-0.132$ & $0.0064$ & $0.0020$ & $\phantom{0}96.0\%$ & $0.0457$ & $0.0059$ & $0.0020$ & $\phantom{0}95.8\%$ & $0.0459$ & $1.009$ $(0.993, 1.025)$ \\
\cline{1-12}
$-0.25$ & Direct & $-0.132$ & $-0.0069$ & $0.0015$ & $100.0\%$ & $0.0347$ & $0.0141$ & $0.0013$ & $100.0\%$ & $0.0326$ & $0.882$ $(0.783, 0.994)$ \\
 & Spillover & $-0.061$ & $0.0147$ & $0.0011$ & $100.0\%$ & $0.0287$ & $-0.0074$ & $0.0014$ & $100.0\%$ & $0.0321$ & $1.256$ $(1.080, 1.460)$ \\
 & Total & $-0.193$ & $0.0078$ & $0.0020$ & $\phantom{0}96.0\%$ & $0.0450$ & $0.0066$ & $0.0020$ & $\phantom{0}95.8\%$ & $0.0453$ & $1.015$ $(1.000, 1.031)$ \\
\cline{1-12}
$-0.50$ & Direct & $-0.132$ & $-0.0346$ & $0.0015$ & $100.0\%$ & $0.0477$ & $-0.0048$ & $0.0012$ & $100.0\%$ & $0.0283$ & $0.352$ $(0.317, 0.389)$ \\
 & Spillover & $-0.117$ & $0.0490$ & $0.0011$ & $100.0\%$ & $0.0549$ & $0.0174$ & $0.0014$ & $100.0\%$ & $0.0361$ & $0.433$ $(0.400, 0.469)$ \\
 & Total & $-0.249$ & $0.0144$ & $0.0020$ & $\phantom{0}93.8\%$ & $0.0470$ & $0.0126$ & $0.0020$ & $\phantom{0}93.8\%$ & $0.0470$ & $1.000$ $(0.986, 1.014)$ \\
\cline{1-12}
$-0.75$ & Direct & $-0.132$ & $-0.0649$ & $0.0014$ & $100.0\%$ & $0.0725$ & $-0.0264$ & $0.0013$ & $100.0\%$ & $0.0396$ & $0.297$ $(0.281, 0.315)$ \\
 & Spillover & $-0.167$ & $0.0796$ & $0.0010$ & $100.0\%$ & $0.0829$ & $0.0390$ & $0.0014$ & $100.0\%$ & $0.0497$ & $0.359$ $(0.339, 0.380)$ \\
 & Total & $-0.299$ & $0.0147$ & $0.0019$ & $\phantom{0}95.0\%$ & $0.0439$ & $0.0127$ & $0.0019$ & $\phantom{0}95.8\%$ & $0.0437$ & $0.991$ $(0.974, 1.009)$ \\
\cline{1-12}
$-1.00$ & Direct & $-0.132$ & $-0.0941$ & $0.0015$ & $100.0\%$ & $0.0999$ & $-0.0467$ & $0.0014$ & $100.0\%$ & $0.0559$ & $0.313$ $(0.297, 0.330)$ \\
 & Spillover & $-0.210$ & $0.1064$ & $0.0010$ & $100.0\%$ & $0.1086$ & $0.0565$ & $0.0013$ & $100.0\%$ & $0.0636$ & $0.344$ $(0.327, 0.361)$ \\
 & Total & $-0.342$ & $0.0122$ & $0.0018$ & $\phantom{0}95.0\%$ & $0.0412$ & $0.0098$ & $0.0018$ & $\phantom{0}95.0\%$ & $0.0409$ & $0.983$ $(0.966, 1.001)$ \\
\hline
\end{tabular}}
\end{table}

\begin{table}[htbp]
\caption{Potential-outcome means in the clever-covariate study ($J=49$, $n_j=48$): true value, bias, RMSE, and $95\%$ credible interval coverage of BART-RE without and with the clever covariate across $500$ replicated datasets for each spillover coefficient $\gamma_{\bar D}$.}
\label{TabSuppS5}
\scalebox{0.85}{%
\begin{tabular}{@{}lccccccccc@{}}
\hline
$\gamma_{\bar D}$ & Potential outcome & Truth & \multicolumn{3}{c}{Without clever covariate} & \multicolumn{3}{c}{With clever covariate} \\
\cline{4-6}\cline{7-9}
 & & & Bias & Cov. & RMSE & Bias & Cov. & RMSE \\
\hline
$\phantom{-}0$ & $\Theta_{(1,1)}$ & $\phantom{-}0.368$ & $\phantom{-}0.004$ & $\phantom{0}96.6\%$ & $0.033$ & $\phantom{-}0.003$ & $\phantom{0}97.0\%$ & $0.033$ \\
 & $\Theta_{(0,0)}$ & $\phantom{-}0.500$ & $-0.003$ & $\phantom{0}95.0\%$ & $0.033$ & $-0.002$ & $\phantom{0}94.2\%$ & $0.033$ \\
 & $\Theta_{(1,0)}$ & $\phantom{-}0.368$ & $\phantom{-}0.021$ & $100.0\%$ & $0.037$ & $\phantom{-}0.033$ & $100.0\%$ & $0.044$ \\
\cline{1-9}
$-0.25$ & $\Theta_{(1,1)}$ & $\phantom{-}0.307$ & $\phantom{-}0.006$ & $\phantom{0}97.2\%$ & $0.031$ & $\phantom{-}0.006$ & $\phantom{0}97.6\%$ & $0.031$ \\
 & $\Theta_{(0,0)}$ & $\phantom{-}0.500$ & $-0.001$ & $\phantom{0}94.4\%$ & $0.033$ & $-0.001$ & $\phantom{0}94.6\%$ & $0.033$ \\
 & $\Theta_{(1,0)}$ & $\phantom{-}0.368$ & $-0.008$ & $100.0\%$ & $0.029$ & $\phantom{-}0.013$ & $100.0\%$ & $0.032$ \\
\cline{1-9}
$-0.50$ & $\Theta_{(1,1)}$ & $\phantom{-}0.251$ & $\phantom{-}0.011$ & $\phantom{0}96.4\%$ & $0.030$ & $\phantom{-}0.009$ & $\phantom{0}96.6\%$ & $0.029$ \\
 & $\Theta_{(0,0)}$ & $\phantom{-}0.500$ & $-0.004$ & $\phantom{0}93.6\%$ & $0.035$ & $-0.003$ & $\phantom{0}94.0\%$ & $0.035$ \\
 & $\Theta_{(1,0)}$ & $\phantom{-}0.368$ & $-0.038$ & $100.0\%$ & $0.047$ & $-0.008$ & $100.0\%$ & $0.030$ \\
\cline{1-9}
$-0.75$ & $\Theta_{(1,1)}$ & $\phantom{-}0.201$ & $\phantom{-}0.011$ & $\phantom{0}95.6\%$ & $0.029$ & $\phantom{-}0.010$ & $\phantom{0}96.0\%$ & $0.029$ \\
 & $\Theta_{(0,0)}$ & $\phantom{-}0.500$ & $-0.004$ & $\phantom{0}94.6\%$ & $0.035$ & $-0.003$ & $\phantom{0}95.2\%$ & $0.035$ \\
 & $\Theta_{(1,0)}$ & $\phantom{-}0.368$ & $-0.069$ & $100.0\%$ & $0.074$ & $-0.029$ & $100.0\%$ & $0.043$ \\
\cline{1-9}
$-1.00$ & $\Theta_{(1,1)}$ & $\phantom{-}0.158$ & $\phantom{-}0.011$ & $\phantom{0}96.0\%$ & $0.025$ & $\phantom{-}0.009$ & $\phantom{0}96.2\%$ & $0.025$ \\
 & $\Theta_{(0,0)}$ & $\phantom{-}0.500$ & $-0.002$ & $\phantom{0}94.0\%$ & $0.034$ & $-0.001$ & $\phantom{0}94.4\%$ & $0.034$ \\
 & $\Theta_{(1,0)}$ & $\phantom{-}0.368$ & $-0.096$ & $100.0\%$ & $0.100$ & $-0.047$ & $100.0\%$ & $0.057$ \\
\hline
\end{tabular}}
\end{table}

\begin{table}[htbp]
\caption{Average $95\%$ credible interval widths and ratio of the average posterior standard deviation to the empirical standard deviation of the posterior means for the cluster-ITT effects ($J=49$, $n_j=48$), without and with the clever covariate, across $500$ replicated datasets for each spillover coefficient $\gamma_{\bar D}$.}
\label{TabSuppS6}
\centering
\begin{tabular}{@{}lccccc@{}}
\hline
$\gamma_{\bar D}$ & Effects & \multicolumn{2}{c}{Avg.\ CI Width} & \multicolumn{2}{c}{Post.\ SD / Emp.\ SD} \\
\cline{3-4}\cline{5-6}
 & & Without & With & Without & With \\
\hline
$\phantom{-}0$ & Direct & $0.486$ & $0.535$ & $3.8$ & $5.0$ \\
 & Spillover & $0.471$ & $0.534$ & $4.5$ & $4.1$ \\
 & Total & $0.188$ & $0.188$ & $1.1$ & $1.1$ \\
\cline{1-6}
$-0.25$ & Direct & $0.476$ & $0.541$ & $3.6$ & $4.7$ \\
 & Spillover & $0.460$ & $0.541$ & $4.8$ & $4.4$ \\
 & Total & $0.183$ & $0.184$ & $1.1$ & $1.0$ \\
\cline{1-6}
$-0.50$ & Direct & $0.473$ & $0.534$ & $3.7$ & $4.9$ \\
 & Spillover & $0.456$ & $0.531$ & $4.7$ & $4.3$ \\
 & Total & $0.178$ & $0.178$ & $1.0$ & $1.0$ \\
\cline{1-6}
$-0.75$ & Direct & $0.472$ & $0.536$ & $3.8$ & $4.7$ \\
 & Spillover & $0.455$ & $0.533$ & $5.0$ & $4.5$ \\
 & Total & $0.171$ & $0.172$ & $1.1$ & $1.0$ \\
\cline{1-6}
$-1.00$ & Direct & $0.466$ & $0.544$ & $3.6$ & $4.6$ \\
 & Spillover & $0.448$ & $0.538$ & $5.3$ & $4.8$ \\
 & Total & $0.162$ & $0.163$ & $1.1$ & $1.0$ \\
\hline
\end{tabular}
\end{table}

\clearpage

\subsection{Strong spillover with smaller clusters}\label{AppendixFc}

We assess whether our findings with a strong spillover effect persist with smaller clusters. We use the same DGM in \eqref{Sim::CCModel}, but use $J=50$ clusters of $n_j=30$ individuals, assigning $25$ clusters to treatment, and setting $\gamma_{\bar D}=-1$. The true direct, spillover, and total effects are $-0.132$, $-0.210$, and $-0.342$, respectively. Smaller clusters provide more variation in treated-arm uptake, with 10th and 90th percentiles of approximately $0.73$ and $0.90$. However, estimating $\Theta_{(1,0)}$ still requires predicting outcomes for treated individuals whose peers are all untreated, a treatment pattern not observed in these datasets.

Table~\ref{TabSuppS7} and Figure~\ref{FigSuppR5} show that including the clever covariate roughly halves the absolute biases of the direct and spillover effect estimates and reduces their MSE by about two-thirds. The paired MSE ratios are close to those under strong spillover in Section~\ref{Sec82}, and the model with the clever covariate gives generally smaller absolute errors. The improvement therefore persists across these two cluster-size designs and independent datasets. Total-effect estimation changes little because including the clever covariate primarily improves estimation of $\Theta_{(1,0)}$, while estimates of the two assignment-specific means remain similar.

For the potential-outcome results, the MSE ratio for $\Theta_{(1,0)}$ is $0.32$, while ratios for $\Theta_{(1,1)}$ and $\Theta_{(0,0)}$ are close to one. Including the clever covariate increases the estimated spillover share from $33\%$ to $48\%$, bringing it closer to the true share of $61\%$. Thus, the gain reflects a more accurate decomposition of the intervention benefit, although the residual component bias remains. For both models, the $95\%$ credible intervals for the direct and spillover effects are wide and cover the true effects in more than $95\%$ of simulated datasets, while coverage for the total effect is approximately $95\%$.

These results support including the clever covariate when spillover is strong. However, both models still rely on extrapolation to predict outcomes for treated individuals with untreated peers. The comparisons under weak or no spillover in Section~\ref{Sec82} also show that the clever covariate does not improve estimation in every setting. Table~\ref{TabSuppS7} provides the full results for the causal effects and potential-outcome means.

\begin{table}[htbp]
\caption{Supplementary clever-covariate study ($J=50$, $n_j=30$, $\gamma_{\bar D}=-1$): bias, $95\%$ credible interval coverage, RMSE, and average interval width of BART-RE without and with the clever covariate for the cluster-ITT effects and potential-outcome means across $500$ replicated datasets. MSE ratio: with/without the clever covariate over the same datasets, with $95\%$ confidence interval.}
\label{TabSuppS7}
\scalebox{0.78}{%
\begin{tabular}{@{}lcccccccccc@{}}
\hline
 & Ground Truth & \multicolumn{4}{c}{Without clever covariate} & \multicolumn{4}{c}{With clever covariate} & MSE ratio \\
\cline{3-6}\cline{7-10}
 & & Bias & Cov. & RMSE & Width & Bias & Cov. & RMSE & Width & (95\% CI) \\
\hline
Direct & $-0.132$ & $-0.091$ & $100.0\%$ & $0.099$ & $0.481$ & $-0.043$ & $100.0\%$ & $0.055$ & $0.564$ & $0.31$ $(0.29, 0.33)$ \\
Spillover & $-0.210$ & $\phantom{-}0.103$ & $100.0\%$ & $0.106$ & $0.461$ & $\phantom{-}0.051$ & $100.0\%$ & $0.061$ & $0.557$ & $0.33$ $(0.31, 0.35)$ \\
Total & $-0.342$ & $\phantom{-}0.011$ & $\phantom{0}95.2\%$ & $0.042$ & $0.168$ & $\phantom{-}0.009$ & $\phantom{0}95.0\%$ & $0.042$ & $0.168$ & $0.97$ $(0.96, 0.99)$ \\
\hline
$\Theta_{(1,1)}$ & $\phantom{-}0.158$ & $\phantom{-}0.009$ & $\phantom{0}94.6\%$ & $0.025$ & $0.100$ & $\phantom{-}0.008$ & $\phantom{0}95.2\%$ & $0.025$ & $0.100$ & $0.95$ $(0.94, 0.97)$ \\
$\Theta_{(0,0)}$ & $\phantom{-}0.500$ & $-0.002$ & $\phantom{0}95.0\%$ & $0.034$ & $0.137$ & $-0.001$ & $\phantom{0}95.0\%$ & $0.034$ & $0.136$ & $0.99$ $(0.98, 1.01)$ \\
$\Theta_{(1,0)}$ & $\phantom{-}0.368$ & $-0.093$ & $100.0\%$ & $0.098$ & $0.473$ & $-0.044$ & $100.0\%$ & $0.056$ & $0.561$ & $0.32$ $(0.30, 0.34)$ \\
\hline
\end{tabular}}
\end{table}

\begin{figure}[htbp]
\centering
\includegraphics[width=\textwidth]{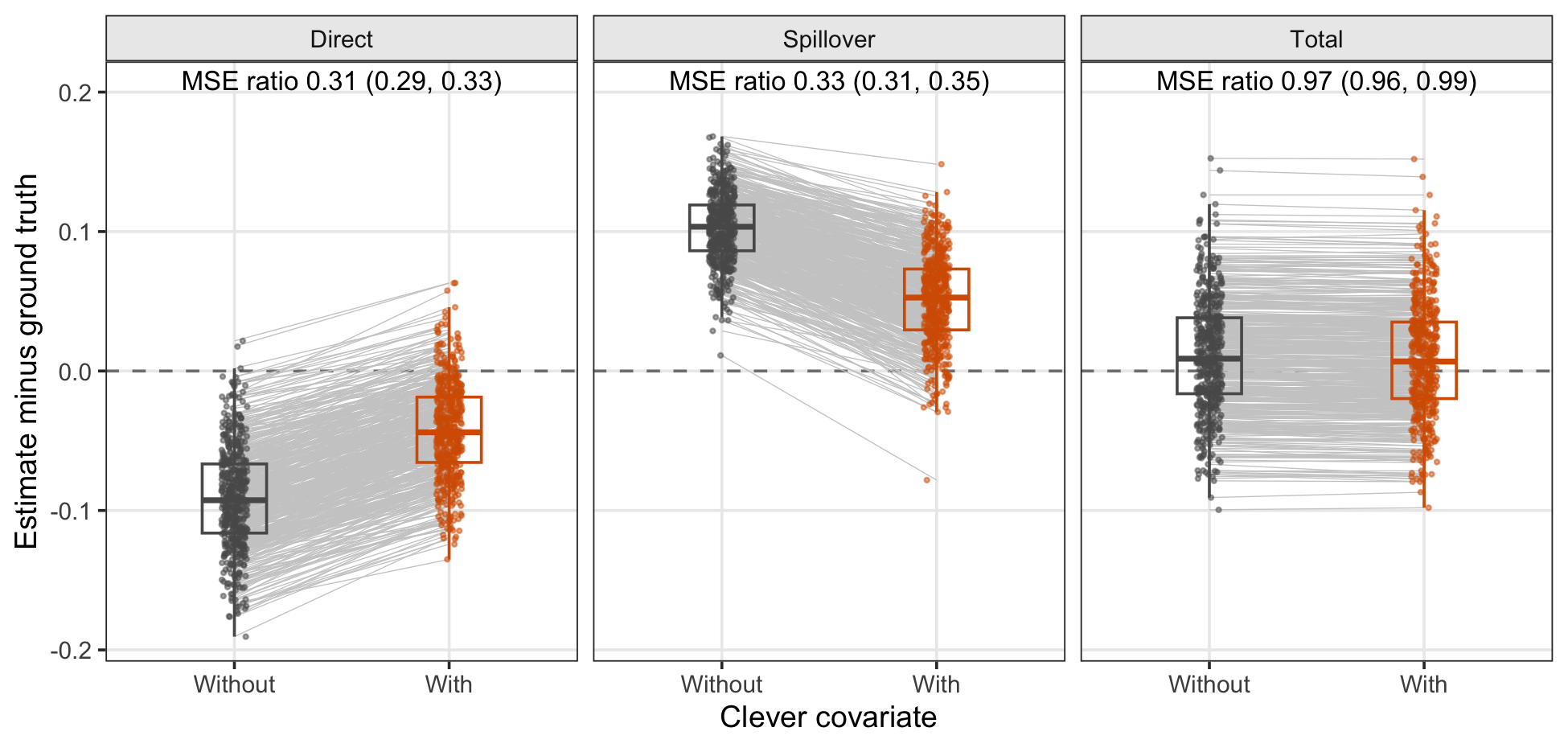}
\caption{Supplementary clever-covariate study ($J=50$, $n_j=30$, $\gamma_{\bar D}=-1$): estimation error of the cluster-ITT direct, spillover, and total effects in each of $500$ replicated datasets without and with the clever covariate. Lines join the two estimates from the same dataset, and boxplots summarize the errors. Each panel reports the MSE ratio with/without the clever covariate and its $95\%$ confidence interval.}
\label{FigSuppR5}
\end{figure}

\newpage

\section{Directed Acyclic Graph (DAG)}\label{AppendixG}
\begin{figure}[htbp]
\centering
\begin{tikzpicture}[
 every node/.style={circle,draw=black,thick,minimum size=1.2cm,font=\small,align=center},
 arrow/.style={->,thick,>=Stealth}]
\node (X) at (0,0) {$\mathbf{X}_j$};
\node (Z) at (0,2.8) {$Z_j$};
\node (D1) at (3,1.4) {$D_{ji}$};
\node (D2) at (3,-1.4) {$D_{ji'}$};
\node (Y1) at (7,1.4) {$Y_{ji}$};
\node (Y2) at (7,-1.4) {$Y_{ji'}$};
\draw[arrow] (Z) -- (D1);
\draw[arrow] (Z) to[bend right=15] (D2);
\draw[arrow] (X) -- (D1);
\draw[arrow] (X) -- (D2);
\draw[arrow] (X) to[bend left=45] (Y1);
\draw[arrow] (X) to[bend right=45] (Y2);
\draw[arrow] (D1) -- (Y1);
\draw[arrow] (D2) -- (Y2);
\draw[arrow,blue] (D1) -- (Y2);
\draw[arrow,blue] (D2) -- (Y1);
\end{tikzpicture}
\caption{DAG of treatment receipt and outcome pathways for two individuals in cluster $j$. The shared node $Z_j$ denotes randomized cluster assignment, and $\mathbf{X}_j$ includes individual and cluster baseline covariates. Blue arrows indicate spillover pathways. An assignment-to-outcome arrow is absent under network exclusion. Random effects and other cluster members are omitted for clarity.}
\label{Fig:DAG}
\end{figure}
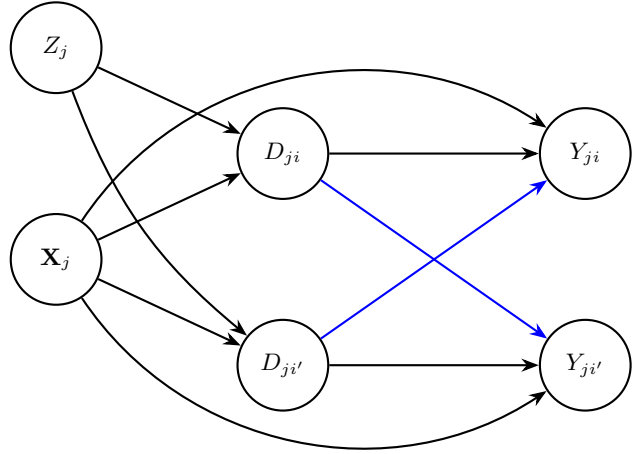

\end{document}